%% file: main.tex
\documentclass[11pt]{article}
\usepackage{lmodern} 
\usepackage[T1]{fontenc}

\usepackage{color}
\usepackage[colorlinks=true,linkcolor=blue,citecolor=blue]{hyperref}
\usepackage{amsmath, amssymb, amsthm}
\usepackage{mathtools}
\usepackage[margin=1in]{geometry}
\usepackage{graphics}
\usepackage{pifont}
\usepackage{tikz}
\usepackage{bbm, bm}

\DeclareMathOperator*{\argmin}{arg\,min}
\usetikzlibrary{arrows.meta}
\usepackage{environ}
\usepackage{framed}
\usepackage{url}
\usepackage[linesnumbered,ruled,vlined]{algorithm2e}
\usepackage[noend]{algpseudocode}
\usepackage[labelfont=bf]{caption}
\usepackage{framed}
\usepackage[framemethod=tikz]{mdframed}
\usepackage{appendix}
\usepackage{graphicx}
\usepackage[textsize=tiny]{todonotes}
\usepackage{tcolorbox}
\allowdisplaybreaks[3]
\usepackage{nicefrac}
\usepackage{thm-restate}
\usepackage[noabbrev,capitalize,nameinlink]{cleveref}
\crefname{equation}{}{}
\usepackage{pdfsync}
\usepackage{tabularx, ragged2e, booktabs}

\usepackage{lipsum}
\usepackage{enumerate}
\usepackage{float}

\usepackage{subcaption}

\DontPrintSemicolon
\SetKw{KwAnd}{and}
\SetProcNameSty{textsc}
\SetFuncSty{textsc}

\newcommand\remove[1]{}
\renewcommand\paragraph[1]{\smallskip \noindent \textbf{#1}}

\newtheorem{lemma}{Lemma}[section]
\newtheorem*{lemma*}{Lemma}
\newtheorem{theorem}[lemma]{Theorem}
\newtheorem{corollary}[lemma]{Corollary}
\newtheorem*{corollary*}{Corollary}
\newtheorem{claim}[lemma]{Claim}

\newtheorem*{theorem*}{Theorem}
\newtheorem*{inducthyp*}{Inductive Hypothesis}

\newtheorem*{rem*}{Remark}

\newtheorem{invar}{Invariant}

\newtheorem{fact}{Fact}

\theoremstyle{definition}
\newtheorem*{definition*}{Definition}
\newtheorem{definition}[lemma]{Definition}

\newcommand\R{\mathbb{R}}

\newcommand\Z{\mathbb{Z}}

\newcommand{\eps}{\varepsilon}
\renewcommand{\O}{\widetilde{O}}

\renewcommand{\l}{\langle}
\renewcommand{\r}{\rangle}

\newcommand{\assign}{\leftarrow}
\newcommand{\otilde}{\O}
\renewcommand{\forall}{\mathrm{\text{ for all }}}

\newcommand{\vol}{\mathrm{vol}}
\newcommand{\econg}{\mathbf{econg}}
\newcommand{\vcong}{\mathbf{vcong}}
\newcommand{\length}{\mathrm{length}}

\newcommand{\bc}{\bm{c}}
\newcommand{\bd}{\boldsymbol{d}}
\renewcommand{\bf}{\bm{f}}
\newcommand{\bg}{\boldsymbol{g}}
\newcommand{\bp}{\boldsymbol{p}}

\newcommand{\bw}{\boldsymbol{w}}

\newcommand{\bu}{\boldsymbol{u}}
\newcommand{\bv}{\boldsymbol{v}}

\newcommand{\bx}{\boldsymbol{x}}
\newcommand{\by}{\boldsymbol{y}}
\newcommand{\bell}{\boldsymbol{\ell}}

\newcommand{\bDelta}{\boldsymbol{\Delta}}
\newcommand{\blambda}{\boldsymbol{\lambda}}
\newcommand{\diag}{\mathrm{diag}}
\renewcommand{\root}{\mathsf{root}}
\newcommand{\bDeg}{\mathbf{deg}}

\newcommand{\Enc}{\textsc{Enc}}

\newcommand{\DynamicBranchingChain}{\textsc{DynamicTreeChain}}

\newcommand{\Rebuild}{\textsc{Rebuild}}
\newcommand{\Shift}{\textsc{Shift}}
\newcommand{\Initialize}{\textsc{Initialize}}
\newcommand{\Update}{\textsc{Update}}
\newcommand{\FindCycle}{\textsc{FindCycle}}

\renewcommand{\hat}{\widehat}
\renewcommand{\tilde}{\widetilde}

\DeclareFontFamily{U}{mathb}{\hyphenchar\font45}
\DeclareFontShape{U}{mathb}{m}{n}{<5> <6> <7> <8> <9> <10> gen * mathb
<10.95> mathb10 <12> <14.4> <17.28> <20.74> <24.88> mathb12}{}
\DeclareSymbolFont{mathb}{U}{mathb}{m}{n}
\DeclareMathSymbol{\rcirclearrow}{\mathbin}{mathb}{'367}

\newcommand{\wt}{\widetilde}
\newcommand{\wh}{\widehat}

\renewcommand{\bar}{\overline}

\foreach \x in {A,...,Z}{
	\expandafter\xdef\csname m\x\endcsname{\noexpand\mathbf{\x}}
}

\foreach \x in {A,...,Z}{
	\expandafter\xdef\csname c\x\endcsname{\noexpand\mathcal{\x}}
}

\newif\ifrandom
\randomtrue

\newcommand{\defeq}{\stackrel{\mathrm{\scriptscriptstyle def}}{=}}

\newcommand{\poly}{{\mathrm{poly}}}

\newcommand{\str}{{\mathsf{str}}}

\newcommand{\wstr}{{\wt{\str}}}
\newcommand{\hstr}{{\wh{\str}}}
\newcommand{\len}{{\mathrm{len}}}
\newcommand{\wlen}{{\wt{\len}}}

\newcommand{\rev}{\mathsf{rev}}

\newcommand{\lvl}{\mathsf{level}}

\newcommand{\bran}{\mathsf{shift}}

\newcommand{\bb}{\boldsymbol{b}}
\newcommand{\ba}{\boldsymbol{a}}
\newcommand{\bPi}{\boldsymbol{\Pi}}
\newcommand{\last}{\mathsf{last}}

\renewcommand{\l}{\langle}
\renewcommand{\r}{\rangle}

\newcommand{\norm}[1]{\left\lVert#1\right\rVert}
\newcommand{\Abs}[1]{\left|#1\right|}
\newcommand{\Set}[2]{\left\{#1 ~\middle\vert~ #2\right\}}
\newcommand{\setof}[1]{\left\{#1\right\}}

\newcommand{\todolater}[1]{}

\newcommand{\Wrange}{\Psi}
\newcommand{\gstr}{{\bar{\str}}}
\newcommand{\grepT}{{\bar{\repT}}}
\newcommand{\passes}{{\mathsf{passes}}}

\newcommand{\repT}{{\mathsf{repT}}}

\newcommand{\gameAct}[1]{\texttt{#1-step}}
\newcommand{\gameActParam}[2]{\texttt{#1-step}(#2)}

\crefname{gameStage}{Game Stage}{Game Stages}

\renewcommand{\SS}{\mathcal{S}}
\newcommand{\abs}[1]{\left|#1\right|}

\DeclareUnicodeCharacter{2113}{$\ell$}

\usepackage[backend=biber, isbn=false, style=alphabetic, backref=true,
doi=false, url=false, maxcitenames=10, mincitenames=5,
maxalphanames=10, maxbibnames=10, minbibnames=5, minalphanames=3,
defernumbers=true, sortlocale=en_US, sorting=ynt, sortcites]{biblatex}
\AtBeginBibliography{\small}
\begin{document}
\pagenumbering{gobble}

\title{A Deterministic Almost-Linear Time Algorithm\\ for Minimum-Cost Flow}

\author{
Jan van den Brand\\ Georgia Tech\\ vdbrand@gatech.edu
\and
Li Chen\thanks{Li Chen was supported by NSF Grant CCF-2106444.}\\ Georgia Tech\\ lichen@gatech.edu
\and
Rasmus Kyng\thanks{The research leading to these results has received funding from the grant ``Algorithms and complexity for high-accuracy flows and convex optimization'' (no. 200021 204787) of the Swiss National Science Foundation.}\\ ETH Zurich \\ kyng@inf.ethz.ch 
\and
Yang P. Liu\thanks{Yang P. Liu was supported by NSF CAREER Award CCF-1844855, NSF Grant CCF-1955039, and a Google Research Fellowship.} \\ Stanford University \\ yangpliu@stanford.edu
\and
Richard Peng\thanks{Richard Peng was partially supported by NSF CAREER Award CCF-1846218,
and the Natural Sciences and Engineering Research Council of Canada (NSERC) Discovery Grant RGPIN-2022-03207} \\
Carnegie Mellon University
\footnote{Part of this work was done while at the University of Waterloo.}\\
yangp@cs.cmu.edu
\and
Maximilian Probst Gutenberg\footnotemark[2]\\ ETH Zurich\\ maxprobst@ethz.ch
\and
Sushant Sachdeva\thanks{Sushant Sachdeva was supported by an NSERC Discovery Grant RGPIN-2018-06398, an Ontario Early Researcher Award (ERA) ER21-16-283, and a Sloan Research Fellowship.
} \\ University of Toronto \\ sachdeva@cs.toronto.edu
\and 
Aaron Sidford\thanks{Aaron Sidford was supported in part by a Microsoft Research Faculty Fellowship, NSF CAREER Award CCF-1844855, NSF Grant CCF-1955039, a PayPal research award, and a Sloan Research Fellowship.} \\
Stanford University\\
sidford@stanford.edu
}
\maketitle

\begin{abstract}
We give a deterministic $m^{1+o(1)}$ time algorithm that computes exact maximum flows and minimum-cost flows on directed graphs with $m$ edges and polynomially bounded integral demands, costs, and capacities. As a consequence, we obtain the first running time improvement for deterministic algorithms that compute maximum-flow in graphs with polynomial bounded capacities since the work of Goldberg-Rao [J.ACM '98].

Our algorithm builds on the framework of Chen-Kyng-Liu-Peng-Gutenberg-Sachdeva [FOCS '22] that computes an optimal flow by computing a sequence of $m^{1+o(1)}$-approximate undirected minimum-ratio cycles. 
We develop a deterministic dynamic graph data-structure to compute such a sequence of minimum-ratio cycles in an amortized $m^{o(1)}$ time per edge update.
Our key technical contributions are deterministic analogues of
the \emph{vertex sparsification} and \emph{edge sparsification} components of the data-structure from
Chen et al. For the vertex sparsification component, we give a method to avoid the randomness in Chen et al. which involved sampling random trees to recurse on.
For the edge sparsification component, we design a deterministic algorithm that maintains an embedding of a dynamic graph into a sparse spanner.
We also show how our dynamic spanner can be applied to give a deterministic data structure that
maintains a fully dynamic low-stretch spanning tree on graphs with
polynomially bounded edge lengths, with subpolynomial average stretch and subpolynomial amortized time per edge update.
\end{abstract}

\newpage
\setcounter{tocdepth}{2}
\tableofcontents
\normalsize
\pagebreak
\pagenumbering{arabic}

\input{intro}
\input{prelims}
\input{approach}
\input{ipm_algorithm}
\input{jtree}
\input{branch_shift_rebuild_game}
\input{spanner}
\input{lsst}

\section*{Acknowledgments}
We thank the anonymous reviewers for their helpful comments.

\begin{refcontext}[sorting=nyt]
\printbibliography[heading=bibintoc]
\end{refcontext}
\appendix

\input{jtree_appendix}
\input{spanner_appendix}

\end{document}

%% file: intro.tex
\section{Introduction}

Given a directed, capacitated graph $G = (V, E, \bu)$ with $n = |V|$ nodes, $m = |E|$ edges, and integer capacities $\bu \in \Z_{\geq 0}^E$, the \emph{maxflow problem} asks to send as much flow as possible on $G$ from a given source vertex $s \in V$ to a sink vertex $t \in V \setminus \{s\}$ without exceeding the capacity constraints. This problem is foundational in combinatorial optimization and algorithm design. It has been the subject of extensive study for decades, starting from the works \cite{D51,HK73,K73,ET75} and is a key subroutine for solving a variety of algorithmic challenges such as edge-connectivity and approximate sparsest cut (e.g., \cite{GomoryH61,KhandekarRV06}).

In the standard setting where the capacities are polynomially bounded, a line of work on combinatorial algorithms culminated in a seminal result of Goldberg and Rao in 1998 \cite{GR98} which showed that the problem can be solved in $\otilde(m \cdot \min\{m^{1/2}, n^{2/3}\})$ time. The algorithm which achieved this result was deterministic and combinatorial; the algorithm consists of a careful repeated computation of blocking-flows implemented in nearly linear time using dynamic trees. Interestingly, despite advances in randomized algorithms for maxflow 
(\cite{LS15,BLNPSSW20,GLP21:arxiv,BGJLLPS21:arxiv}) 
 and deterministic algorithms in special cases (e.g., unit capacity graphs \cite{M13} and planar graphs \cite{borradaile2009n, borradaile2017multiple}), the runtime in \cite{GR98} has remained the state-of-the-art among deterministic algorithms in the general case of polynomially bounded capacities.

This gap between state-of-the-art runtimes for deterministic and randomized algorithms for maxflow is particularly striking in light of recent advances: \cite{chen2022maximum} provided an almost linear, $m^{1 + o(1)}$, time randomized maxflow algorithm and
\cite{BLLSSSW21} provided an $\otilde(m + n^{1.5})$ time randomized algorithm which runs in nearly linear time for dense graphs. Unfortunately, as we discuss in \Cref{sec:intro:recent_randomized_algorihtms}, there are key barriers towards efficiently derandomizing both \cite{chen2022maximum} and \cite{BLLSSSW21} as well as prior improvements~\cite{GLP21:arxiv,BGJLLPS21:arxiv,AMV21:arxiv}.

These results raise key questions about the power of randomization in designing flow algorithms. While there is complexity theoretic evidence that randomization does not affect the polynomial time solvability of decision problems \cite{IW97} it is less clear what fine-grained effect randomization has on the best achievable runtimes or whether or not a problem can be solved in almost linear time \cite{CT21}. The problem of obtaining faster deterministic algorithms for maxflow is of particular interest given extensive research over the past decade on obtain faster deterministic algorithms for expander decompositions and flow problems \cite{CGLNPS21, KMP22}, and applications to connectivity problems~\cite{KawarabayashiT19,LiP20,Li21}.

In this paper we provide a deterministic algorithm that solves minimum-cost flow and maxflow in $m^{1 + o(1)}$ time. We obtain this result by providing an efficient deterministic implementation of the recent flow framework of \cite{chen2022maximum} which reduced the minimum cost flow problem to approximately solving a sequence of structured minimum ratio cycle problems. We also obtain the same running time for deterministically finding flows on graphs that minimize convex edge costs. Further, the techniques we develop have potential broader utility; for example, we show that our techniques can be used to design a deterministic algorithm that dynamically maintains low-stretch trees under insertions and deletions with polynomially bounded lengths (see \Cref{sec:intro:results})

\paragraph{Paper Organization.}
In the remainder of this introduction we elaborate on randomized maxflow algorithms and the barriers to their derandomization (\Cref{sec:intro:recent_randomized_algorihtms}), present our results (\Cref{sec:intro:results}), give a coarse overview of our approach (\Cref{sec:intro:approach}), and cover additional related work (\Cref{sec:intro:related}). We then cover preliminaries in \Cref{sec:prelim} and give a more technical overview in \Cref{sec:tech_overview}. We present the flow framework in \Cref{sec:framework}, build the main dynamic recursive data structure in \Cref{sec:jtree}, give a preliminary analysis of its quality in \Cref{sec:cycleshift}, strengthen the data structure by periodically rebuilding data structure levels in \Cref{sec:rebuilding},
and give the deterministic spanner in \Cref{sec:spanner}. Finally, we briefly describe our deterministic, dynamic low-stretch tree data structures in \Cref{sec:lsst}.

\subsection{Randomized Maxflow Algorithms}
\label{sec:intro:recent_randomized_algorihtms}

Although the runtime of deterministic algorithms solving maxflow on graphs with polynomially bounded capacities has not been improved since \cite{GR98}, there have been significant advances towards designing randomized maxflow algorithms. Here we provide a brief survey of these advances and discuss the difficulty in obtaining deterministic counterparts of comparable efficiency.

\paragraph{Electric Flow Based Interior Point Methods.} A number of randomized algorithms over the past decade have improved upon the complexity of maxflow by leveraging and building upon interior point methods (IPMs). IPMs are a broad class of continuous optimization methods that typically reduce continuous optimization problems, e.g., linear programming, to solving a sequence of linear systems. In the special case of maxflow the linear systems typically correspond to electric flow or Laplacian system solving and can be solved in nearly linear time \cite{ST04}.

Combining this approach with improved IPMs, \cite{LS19:arxiv} obtained an $\otilde(m \sqrt{n})$ time maxflow algorithm. Further robustifying this optimization method and using a range of dynamic data structures for maintaining decompositions of a graph into expanders, sparsifiers, and more, \cite{BLLSSSW21} obtained an improved $\otilde(m + n^{1.5})$ time maxflow algorithm.
Incorporating additional dynamic data structures for maintaining types of vertex sparsifiers (and more) then led to runtimes of $\tilde{O}(m^{3/2-1/328})$ \cite{GLP21:arxiv} and $\tilde{O}(m^{3/2-1/58})$ \cite{BGJLLPS21:arxiv}. 

Unfortunately, despite improved understanding of IPMs (in particular deterministic robust linear programming methods \cite{B20}) and deterministic Laplacian system solvers \cite{CGLNPS21} it is unclear how to obtain deterministic analogs of these maxflow results. 
Each result either uses $1/\poly(\log n)$-accurate estimates of effective resistances \cite{LS19:arxiv,BLLSSSW21} 
or edge or vertex sparsifiers of similar accuracy \cite{GLP21:arxiv,BGJLLPS21:arxiv}. 
Obtaining deterministic algorithms for either is an exciting open problem in algorithmic graph theory (and is left unsolved by this paper).

\paragraph{Minimum Ratio Cycle Based Interior Point Methods.} 
In a recent breakthrough result \cite{chen2022maximum} leveraged a different type of IPM. 
This method obtained an almost-linear time algorithm for maxflow and instead used an $\ell_1$-counterpart to the more standard $\ell_2$-based IPMs that reduce maxflow to electric flow.
Using this IPM, \cite{chen2022maximum} essentially reduced solving maxflow to solving a dynamic sequence of minimum ratio cycle problems (e.g., \cref{def:ipm:datastructure}). 

On the one hand, \cite{chen2022maximum} seems to create hope in overcoming the obstacles of faster deterministic maxflow algorithms. 
Using \cite{chen2022maximum} it is indeed known how to deterministically solve each individual minimum ratio cycle problem to sufficient accuracy in almost linear time.
On the other hand, unfortunately \cite{chen2022maximum} required a dynamic data structure for solving these problems in amortized $m^{o(1)}$-per instance
and to obtain their runtime, \cite{chen2022maximum} made key use of randomization. In \Cref{sec:intro:approach} we elaborate on the obstacles in avoiding this use of randomization and our main results, which are new algorithmic tools which remove this need. 

\subsection{Our Results}
\label{sec:intro:results}

We give a deterministic algorithm for computing min-cost flows on graphs.
\begin{theorem}[Min-cost flow]
\label{thm:mincost}
There is a deterministic algorithm that given a $m$-edge graph with integral vertex demands and edge capacities bounded by $U$ in absolute value, and integral edge costs bounded by $C$ in absolute value, computes an 
(exact) minimum-cost flow in time $m^{1+o(1)} \log U \log C$.
\end{theorem}

Our algorithm extends to finding flows that minimize convex edge costs to high-accuracy, for example, for matrix scaling, entropy-regularized optimal transport, $p$-norm flows, and $p$-norm isotonic regression. See \cite[Section 10]{chen2022maximum} for a (deterministic) reduction of these problems to a sequence of minimum ratio cycle problems satisfying the relevant stability guarantees.

Additionally, the components of our data structure can be used to deterministically maintain a low-stretch tree under dynamic updates (see \cref{sec:prelim} and \cref{thm:an} for a formal definition of edge stretch, and a concrete low-stretch tree statement).
Previously, deterministic algorithms for maintaining a low-stretch tree with subpolynomial update time were only known for unweighted graphs and those undergoing only edge deletions, achieved by combining the previous result \cite{chechik2020dynamic} with derandomization techniques in \cite{BGS21, C21}. Even among randomized algorithms, the only way the authors know how to achieve an algorithm that maintains low-stretch trees on graphs with polynomially bounded edge lengths in subpolynomial update time is by adapting the components of \cite{chen2022maximum} to the setting of low-stretch trees.

\begin{theorem}[Dynamic low stretch tree]
\label{thm:lsst}
There is a deterministic data structure that given a dynamic $n$-node graph undergoing insertions and deletions of edges with integral lengths bounded by $\exp((\log n)^{O(1)})$, maintains a low-stretch tree with average stretch $n^{o(1)}$ in worst-case $n^{o(1)}$ time per update. The data structure maintains the tree in memory with $n^{o(1)}$ amortized recourse per update; the data structure can be be modified to output the changes explicitly with amortized, rather than worst-case, $n^{o(1)}$ time per update.
\end{theorem}

\subsection{Our Approach}
\label{sec:intro:approach}

In this paper we obtain an almost linear time algorithm for maxflow by essentially showing how to eliminate the use of randomness in each of the places it was used \cite{chen2022maximum}. Here we elaborate on these uses of randomness and the techniques we introduce; we provide a more detailed overview of our approach in \Cref{sec:tech_overview}.

\paragraph{Randomization in \cite{chen2022maximum}.}
At a high level, \cite{chen2022maximum} treats minimum ratio cycle as an instance of the more general minimum cost transshipment problem on undirected graphs. To solve this, \cite{chen2022maximum} applies a time-tested technique of recursively building partial trees (to reduce the number of vertices) and sparsifying (to reduce the number of edges). This approach was pioneered by \cite{ST04}, and has since been used in multiple algorithms \cite{KMP11,KLOS14,S13,KPSW19,CPW21} and dynamic data structures \cite{CGHPS20}.

More precisely, for a parameter $k$ the partial trees are a collection of $\otilde(k)$ forests with $O(m/k)$ components where the stretch of every edge in a component is $\otilde(1)$ on average; here, the stretch of an edge refers to the ratio of the length of routing the edge in the forest to the length of the edge itself. The algorithm uses the partial trees to recursively processes the graphs resulting from contracting each forest. While the forests can be computed and even dynamically maintained deterministically, recursively processing all the partial trees is prohibitively expensive because the total number of components is still $\O(k \cdot m/k) = \O(m)$, i.e., there is no total size reduction. Thus, \cite{chen2022maximum} (motivated in part by \cite{M10,GKKLP18}) showed that it sufficed to subsample only $\O(1)$ trees to recurse on. This is the first and, perhaps, most critical use of randomness in the \cite{chen2022maximum} algorithm. In particular, it initially seems difficult to design a data structure that maintains all the trees without having a prohibitive runtime.

The dynamic sparsifier constructed in \cite{chen2022maximum} was a spanner with \emph{explicit embedding}, i.e., the algorithm maintained a subgraph $H \subseteq G$, and for each edge $e \in G$, a path in $H$ with few edges that connected its endpoints. This graph $H$ was maintained with low recourse under edge insertions, deletions, and \emph{vertex splits}, where a vertex becomes two vertices, and edges are split between them. The spanner was constructed by maintaining an expander decomposition and uniform sampling edges in each expander. This is the second use of randomness in \cite{chen2022maximum}, though it is conceptually easier to circumvent due to recent progress on deterministic expander decomposition and routings \cite{CGLNPS21,CS21}. 

\paragraph{Removing randomness from sampling forests.} To understand how we remove randomness from sampling the forests, it is critical to discuss how \cite{chen2022maximum} handled the issue of \emph{adaptive adversaries} in the dynamic updates to the data structure (i.e., that the input to the dynamic minimum ratio cycle data structures could depend on the data structure's output).
In particular, the future updates to the data structure may depend on the trees that were randomly sampled. To handle this, \cite{chen2022maximum} observed that the IPM provided additional stability on the dynamic minimum ratio cycle problem, in the sense that there was a (sufficiently good) solution $\bDelta^*$ to the minimum ratio cycle problem 
$\argmin_{\mB^\top \bDelta = 0} \bg^\top \bDelta / \|\mL\bDelta\|_1$ which changed slowly.

In a similar way, our deterministic min-ratio cycle data structure does \emph{not} work for general dynamic minimum ratio cycle, and instead heavily leverages the stability of a solution $\bDelta^*$. As in \cite{chen2022maximum}, our algorithm computes the $\O(k)$ partial trees. We know that out of these forests, there exists at least one of them (in fact, at least half of them) that we can successfully recurse on. \cite{chen2022maximum} chooses $\O(1)$ random forests to recurse on, leveraging that at least one of these forests is good with high probability.  As discussed, we cannot afford to recurse on all $\O(k)$ forests as this requires dynamically maintaining $\Omega(m)$ trees at every step. Consequently, to obtain a deterministic algorithm we instead show that it suffices to recurse one forest at a time. We recurse on the first forest until we conclude that it did not output a valid solution, then we switch to the next forest, and repeat (wrapping around if necessary). This way, we only maintain one recursive chain and the corresponding spanning tree at each point in time. 
We argue that the runtime is still acceptable, and more interestingly, that we do not need to switch between branches very frequently. We formalize this, we analyze what we call the \emph{shift-and-rebuild game} in \Cref{sec:rebuilding}, and extend the adaptive adversary analysis of \cite{chen2022maximum} to our new algorithm.

\paragraph{Deterministically constructing spanners with embeddings.} At a high level, \cite{chen2022maximum} gives a deterministic procedure of reducing dynamic spanners to static spanners with embeddings. To construct the static spanner, \cite{chen2022maximum} decomposed $G$ into expanders, sparsified each expander by random sampling, and then embedded $G$ into the sparsifier using a decremental shortest path data structure \cite{CS21}. The expander decomposition can be computed deterministically using \cite{CGLNPS21}. Thus, the remaining randomized component was the construction of the spanner by subsampling.
Instead, we construct the spanner by constructing a deterministic expander $W$ on each piece of the expander decomposition, embedding $G$ into $W$, and then embedding $W$ back into $G$ (both using the deterministic decremental shortest path data structure \cite{CS21}).
The set of edges in $G$ used to embed $W$ forms the spanner.
For our overall maxflow algorithm, we require additional properties of the dynamic spanner algorithm beyond the embedding; see \Cref{thm:spanner}.

\subsection{Additional Related Work}
\label{sec:intro:related}

\paragraph{Derandomization for flow-related problems.} Deterministic algorithms for sparsest cut, balanced cut, and expander decomposition \cite{CGLNPS21} can be directly applied to give a variety of deterministic algorithms for flow problems, including solving Laplacian linear systems (electric flows), $p$-norm flows on unit graphs \cite{KPSW19}, and more recently, directed Laplacian linear systems \cite{KMP22}. While we utilize deterministic expander decompositions and routings from \cite{CGLNPS21} to give a deterministic spanner with embeddings, these methods seem unrelated to the problem of avoiding subsampling the partial trees.

\paragraph{Maxflow / Min-cost flow.}
Over the last several decades there has been extensive work on the
maxflow and minimum cost flow problems
\cite{GR98,GT87,DS08,D70,D73,GG88,BK04,GT88b,G08,GHKKTW15,OG21, H08,
  G95,CKMST11, M13,M16, S13, S17, KLOS14, T85,GT88b,OPT93,O93,O96,
  GT87,GT89b,DPS18, CohenMSV17,KLS20, LS20,AMV20, AMV21:arxiv, BLLSSSW21,GLP21:arxiv,BGJLLPS21:arxiv}. Some
of these algorithms, primarily in the instance of unit-capacity
maxflow \cite{M13,M16,LS20,KLS20}, can be made deterministic using
deterministic flow primitives from \cite{CGLNPS21}.

\paragraph{Connectivity problems.} There is a long line of work on applications of maxflow to connectivity problems, including sparsest cuts, Gomory-Hu trees, and global mincuts \cite{Gusfield90,AroraK07,OrecchiaSVV08,Sherman09,KawarabayashiT19,AbboudK0PST22,ApersW22,NanongkaiSY19,LiNPSY21,AmeranisCOT23,AbboudKT21,Li21}.
Some of these algorithms for global mincut can be made deterministic \cite{KawarabayashiT19,Li21}, though the techniques often rely on expander decomposition, which, again, does not resolve our issue of sampling partial trees.
Since this work was announced, \cite{NSY23} gave a deterministic reduction from $k$-vertex-connectivity to computing $m^{o(1)}k^2$ maxflows to achieve a deterministic algorithm for $k$-vertex-connectivity running in $m^{1+o(1)}k^2$ time.

%% file: prelims.tex
\section{Preliminaries}
\label{sec:prelim}

\paragraph{General notation.} We denote vectors by boldface lowercase letters and matrices by boldface uppercase letters. Often, we use uppercase letters to denote diagonal matrices corresponding to vectors with the matching lowercase letter, e.g., $\mL = \diag(\bell)$. For vectors $\bx, \by$ we define the vector $\bx \circ \by$ as the entrywise product, i.e., $(\bx \circ \by)_i = \bx_i\by_i$. We also define the entrywise absolute value of a vector $|\bx|$ as $|\bx|_i = |\bx_i|$.
For positive real numbers $a, b$ we write $a \approx_{\alpha} b$ for some $\alpha > 1$ if $\alpha^{-1}b \le a \le \alpha b$. 
 For integer $h$ we let $[[h]] \defeq \setof{0,1,\ldots,h}$, and $[h] \defeq \setof{1,\ldots,h}$.
For positive vectors $\bx, \by \in \R^{n}_{>0}$, we say $\bx \approx_{\alpha} \by$ if $\bx_i \approx_{\alpha} \by_i$ for all $i \in [n]$.

\paragraph{Graphs.} We consider multi-graphs $G$ with edge set $E(G)$ and vertex set $V(G)$. When the graph is clear from context, we use $E$ for $E(G)$, $V$ for $V(G)$, $m = |E|$, and $n = |V|$.  We assume that each edge $e \in E$ has an implicit direction
and overload the notation slightly by writing $e = (u, v)$ where $u$ and $v$ are the tail and head of $e$ respectively (note that technically multi-graphs do not allow for edges to be specified by their endpoints).
We let $\rev(e)$ be the edge $e$ reversed: if $e = (u, v)$ points from $u$ to $v$, then $\rev(e)$ points from $v$ to $u$.

A \emph{flow vector} is a vector $\bf \in \R^E$. If $\bf_e \ge 0$, this means that $\bf_e$ units flow in the implicit direction of the edge $e$ chosen, and if $\bf_e \le 0$, then $|\bf_e|$ units flow in the opposite direction. A \emph{demand vector} is a vector $\bd \in \R^V$ with $\sum_{v \in V} \bd_v = 0$.
For an edge $e = (u, v) \in G$ we let $\bb_e \in \R^V$ denote the demand vector of routing one unit from $u$ to $v$, i.e., $\bb_e$ has a $1$ at $u$, $-1$ at $v$, and $0$ elsewhere.
Define the edge-vertex incidence matrix $\mB \in \R^{E \times V}$ as the matrix whose rows are $\bb_e$. We say that a flow $\bf$ routes a demand $\bd$ if $\mB^\top \bf = \bd$.

We denote by $\deg_G(v)$ the combinatorial degree of $v$ in $G$, i.e., the number of incident edges. We let $\Delta_{\max}(G)$ and $\Delta_{\min}(G)$ denote the maximum and minimum degree of graph $G$. We define the volume of a set $S \subseteq V$ as $\vol_G(S) \defeq \sum_{v \in S} \deg_G(v)$.

Given a set of edges $F \subseteq E(G)$, we define $G/F$ to be the graph where the edges in $F$ are contracted. In this paper, typically this operations is performed for forests $F$.

\paragraph{Dynamic Graphs.} In this paper, we say that $G$ is a dynamic graph if it undergoes a sequence of updates. In this paper, the graphs we study will undergo three main types of updates.
\begin{itemize}
    \item \textbf{Edge insertion}: an edge $e = (u, v)$ is added to the graph. The edge is encoded by its endpoints, and when necessary, edge lengths and gradients will also be provided.
    \item \textbf{Edge deletion}: an edge $e = (u, v)$ is deleted from the graph. The edge is encoded by its label in the graph.
    \item \textbf{Vertex split}: a vertex $v$ becomes two vertices $v_1$ and $v_2$, and the edges adjacent to $v$ are split between $v_1$ and $v_2$. Precisely, every edge $e_i = (v, u_i)$ is assigned to either $v_1$ or $v_2$, becoming edge $(v_1, u_i)$ or $(v_2, u_i)$ respectively. This operation is encoded by listing out the edges moved to the one of $v_1, v_2$ with a smaller degree. Thus the encoding size is approximately $\min\{\deg(v_1), \deg(v_2)\}$.
\end{itemize}
In this paper, instead of having our dynamic graphs undergo a single update at a time, we think of them as undergoing \emph{batches} $U^{(1)}, U^{(2)}, \ldots$ of updates, where each batch $U^{(i)}$ denotes a set of updates to apply.

We let $|U^{(t)}|$ denote the total number of updates in the batch, i.e., the total number of edge insertions, deletions, and vertex splits. $\Enc(u)$ of an update $u \in U^{(t)}$ denotes its encoding size. As mentioned above, each insertion and deletion can be encoded in size $\O(1)$, while each vertex split can be encoded in size $\O(\min\{\deg(v_1), \deg(v_2)\})$. Finally, the encoding size of a batch $U^{(t)}$ is the sum of the encoding sizes of each of its updates.

Note that $\Enc(U^{(t)}) = \Omega(|U^{(t)}|)$, but may be even larger. However, we can bound the total encoding size using the following lemma.
\begin{lemma}
\label{lemma:encodingSize}
For a dynamic graph $G$ that undergoes batches of updates $U^{(1)}, U^{(2)}, \ldots$ if $G$  initially has $m$ edges then we can bound the total encoding size as
$ \sum_t \Enc(U^{(t)}) = \O\left(m + \sum_t |U^{(t)}|\right).$
\end{lemma}
\begin{proof}
  Each edge insertion/deletion only contributes $\O(1)$ to the
  encoding size. Thus, the size of encodings of edge/insertions
  deletions is at most $\O\left (\sum_t |U^{(t)}|\right).$ In
  order to account for vertex splits, consider the potential $\Phi = \sum_v \deg(v) \log \deg(v)$. It is straightforward to verify that an edge insertion can only increase the potential by $O(\log m)$. When a vertex $v$ is split into $u_1, u_2$, the potential decreases by at least $\Omega(\min(\deg(u_1), \deg(u_2)))$.
\end{proof}

\paragraph{Paths, Flows, and Trees.} Given a path $P$ in $G$ with vertices $u,v$ both on $P$, then we let $P[u,v]$, which is another path, denote the path segment on $P$ from $u$ to $v$. We note that if $v$ precedes $u$ on $P$, then the segment $P[u,v]$ is in the reverse direction of $P$.
For a $a$ to $b$ path $P$ and a $b$ to $c$ path $Q$ we let $P \oplus Q$ denote the $a$ to 
$c$ path that is the concatenation of $P$ and $Q$.

For a forest $F$, we use $F[u,v]$ to denote the unique simple path from $u$ to $v$ along edges in the forest $F$; we ensure that $u, v$ are in the same connected component of $F$ whenever this notation is used. Additionally, we let $\bp(F[u,v]) \in \R^{E(G)}$ denote the flow vector which routes one unit from $u$ to $v$ along the path in $F$. Thus, $|\bp(F[u,v])|$ is the indicator vector for the path from $u$ to $v$ on $F$. Note that $\bp(F[u,v])+\bp(F[v,w])=\bp(F[u,w])$ for any vertices $u, v, w \in V$.

The \emph{stretch} of $e = (u, v)$ with respect to a tree $T$ and lengths $\bell \in \R^E_{>0}$ is defined as
\[ \str^{T,\bell}_e \defeq 1 + \frac{\langle \bell, |\bp(T[u, v])|\rangle}{\bell_e} = 1 + \frac{\sum_{e' \in T[u,v]} \bell_{e'}}{\bell_e}. \]
This differs slightly from the more common definition of stretch because due to the additive $1$; we choose this definition to ensure that $\str^{T,\bell}_e \ge 1$ for all $e$. We define the stretch of an edge $e = (u, v)$ with respect to a forest $F$ analogously if $u, v$ are in the same connected component of $F$. Later in \Cref{def:stretchf}, we introduce a notion of stretch when $u, v$ are \emph{not} in the same component of a rooted forest. In this case, the stretch is instead defined as the total distance of $u, v$ to their respective roots divided by the length of $e$. As stated in the following theorem, it is known how to efficiently construct trees with polylogarithmic average stretch with respect to underlying weights; we call these low-stretch spanning trees (LSSTs).
\begin{theorem}[Static LSST \cite{AN19:journal}]
\label{thm:an}
Given a graph $G = (V, E)$ with lengths $\bell \in \R^E_{>0}$ and
weights $\bv \in \R^E_{>0}$ there is an algorithm that runs in time
$\O(m)$ and computes a tree $T$ such that
$\sum_{e \in E} \bv_e \str^{T, \bell}_e \le \gamma_{LSST}\|\bv\|_1$
for some $\gamma_{LSST} \defeq O(\log n \log \log n)$.
\end{theorem}
In this paper, in contrast to eg., \cite[Lemma 6.5]{chen2022maximum}, we often use the cruder upper bound of of $\gamma_{LSST} = O(\log^2 n)$. We do this to simplify the presentation as it does not effect the final asymptotic bounds claimed.

\paragraph{Graph Embeddings.}
Given weighted graphs $G$ and $H$ with $V(G) \subseteq V(H)$, we say that $\Pi_{G \xrightarrow{} H}$ is a \emph{graph-embedding} from $G$ into $H$ if it maps each edge $e^G = (u,v) \in E(G)$ to a $u$-$v$ path $\Pi_{G \xrightarrow{} H}(e^G)$ in $H$.
Let $\bw_G$ be the weight function of $G$ and $\bw_H$ be the weight function of $H$. We define the \emph{congestion} of an edge $e^H$ by
\[
\econg(\Pi_{G \xrightarrow{} H}, e^H) \defeq \frac{\sum_{ e^G \in E(G) \text{ with } e^H \in \Pi_{G \xrightarrow{} H}(e^G)} \bw_G(e^G)}{\bw_H(e^H)}
\] 
and the congestion of the embedding by $\econg(\Pi_{G \xrightarrow{} H}) \defeq \max_{e^H \in E(H)} \econg(\Pi_{G \xrightarrow{} H}, e^H)$. Analogously, the congestion of a vertex $v^H \in V(H)$ is defined by 
\[
\vcong(\Pi_{G \xrightarrow{} H}, v^H) \defeq \sum_{e^G \in E(G) \text{ with } v^H \in \Pi_{G \xrightarrow{} H}(e^G)} \bw_G(e^G)
\] and the 
vertex-congestion of the graph-embedding by \[ \vcong(\Pi_{G \xrightarrow{} H}) \defeq \max_{v^H \in V(H)} \vcong(\Pi_{G \xrightarrow{} H}, v^H). \] We define the length of the embedding by $\length(\Pi_{G \xrightarrow{} H}) \defeq \max_{e^G \in E(G)} |\Pi_{G \xrightarrow{} H}(e^G)|$.

Given graphs $A, B, C$ and graph-embeddings $\Pi_{B \to C}$ from $B$ into $C$ and $\Pi_{A \to B}$ from $A$ to $B$. We denote by $\Pi_{B \to C} \circ \Pi_{A \to B}$ the graph embedding of $A$ into $C$ obtained by mapping each edge $e^A = (u,v) \in E(A)$ with path $\Pi_{A \to B}(e^A) = e^B_1 \oplus e^B_2 \oplus \ldots \oplus e^B_k$ in $B$ to the path $\Pi_{B \to C}(e^B_1) \oplus \Pi_{B \to C}(e^B_2) \oplus \ldots \oplus \Pi_{B \to C}(e^B_k)$. The following useful fact is straightforward from the definitions. 

\begin{fact}\label{fact:transitiveEmbeddingCong}
Given graphs $A, B, C$ and graph-embeddings $\Pi_{B \to C}$ from $B$ into $C$ and $\Pi_{A \to B}$ from $A$ to $B$. Then, $\vcong(\Pi_{B \to C} \circ \Pi_{A \to B}) \leq \vcong(\Pi_{B \to C}) \cdot \econg(\Pi_{A \to B})$.
\end{fact}

\paragraph{Computational Model.} For problem instances encoded with $z$ bits, all algorithms developed in this paper work in fixed-point arithmetic where words have $O(\log^{O(1)} z)$ bits, i.e., we prove that all numbers stored are in $[\exp(-\log^{O(1)}z), \exp(\log^{O(1)}z)]$. In particular, \cref{thm:ipmMain} says that the min-ratio cycle problems solved by our algorithm satisfy \cref{def:hiddenStableFlowChasing}, where item \ref{item:quasipoly} says that all weights and lengths are bounded by $\exp(\log^{O(1)}m)$.

%% file: approach.tex
\section{Technical Overview}
\label{sec:tech_overview}
Our approach for obtaining a \emph{deterministic} almost-linear time min-cost flow algorithm follows the framework of the recent randomized algorithm in \cite{chen2022maximum}. We start by reviewing the algorithm in \cite{chen2022maximum} and then lay out the challenges in obtaining deterministic analogs of its randomized components. By scaling arguments (see \cite[Lemma C.1]{chen2022maximum}), we
assume that $U, C \le m^{O(1)}$.

\subsection{The Randomized 
Algorithm in \texorpdfstring{\cite{chen2022maximum}}{chen2022}}
\label{sub:rand_summary}

\paragraph{The Outer-Loop: An $\ell_1$-Interior Point Method.} The starting point for the randomized algorithm in \cite{chen2022maximum} is a new $\ell_1$-interior point method (IPM), which is actually completely \emph{deterministic}. This method uses a \emph{potential reduction} IPM inspired by \cite{K84}, where in each iteration, the potential function
$
    \Phi(\bf) \defeq 20m \log(\bc^\top \bf - F^*) + \sum_{e \in E} \left((\bu^+_e - \bf_e)^{-\alpha} + (\bf_e - \bu^-_e)^{-\alpha} \right)
$
is reduced. Here, $\alpha = 1/\Theta(\log m)$, but the reader can think of the barrier $x^{-\alpha}$ as the more standard $-\log x$ for simplicity. 

\cite{chen2022maximum} showed that one can assume that an initial feasible solution $\bf^{(0)}$ is given that routes the demand and has $\Phi(\bf^{(0)}) \leq O(m \log m)$ and that the IPM can be terminated once the potential function value is at most $- 200m\log m$, as at this point, one can round the flow to an exact solution using an isolation lemma; see \cite[Lemma 4.11]{chen2022maximum}. While this step is randomized, it can easily be derandomized using an alternate flow rounding procedure, as is explained later at the start of \cref{subsec:detmincost}. We next discuss how to achieve a potential reduction of $\Phi(\bf)$ in each iteration by $m^{-o(1)}$. This yields that the IPM terminates within $m^{1+o(1)}$ steps. 

To obtain a potential reduction of $m^{-o(1)}$ at each step, given a current feasible flow $\bf$, the update problem involves finding an update direction $\bDelta$ to update the flow to $\bf + \bDelta$ such that  (a) $\bDelta$ is a circulation, i.e., adding it to $\bf$ does not change the net routed demands and (b) $\bDelta$ approximately minimizes the inner product with a linear function (the gradient of $\Phi$), relative to an $\ell_1$-norm that arises from the second derivatives of $\Phi$.
Letting $\bg \in \R^{E}$ denote this gradient and letting $\bell \in \R^{E}_+$ be the edge length (both with respect to the current flow $\bf$), we can write the update problem as
\begin{align}
\min_{\substack{\bDelta \in \R^E : \mB^\top \bDelta = \textbf{0} }} \frac{\bg^\top \bDelta}{\norm{\diag(\bell) \bDelta}}_1. \label{eq:flowMinRatio} \end{align}
We refer to this update problem henceforth as the \emph{min-ratio cycle problem}, since, by a cycle-decomposition argument, the optimal value is always realized by a simple cycle. As shown in \cite{chen2022maximum}, the update problem has several extremely useful properties:
\begin{enumerate}
    \item At every time step $t$, the direction from the current solution $\bf^{(t)}$ towards the optimal flow $\bf^*$, henceforth called the \emph{witness} $\bDelta^{(t)} \defeq \bf^* - \bf^{(t)}$, achieves  $\frac{\bg^\top \bDelta^{(t)}}{\norm{\diag(\bell) \bDelta^{(t)}}}_1 \le -\frac{1}{\Theta(\log m)}.$
    \item Performing the update with a cycle $\bDelta$ with $\frac{\bg^\top \bDelta}{\norm{\diag(\bell) \bDelta}}_1 = -\kappa$
    reduces the potential by $\Omega(\kappa^2)$.
    Thus, even finding an $m^{o(1)}$-approximate min-ratio cycle reduces the potential by $m^{-o(1)}$.
    After $m^{1+o(1)}$ iterations of updates, the potential becomes small enough, and we can round the current flow to an exact solution.
    \item 
    The convergence rate is unaffected if we use approximations $\hat{\bg}$ and $\hat{\bell}$ of the gradient $\bg$ and the lengths $\bell$ such that both $\hat{\bg}$ and $\hat{\bell}$ are updated only a total of $m^{1+o(1)}$ times (here, by update we mean that a single entry of $\hat{\bg}$ and $\hat{\bell}$ is changed) throughout the entire algorithm.
\end{enumerate}
In this way, the $\ell_1$-IPM gives a \emph{deterministic} reduction of (exact) min-cost flow to solving a sequence of stable min-ratio cycle problems.

\paragraph{A Data Structure for the Min-Ratio Cycle Problem.}
Since when solving min-cost flow by approximately solving a sequence of min-ratio cycles, the underlying graph remains the same, and gradient and lengths change sporadically throughout the algorithm, it is useful to think about the repeated solving of the min-ratio cycle problem as a \emph{data structure problem}. This problem is is formalized in \cref{def:ipm:datastructure}. \cite{chen2022maximum} designs a \emph{randomized} data structure for the min-ratio cycle problem which supports the following operations:
\begin{itemize}
    \item $\textsc{Initialize}(G, \hat{\bg}^{(0)}, \hat{\bell}^{(0)})$: initialize the data structure for graph $G$ and the initial approximate gradients, $\hat{\bg}^{(0)}$, and lengths, $\hat{\bell}^{(0)}$, on the edges of $G$.
    \item $\textsc{Update}(\hat{\bg}^{(t)}, \hat{\bell}^{(t)}):$ the $t$-th update replaces current gradient and lengths by $\hat{\bg}^{(t)}$ and $\hat{\bell}^{(t)}$.  
    \item $\textsc{Query}():$ returns a cycle whose ratio with respect to the current gradient $\hat{\bg}^{(t)}$ and lengths $\hat{\bell}^{(t)}$ is within a $m^{o(1)}$ factor of $\bDelta^{(t)} = \bf^* - \bf^{(t)}$.
\end{itemize}
In the $\textsc{Update}$ operation, $\hat{\bg}^{(t)}, \hat{\bell}^{(t)}$ are described by their changes from $\hat{\bg}^{(t-1)}, \hat{\bell}^{(t-1)}$. By the above discussion, there are $m^{o(1)}$ coordinate changes on average per instance.

Note that the output cycle returned by $\textsc{Query}()$ may have nonzero flow on $\Omega(n)$ edges for each of the $\Omega(m)$ iterations (this is often referred to as the \emph{flow decomposition barrier}). Thus we cannot efficiently, explicitly output the solutions. To overcome this issue, the data structure in \cite{chen2022maximum} maintains a $s = m^{o(1)}$ spanning trees $T_1, T_2, \ldots, T_s$ of the graph $G$. Each such tree is itself a dynamic object, i.e., these trees undergo changes over time in the form of edge insertions and deletions. However, the total number of such edge insertions and deletions is at most $m^{1+o(1)}$ when amortizing over the sequence of updates generated by the IPM. Using these dynamic trees $T_1, T_2, \ldots, T_s$, whenever the operation $\textsc{Query}()$ is invoked, the data structure in \cite{chen2022maximum} finds (with high probability) an approximate min-ratio cycle that consists of $m^{o(1)}$ subpaths of a tree $T_i$ and $m^{o(1)}$ additional edges. Using the start and end points of each tree path, the query operation can encode each solution efficiently, as desired.

As shown in \cite{chen2022maximum}, the data structure can overall be implemented to run in amortized $m^{o(1)}$ time per query and update, yielding an almost-linear algorithm for the min-cost flow problem.

\paragraph{Maintaining Trees in the Data Structure.} It remains to review how the data structure in \cite{chen2022maximum} \emph{efficiently} maintains a set of dynamic trees $T = \{T_1, T_2, \ldots, T_s\}$ such that one of the trees yields a cycle with sufficient ratio with high probability, and how to query this cycle.

To construct the set of trees $T$, \cite{chen2022maximum} draws on the theory of low-stretch spanning trees (LSSTs). Let $G = G^{(0)}$ be the original graph whose edge lengths are given by the vector $\hat{\bell}^{(0)}$. \cite{chen2022maximum} applied a standard multiplicative weights argument \cite{M10,S13,KLOS14} to construct a set of $k$ (partial) trees $T_1, \dots, T_k$ such that every edge $e$ had average stretch $\O(1)$ over these $k$ trees. Thus, if $\hat{\bell}$ is the vector of stretches of a random tree among the $T_i$, then the \emph{witness} $\bDelta^{(0)} = \bf^* - \bf^{(0)}$ satisfies in expectation $\| \diag(\Tilde{\bell}) \bDelta^{(0)} \|_1 \le m^{o(1)} \| \diag(\hat{\bell}^{(0)}) \bDelta^{(0)} \|_1$
By Markov's inequality, this same guarantee (up to constants) must hold with probability at least $1/2$.
When this occurs, we say the \emph{stretch of the witness with respect to the tree} is low.
By sampling $O(\log m)$ trees among $\{T_1, \dots, T_k\}$ \cite{chen2022maximum} ensures that this occurs in at least one tree with high probability.
A basic flow decomposition result then implies that one of the \emph{fundamental cycles} formed by an off-tree edge and the tree-path (in $T$) between its endpoints yields an $m^{o(1)}$-approximate solution. 

Now, consider what happens after the current flow solution is changed from $\bf^{(0)}$ to $\bf^{(1)}$ by adding the first update.
This changes the witness from $\bDelta^{(0)}$ to $\bDelta^{(1)} = \bf^* - \bf^{(1)}$ and changes the (approximate) gradient from $\hat{\bg}^{(0)}$ to $\hat{\bg}^{(1)}$, and lengths from $\hat{\bell}^{(0)}$ to $\hat{\bell}^{(1)}$. To solve the next update problem, the sampled trees have to be updated so that the stretch of the new witness $\bDelta^{(1)}$ is again low with respect to at least one of the trees (now with respect to $\hat{\bell}^{(1)}$).

To update the trees, for each sampled tree $T$ some edges are removed and then replaced by new edges. To obtain an efficient implementation, \cite{chen2022maximum} applies the well-established technique of maintaining a hierarchy of partial trees/forests. At each level, a partial tree is computed, and the next level then finds again a partial tree in the graph where edges in the partial tree at the higher levels are contracted. Let us illustrate how such a partial tree is found. At the highest level of the hierarchy, a partial tree/forest $F$ is computed with $m/k$ connected components for some target value $k = m^{o(1)}$. $F$ is computed so that it only undergoes edge deletions, and at most $\O(1)$ per update. Additionally, either a fundamental cycle of $F$ or a cycle in $G/F$ (the graph where $F$ is contracted) has ratio within $\O(1)$ factor of the desired min-ratio cycle in $G$. This is illustrated in 
\Cref{fig:contraction}.
This reduces the problem of finding a min-ratio cycle mainly to finding such a cycle in the graph $G / F$, which has at most $m/k$ vertices. 

We refer to the step of maintaining $F$ and contracting to the graph $G/F$ as the \emph{vertex sparsification} phase. However, $G / F$ might still contain almost all edges of $G$. 
To reduce the edge count, \cite{chen2022maximum} computes a spanner $G'$ of $G/F$ that yields a reduction in the number of edges to roughly $m/k$. We refer to this as \emph{edge sparsification}, and give a more detailed overview of the construction in \cite{chen2022maximum} below. The spanner also allows us to either obtain the solution to the min-ratio cycle problem directly from the spanner construction, or approximately preserve the solution quality in $G'$. The algorithm then recurses, again building a partial tree (forest) $F'$ on $G'$, finding a spanner, and so on. The tree $T$ is taken as the union of the contracted forests $F, F'$, and the forests found in deeper recursion levels. 

Whenever lengths and gradients change, updates are handled by making a few adjustments to the forest and then propagating changes to the deeper levels. By controlling carefully the propagation, the total number of updates across levels remains small.

\begin{figure}[ht]
\centering
\includegraphics[width=0.7\textwidth]{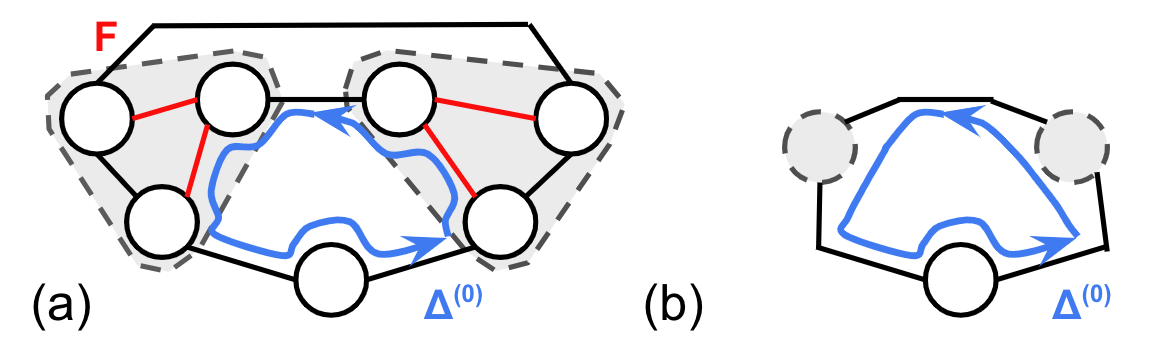}
\caption{In (a), we see a graph $G$ and a forest $F$ (the subgraph shown with red edges). The two non-trivial connected components of $F$ are encircled by a dotted border bash. In blue, we show the witness circulation $\bDelta^{(0)}$. In (b), we see the graph $G/F$ obtained by contracting the components of $G$ and the circulation $\bDelta^{(0)}$ again in blue mapped to $G/F$.
The algorithm ensures that since each contracted edge is approximated well by a path in $F$, either a solution to (\ref{eq:flowMinRatio}) is formed by one of the fundamental cycles, or that the mapped circulation in $G/F$ is a good solution to (\ref{eq:flowMinRatio}).}
\label{fig:contraction}
\end{figure}

\paragraph{Edge Sparsification.} As the forest $F$ undergoes edge deletions, the graph $G/F$ undergoes edge deletions and vertex splits. To design an algorithm to maintain a spanner of $G/F$, \cite{chen2022maximum} gave a deterministic reduction from maintaining a spanner in an unweighted graph under edge deletions and vertex splits to statically constructing a spanner with low-congestion \emph{edge embeddings}. This means that for a spanner $H \subseteq G$, every edge $e \in E(G) \setminus E(H)$ is mapped into a short path $\Pi_{G \to H}(e)$ in $H$ between its endpoints with at most $m^{o(1)}$ edges, such at every vertex in $H$ has at most $m^{o(1)} \deg_G(v)$ paths through it. Note that at least $\deg_G(v)$ paths go through $v$ in any embedding, so having low vertex congestion means that we match this bound up to the $m^{o(1)}$ factor.
It is worth noting that while we only require a spanner of $G/F$ for the algorithm, the reduction only works with a low-congestion embedding. Additionally, our dynamic low-stretch tree data structure makes use of this additional low-congestion property.

Thus, we focus on statically constructing spanners of unweighted graphs with low-congestion embeddings. \cite{chen2022maximum} designed the following algorithm to achieve this. We may assume that $G/F$ has is unweighted by the standard trick of bucketing edges in $\O(1)$ groups whose lengths are within a factor of $2$. First, the graph $G/F$ is decomposed into expanders $H_1, H_2, \ldots, H_{\ell}$, where each vertex appears in $O(\log m)$ expanders, and each expander is almost-uniform-degree in that every degree is within an $O(\log m)$ factor of the average. Thus, for each graph $H_i$, \emph{random} sampling each edge uniformly with probability about $\O(1)$ divided by the degree of $H_i$ yields a graph $H'_i$ that again is an almost-uniform-degree expander, except now with polylogarithmic degrees. We let the spanner $G'$ of $G$ be the union of all such sampled graphs $H_i'$ and clearly $G'$ is sparse, i.e., has at most $\tilde{O}(|V(G/F)|)$ edges. This is the only randomized component of the edge sparsification step.

While proving that $G'$ is a spanner of $G$ is rather straightforward, we also must construct a low-congestion embedding of $G$ into $G'$. In \cite{chen2022maximum}, this is achieved by embedding each graph $H_i$ into the corresponding down-sampled graph $H_i'$ for every $i$. In \cite{chen2022maximum}, this is achieved by a deterministic procedure that internally uses the decremental shortest paths data structure on expanders by Chuzhoy and Saranurak \cite{CS21}. Finally, \cite{chen2022maximum} takes the embedding from $G$ into $G'$ to be the union of the embeddings from $H_i$ to $H_i'$ for all $i$.

We conclude our discussion on edge sparsification by describing how to find a min-ratio cycle from the spanner construction. Given a spanner $G'$ of $G/F$ with embedding, a flow decomposition arguments shows that either some \emph{spanner cycle} $e \oplus \Pi_{(G/F) \to G'}(e)$ has ratio within $m^{o(1)}$ of $\bDelta^{(t)}$, or the
circulation in $G'$ achieved by routing $\bDelta^{(t)}$ along the paths $\Pi_{(G/F) \to G'}$ into $G'$ has ratio within $m^{o(1)}$ of $\bDelta^{(t)}$. By maintaining the paths $\Pi_{(G/F) \to G'}(e)$ explicitly, and recursing on $G'$, our data structure can efficiently query for a min-ratio cycle. This argument is covered in more detail in \Cref{sec:routing}.

\paragraph{A Note on the Interaction Between Data Structure and Witness.}
The above description of the data structure is an oversimplification and hides many key details. Perhaps most importantly, the proof of correctness for the data structure in \cite{chen2022maximum} crucially hinges on the existence of the witness $\bDelta^{(t)} = \bf^{*}-\bf^{(t)}$ in order to show that the near-optimal cycle $\bDelta^{(t)}$ does not ever incur too much stretch even under possibly adaptive updates. Put another way, the data structure does \emph{not} work against general adaptive adversaries, whose updates can depend on the randomness of the data structure, but can be used to solve min-cost flow due to the stability of the witness solution. Similarly, in this paper we do not design a deterministic data structure for general min-ratio cycle instances. Instead, we also require that the update sequence admits a stable witness; leveraging the stable witness in both cases require modifications to both the LSST and spanner data structures, and these are deferred to the main body of the paper.

\subsection{A Deterministic Min-cost Flow Algorithm}
\label{subsec:detmincost}

Building on the exposition of the algorithm in \cite{chen2022maximum} given in \Cref{sub:rand_summary}, we are now ready to discuss the key changes necessary to obtain our deterministic algorithm. Here, we highlight the parts of \cite{chen2022maximum} that required randomization and outline strategies to remove the randomness.

\paragraph{Derandomizing the IPM Framework.}
The main challenge in derandomizing the framework of \cite{chen2022maximum} is in derandomizing the vertex and edge sparsification routines and solving the requisite dynamic min-ratio cycle problem. Indeed, derandomizing the remainder of the IPM framework is straightforward, because both the IPM presented in the last section and the procedure that maintains the approximate gradient $\hat{\bg}$ and the lengths $\hat{\bell}$ are completely deterministic. The only use of randomness in the above approach, beyond the min-ratio cycle data structure, occurred as \cite{chen2022maximum} rounded the solution when the potential is sufficiently small, i.e.  $\Phi(\bf) \le - \Omega(m\log m)$, via the Isolation Lemma. However, the use of the Isolation Lemma can be replaced by a deterministic flow rounding procedure using Link-Cut trees \cite{ST83} as was shown in \cite{KangP15} (see \Cref{lemma:detflowrounding}).

In the min-ratio cycle data structure of \cite{chen2022maximum}, there are two randomized components:
\begin{enumerate}
\item $\O(1)$ forests are sampled at each level of the hierarchy, and
\item The spanner of $G/F$ is constructed by decomposing $G/F$ into expanders, and random sampling within each expander.
\end{enumerate}
Below we discuss how to remove the randomness from the first \emph{vertex sparsification} step, and then discuss the second \emph{edge sparsification} step.

\paragraph{Derandomizing Vertex Sparsification.}
Recall that the vertex sparsification construction described above computes a set of $k$ forests $F_1, \dots, F_k$. Of these, $\O(1)$ are sampled, and for each sampled forest $F$, the algorithm recurses on $G/F$. A natural approach to derandomize this is to instead recurse on \emph{all} $k$ forests in the collection to deterministically ensure that some forest has low stretch of the witness $\bDelta^{(t)}$. Unfortunately this is too expensive, as it leads to $\Omega(m)$ trees being maintained overall. Additionally, every update to the input graph may change every tree, and this approach would therefore lead to linear time per update which is far more than we can afford.

However, we show that, somewhat surprisingly, the following strategy works: instead of directly recursing on all forests, and therefore, on all graphs $G / F_1, G / F_2, \ldots, G / F_s$, we can recurse only on the first such tree $G / F_1$ and check if we find a solution to the min-ratio cycle problem. If we do, we do not need to check $G / F_2, \ldots, G / F_s$ at that moment. Otherwise, we move on to $G / F_2$ (we refer to this as a \emph{shift}), but now know that $G / F_1$ did at some point not give a solution to the min-ratio cycle and $F_1$ is therefore not a good forest so that we never have to revisit it. Carefully shifting through these forests $F_1, F_2, \ldots, F_s$, it then suffices to only forward the updates to $G$ (in the form of updates to $\hat{\bell}$) to the forest that is currently used. When we move to the next forest after failing to identify a solution to the min-ratio cycle problem, we then apply all updates that previously happened to $G$ to the contracted graph. We apply this shifting procedure recursively.

Now we need to understand why this improved the amortized update time to $m^{o(1)}$, and perhaps more interestingly, why the algorithm finds an approximate min-ratio cycle without cycling through many graphs $G/F_i$ over the course of the algorithm. At a high level, the runtime is acceptable because the number of dynamic updates to the forests $F_i$ is at most that of the randomized case, as we only maintain a single branch of the recursion at a time. Shifting between forests does not cause dynamic updates, and thus can be charged to the original construction cost.

To understand why the algorithm does not have to shift through several graphs $G/F_i$ every iteration, recall that the witness $\bDelta^{(t)} = \bf^* - \bf^{(t)}$ is stable in that only $\O(1)$ edges values change by a constant factor multiplicatively on average per iteration, and that these edges are passed to the data structure. This allows us to show that if we find a forest $F_i$ whose stretch against the witness $\bDelta^{(t)}$ was small, then it stays small until we must rebuild after about $m/k$ updates, or $\|\diag(\bell^{(t)}) \bDelta^{(t)}\|_1$ decreases by a constant, which can only happen $\O(1)$ times. This way, over the course of $m/k$ updates, our algorithm only needs to shift $\O(k)$ total times. The major challenge towards formalizing this analysis is that the data structure has multiple levels, which severely complicates the condition that a forest $F_i$ maintains small stretch for several iterations, because we do not know which level caused the failure. We analyze this algorithm through what we call the \emph{shift-and-rebuild game} (\Cref{sec:rebuilding}), a generalization of the (simpler) rebuilding game in \cite{chen2022maximum}.

\paragraph{Derandomizing Edge Sparsification.} From the description given above, the main challenge for derandomization of the edge sparsification procedure from \cite{chen2022maximum} is to find a spanner $H'_i$ of an almost-uniform-degree expander $H_i$ such that $H'_i$ consists of few edges and such that we can find a small vertex congestion short-path embedding of the graph $H_i$ into $H_i'$. 

We use the following natural derandomization approach: given $H_i$ with maximum-degree $d_i^{\max}$, we first deterministically construct a constant-degree expander $W$ over the vertex set of $H_i$. Using the tools from \cite{chen2022maximum}, we then compute an embedding $\Pi_{W \to H_i}$ from $W$ into the graph $H_i$ with $m^{o(1)}$ vertex congestion using only short paths. Reusing these tools, we also compute an embedding $\Pi_{H_i \to W}$ from $H_i$ into $W$ with $m^{o(1)} \cdot d_i^{\max}$ \emph{edge congestion} using only short paths.

Now consider the embedding given by $\Pi_{W \to H_i} \circ \Pi_{H_i \to W}$ which maps edges from $H_i$ to paths in $W$ and then back to paths in $H_i$. We claim that the graph $H_i'$ consisting of the edges in the image of $\Pi_{W \to H_i}$ is a spanner, and $\Pi_{W \to H_i} \circ \Pi_{H_i \to W}$ embeds $H_i$ into $H_i'$ with low vertex congestion and short paths.
To see this, we first show that $H_i'$ is sparse, i.e., it has at most $|V(H_i)|m^{o(1)}$ edges. This follows because $W$ has only $O(|V(H_i)|)$ edges by construction, and each edge is mapped to a path of length at most $m^{o(1)}$. Thus the image of $\Pi_{W \to H_i}$ consists of at most $|V(H_i)|m^{o(1)}$ edges. Further, using \Cref{fact:transitiveEmbeddingCong}, we immediately obtain that $\Pi_{W \to H_i} \circ \Pi_{H_i \to W}$ has vertex congestion $m^{o(1)} \cdot d_i^{\max}$ and it is not hard to see that each embedding path in $\Pi_{W \to H_i} \circ \Pi_{H_i \to W}$ is short.

There are some additional side constraints that the
spanners need to satisfy to work in the framework of our overall
algorithm, relating to leveraging the stability of the witness
$\bDelta^{(t)}$. Ensuring that these constraints are met requires additional careful analysis, which we give in \Cref{sec:spanner}.

\paragraph{Dynamic Low-Stretch Trees.} Our algorithm that dynamically maintains low-stretch trees uses a very similar hierarchical data structure as to our dynamic min-ratio cycle algorithm. At the top level, we statically compute a low-stretch tree, and maintain a partial forest $F$ with $O(m/k)$ connected components under edge updates. We then maintain a spanner of $G/F$ with explicit edge embeddings by applying the deterministic edge sparsification algorithm described above. Finally, we recurse on the spanner of $G/F$.

%% file: ipm_algorithm.tex
\section{Flow Framework}
\label{sec:framework}
In this section, we discuss our main algorithm for solving flow problems to high accuracy.

We first note that in order to solve a min-cost flow problem exactly, it suffices to find a good enough fractional solution. We use the following result which, as an immediate corollary, shows that to solve min-cost flow it suffices to find a feasible fractional
flow $\bf$ with a cost that is within an additive $\nicefrac{1}{2}$ of the optimal
cost.

\begin{lemma}[{\cite[Section 4]{KangP15}}]
\label{lemma:detflowrounding}
There is a deterministic algorithm which when given a feasible fractional flow $\bf$ in a $m$-edge $n$-vertex mincost flow instance with integer capacities outputs a feasible integer flow $\bf'$ with cost no larger than $\bf$, in $O(m \log m)$ time.
\end{lemma}

To find such an approximate min-cost flow, we use the IPM algorithm introduced in Chen et al.~\cite{chen2022maximum} that can find an almost-optimal fractional solution to the min-cost flow problem by solving a sequence of min-ratio cycle problems. In order to state the guarantees of the algorithm, we first define the min-ratio cycle problem and the dynamic variant of it that we consider.

\begin{definition}[Min-Ratio Cycle]
Given a graph $G(V, E),$ gradients $\bg \in \R^E$, and lengths $\bell \in \R_{>0}^E,$ the \emph{min-ratio cycle problem} seeks a circulation $\bDelta$ satisfying $\mB^\top\bDelta=0$ 
that (approximately) minimizes $\frac{\l \bg, \bf \r}{\|\mL\bf\|_1}$ where $\mL = \diag(\bell)$.
\end{definition}

Observe that the minimum objective value of the min-ratio cycle problem is non-positive since for any circulation $\bDelta$, the flow $-\bDelta$ is also a circulation.

Extending the problem definition to the dynamic setting, a dynamic min-ratio cycle problem with $T$ instances is described by a dynamic graph $G^{(t)}$, gradients $\bg^{(t)} \in \R^E$, and lengths $\bell^{(t)} \in \R_{>0}^E$, where the dynamic graph is undergoing a batch of updates $U^{(1)},\ldots, U^{(T)}.$

The IPM algorithm from \cite{chen2022maximum} requires solving a dynamic min-ratio cycle problem.
The data structure from Chen et al. for solving the dynamic min-ratio cycle problem requires a stability condition, roughly requiring that there is a dynamic witness for the problem instances whose length changes slowly across iterations. This condition is captured in the following definitions: 

\begin{definition}[Valid pair]
\label{def:validpair}
For a graph $G = (V, E)$ with lengths $\bell \in \R^E_{>0}$, we say that $\bc, \bw \in \R^E$ are a \emph{valid pair} if $\bc$ is a circulation and $|\bell_e\bc_e| \le \bw_e$ for all $e \in E$.
\end{definition}
\begin{restatable}[Hidden Stable $\alpha$-Flow Updates]{definition}{defHiddenStableFlowChasing}
  \label{def:hiddenStableFlowChasing}
  We say that a dynamic min-ratio cycle instance described by a dynamic graph $G^{(t)},$ gradients $\bg^{(t)},$ and lengths $\bell^{(t)}$
  satisfies the \emph{hidden stable $\alpha$-flow chasing} property if there are hidden dynamic circulations $\bc^{(t)}$ and hidden dynamic upper bounds $\bw^{(t)}$ such that the following holds at all stages $t$:
  \begin{enumerate}
  \item \label{item:circulation}
  $\bc^{(t)}$ is a circulation, i.e., $\mB_{G^{(t)}}^\top \bc^{(t)} = 0$.
  \item \label{item:width}
  $\bc^{(t)}$ and $\bw^{(t)}$ are a valid pair with respect to $G^{(t)}$.
  \item $\bc^{(t)}$ has sufficiently negative objective value relative to $\bw^{(t)},$  i.e.,  $\frac{\langle \bg^{(t)}, \bc^{(t)}\rangle}{\|\bw^{(t)}\|_1} \le -\alpha.$
  
  \item \label{item:widthstable}
  For any edge $e$ in the current graph $G^{(t)}$, and any stage $t' \leq t$, if the edge $e$ was not explicitly inserted after stage $t'$, then $\bw^{(t)}_e \le 2 \bw^{(t')}_e$.
  However, between stage $t'$ and $t$, endpoints of edge $e$ might change due to vertex splits.
  \item \label{item:quasipoly}
  Each entry of $\bw^{(t)}$ and $\bell^{(t)}$ is quasipolynomially lower and upper-bounded:
    \[ \log \bw^{(t)}_e \in [-\log^{O(1)} m, \log^{O(1)} m] \text{   and   } \log \bell^{(t)}_e \in [-\log^{O(1)} m, \log^{O(1)} m] \forall e \in E(G^{(t)}). \]
  \end{enumerate}
\end{restatable}
Intuitively, \cref{def:hiddenStableFlowChasing} says that even while $\bg^{(t)}$ and $\bell^{(t)}$ change, there is a witness circulation $\bc^{(t)}$ that is fairly stable.
In particular, there is an upper bound $\bw^{(t)}$ on the coordinate-wise lengths of $\bc^{(t)}$ that stays the same up to a factor of $2$, except on edges that are explicitly updated.
Interestingly, even though both $\bc^{(t)}$ and $\bw^{(t)}$ are hidden from the data structure, their existence is sufficient to facilitate efficient implementations.
For brevity, use the term \emph{Hidden Stability} to refer to \cref{def:hiddenStableFlowChasing} in the rest of the paper.

\begin{definition}\label{def:ipm:datastructure}
The problem of $\kappa$-approximate Dynamic Min-Ratio Cycle with Hidden Stability asks for a data structure that, at every stage $t$, finds a
  circulation $\bDelta^{(t)},$ i.e.,
  $\mB_{G^{(t)}}^\top \bDelta^{(t)} = 0$ such that
  $ \frac{\langle \bg^{(t)}, \bDelta\rangle}{\|\mL^{(t)}\bDelta\|_1}
  \le -\kappa\alpha$. Additionally, we require that the data structure maintains a flow $\bf \in \R^E$
  that is initialized at $\mathbf{0}$, and supports the following
  operations:
  \begin{enumerate}
  \item \label{item:positiveflow} 
    $\textsc{Update}(U^{(t)}, \bg^{(t)}, \bell^{(t)},
    \eta)$.
    Apply edge insertions/deletions specified in updates
    $U^{(t)}$ and update gradients $\bg^{(t)}$ and lengths
    $\bell^{(t)}$ for these edges. Find a circulation $\bDelta^{(t)}$
    that approximately solves the min-ratio problem as noted above.
    Update $\bf \assign \bf - \beta\bDelta^{(t)},$ where
    $\beta = \frac{\eta}{(\bg^{(t)})^{\top}\bDelta^{(t)}}.$
  \item $\textsc{Query}(e)$. Returns the value $\bf_e.$
  \item $\textsc{Detect}()$. For a fixed parameter $\eps$, where
    $\bDelta^{(t)}$ is the update vector at stage $t$, returns
    \begin{align}
      \label{eq:detect} S^{(t)} \defeq \left\{ e \in E : \bell_e \sum_{t' \in [\last^{(t)}_e+1,t]} |\bDelta_e^{(t')}| \ge \eps \right\}
    \end{align}
    where $\last^{(t)}_e$ is the last stage before $t$ that $e$ was
    returned by $\textsc{Detect}()$.
  \end{enumerate}
\end{definition}
Observe that the approximation ratio holds only with respect to the quality of the hidden stable witness circulation $\bc^{(t)},$ and not with respect to the best possible circulation. As a sanity check, if the data structure could find and return $\bc^{(t)}$ at each iteration, it would achieve a 1-approximation.
Thus, the data structure guarantee can be interpreted as efficiently representing and returning a cycle whose quality is within a $m^{o(1)}$ factor of $\bc^{(t)}$. Eventually, we will add $\bDelta^{(t)}$ to our flow efficiently by using link-cut trees to efficiently implement $\textsc{Update}(\cdot)$.

The following theorem encapsulates the IPM algorithm presented
in~\cite{chen2022maximum} and its interface with the dynamic min-ratio
cycle data structure.
\begin{theorem}[\cite{chen2022maximum}]\label{thm:ipmMain}
  Assume we have access to a $\kappa$-approximate dynamic min-ratio cycle with hidden stability data structure, for some
  $\kappa \in (0,1]$ as in \Cref{def:ipm:datastructure}.
  Then, there is a deterministic IPM-based algorithm that given a min-cost flow problem with integral costs and capacities bounded by
  $\exp((\log n)^{O(1)})$ in absolute value, solves $\tau = \O(m\kappa^{-2})$ min-ratio cycle instances, and returns a flow with cost within additive $1/2$ of optimal.
  These
  $\widetilde{O}(m\kappa^{-2})$ many min-ratio cycle instances satisfy
  the hidden stable $\alpha$-flow property for
  $\alpha = 1/\Theta(\log m)$.

  Over these min-ratio cycle instances, the total sizes of the updates
  is $\sum_{t \in \tau} |U^{(t)}| = \O(m\kappa^{-2}),$ and the
  algorithm invokes \textsc{Update}, \textsc{Query}, and
  \textsc{Detect} $\O(m\kappa^{-2})$ times. Furthermore, it is guaranteed
  that over all these instances, the total number of edges included in
  any of the \textsc{Detect} outputs is $\O(m\kappa^{-2}).$

  The algorithm runs in time $\O(m\kappa^{-2})$ plus the
  time taken by the data structure.  
\end{theorem}   
The above result can be generalized to arbitrary integer costs and
capacities at the cost of a $O(\log C \log mU)$ factor in the running
time by cost/capacity scaling \cite[Lemma C.1]{chen2022maximum}.

In the next section, we build our new deterministic data structure for approximate dynamic min-ratio cycles with hidden stability.

%% file: jtree.tex
\section{Data Structure Chain}
\label{sec:jtree}

This section is devoted towards building the core of the
data structure for approximately solving dynamic min-ratio cycle with hidden stability. The following is the main theorem we prove.

\begin{restatable}[Dynamic Min-Ratio Cycle with Hidden Stability]{theorem}{MMCHSF}
  \label{thm:MMCHiddenStableFlow}
  There is a deterministic data structure that $\kappa$-approximately
  solves the problem of dynamic min-ratio cycle with hidden stability for
  $\kappa = \exp(-O(\log^{17/18}m \cdot \log\log m)).$
    Over $\tau$ batches of updates
  $U^{(1)}, \ldots, U^{(\tau)},$ the algorithm runs in time
  $m^{o(1)}(m+ \sum_{t \in [\tau]} |U^{(t)}|)$.

  The data structure maintains a spanning tree $T \subseteq G^{(t)}$
  and returns a cycle $\bDelta$ represented as $m^{o(1)}$ paths on $T$
  (specified by their endpoints) and $m^{o(1)}$ explicitly given
  off-tree edges, and supports \textsc{Update} and
  \textsc{Query} operations in $m^{o(1)}$ amortized time. The running
  time of \textsc{Detect} is $m^{o(1)}|S^{(t)}|,$ where $S^{(t)}$ is the
  set of edges returned by \textsc{Detect}.
\end{restatable}
Combining this data structure with \cref{thm:ipmMain} and the flow rounding procedure in \cref{lemma:detflowrounding} shows \cref{thm:mincost}. Most of this section is devoted towards building the data structure in \cref{thm:MMCHiddenStableFlow}
and establishing how it finds approximately optimal
min-ratio cycles. 
\Cref{sec:lsf,sec:mwu,sec:core} focus on introducing the general layout of the data structure, and provides multiple definitions.
\Cref{sec:treeChainShifts} presents the dynamic data structure and states its properties.
Finally, we prove \cref{thm:MMCHiddenStableFlow} in \Cref{sec:datastructure-theorem} by
integrating link-cut trees to implement other required operations.

\paragraph{Comparison to \cite{chen2022maximum}.} 
This section is similar to \cite[Section 6]{chen2022maximum} in many ways and several parts are similar; here we describe the key differences (beyond tuning the presentation for this paper) and reasons for repeating some similar proofs. Sections \cref{sec:lsf} to \cref{sec:core} are largely the same, but with the slight difference that the dynamic graphs in these sections undergo edge insertions, deletions, and \emph{vertex splits}.
\cref{sec:treeChainShifts} deviates from \cite{chen2022maximum} by maintaining only one branch at each level, instead of $O(\log n)$ branches.
In \cite{chen2022maximum}, it was assumed that the tree-based data structures only underwent edge insertions and deletions, while the vertex splits were limited to the spanner.
Here, we allow all graphs to undergo vertex splits to ensure a tighter amortized runtime bound: over the course of $T$ updates, the total runtime and recourse of the data structure is $m^{o(1)}T$. This contrasts with \cite{chen2022maximum} where the total runtime and recourse was $m^{o(1)}(T + m/k)$. We cannot afford this because we visit all $\O(k)$ trees constructed in \cref{lemma:strMWU} over $O(m/k)$ iterations, and thus only stay with a single branch for about $m/k^2$ updates. 

\subsection{Dynamic Low-Stretch Forests (LSF)}
\label{sec:lsf}

As noted in the overview, the data structure is similar in
construction to the one from~\cite{chen2022maximum}. In order to
describe the data structure, we re-state several definitions verbatim from~\cite{chen2022maximum}.

\begin{table}[!ht]
\centering
\begin{tabularx}{\linewidth}{c*{2}{>{\RaggedRight\arraybackslash}m{0.7\linewidth}}}
    \toprule
    \textbf{Variable} & \textbf{Definition} \\
    \midrule
    $\bell^{(t)}, \bg^{(t)}$ & Lengths and gradients on a dynamic graph $G^{(t)}$ after stage $t.$ \\
    \hline
    $\bc^{(t)}, \bw^{(t)}$ & Hidden circulation and upper bounds with $|\bell^{(t)} \circ \bc^{(t)}| \le \bw^{(t)}.$ \\
    \hline
    $\mathcal{C}(G, F)$ & Core graph from a spanning forest $F$ \\
    \hline
    $\SS(G, F)$ & Sparsified core graph $\SS(G, F) \subseteq \mathcal{C}(G, F)$\\
    \hline
    $\cF^G = \{(T^G_j, F^G_j, \wstr^{j}_e)\}_{j=0}^{k-1}$
    & Collection of LSFs of $G$ (\cref{lemma:globalstretch}, \cref{lemma:strMWU}) \\
    \hline
    $d$ & Recursion levels \\
    \hline
    $k = m^{1/d}$ & Reduction factor \\
    \hline
    $\cG = \{G_0, \ldots, G_d\}$ & $d$-level tree chain (\cref{def:TreeChainShift}) \\
    \hline
    $\bran_i$ & Shift index for $G_i$ (\cref{def:TreeChainShift}) \\
    \hline
    $T^{\cG}$ & Spanning tree in $G_0$ corresponding to the tree chain $\cG$ (\cref{def:TreeChainShift}) \\
    \hline
    $\repT_i$ & Representative time stamp (\cref{def:repT}) \\
    \hline
    $\hstr_i$ & How much $\bw^{(\repT_i)}$ is stretched by $F^{G_i}_{\bran_i}$, the current LSF of $G_i$ (\cref{def:repT}) \\
    \bottomrule
  \end{tabularx}
\caption[Important definitions and notation to describe the data structure.]{Important definitions and notation to describe the data structure. In general, a $(t)$ superscript is the corresponding object at the $t$-th stage of a sequence of updates.}
\label{tab:glossaryDataStructureChain}
\end{table}

In the following subsections we describe the components of the data structure we maintain to show \cref{thm:MMCHiddenStableFlow}.
At a high level, our data structure maintains $d$ levels of graphs.
The graph size is reduced approximately by a factor of $k = m^{1/d}$ in each level.
The size reduction consists of two parts.
First, we reduce the number of vertices by maintaining a spanning forest $F$ of $\O(m/k)$ connected components and recurse on $G/F$, the graph obtained from $G$ by contracting each connected component of $F$ into a single vertex.
Next, we reduce the number of edges in $G/F$, which might have up to $m$ edges, to $m^{1+o(1)}/k$ via the dynamic sparsifier stated in \cref{thm:spanner}.
We start by defining a rooted spanning forest and its induced stretch.
\begin{definition}[Rooted Spanning Forest]
  \label{def:spanningforest}
  A \emph{rooted spanning forest} of a graph $G = (V, E)$ is a forest $F$ on $V$ such that each connected component of $F$ has a unique distinguished vertex known as the \emph{root}. We denote the root of the connected component of a vertex $v \in V$ as $\root^F_v$.
\end{definition}
\begin{definition}[Stretches of $F$]
  \label{def:stretchf}
  Given a rooted spanning forest $F$ of a graph $G = (V, E)$ with lengths $\bell \in \R_{>0}^E$, the stretch of an edge $e = (u, v) \in E$ is given by
  \begin{align*}
    \str^{F,\bell}_e \defeq
    \begin{cases}
      1 + \left\langle \bell, |\bp(F[u,v])|\right\rangle/\bell_e &~\text{ if } \root^F_u = \root^F_v \\
      1 + \left\langle \bell, |\bp(F[u, \root^F_u])| + |\bp(F[v, \root^F_v])| \right\rangle/\bell_e &~\text{ if } \root^F_u \neq \root^F_v,
    \end{cases}
  \end{align*}
  where $\bp(F[\cdot,\cdot]),$ as defined in \cref{sec:prelim}, maps a path to its signed indicator vector.
\end{definition}

When $F$ is a spanning tree \cref{def:stretchf} coincides with the definition of stretch for a LSST.
Otherwise, $\str^{F,\bell}_e$ measures how the concatenation of the two paths from endpoints to the roots stretches stretches $e.$

The goal of the remainder of this section is to give an algorithm to
maintain a \emph{Low Stretch Forest (LSF)} of a dynamic graph
$G.$
As a spanning forest decomposes a graph into vertex-disjoint connected subgraphs, a LSF consists of a spanning forest $F$ of low stretch.
The algorithm produces stretch upper bounds that hold throughout all updates to the graph, and the number of connected components of $F$ grows by $\O(1)$ per update in an amortized sense.
At a high level, for any edge insertion or deletion, the algorithm will force both endpoints to become roots of some component of $F$.
This way, any inserted edge will actually have stretch $1$ because both endpoints are roots.
Also, any deleted edge does not appear in $F$ and we can handle the deletion recursively on $G/F.$
\begin{lemma}[Dynamic Low Stretch Forest]
\label{lemma:globalstretch}
There is a deterministic algorithm with total runtime $\O(m)$ that on a graph $G = (V, E)$ with lengths $\bell \in \R^E_{>0}$, weights $\bv \in \R^E_{>0}$, and parameter $k > 0$, initializes a tree $T$ spanning $V$, and a rooted spanning forest $F \subseteq T$, an edge-disjoint partition $\cW$ of $F$ into $O(m/k)$ sub trees and stretch overestimates $\wstr_e$.
The algorithm maintains $F$, whose set of edges is decremental over time, against $\tau$ batches of updates to $G$, say $U^{(1)}, U^{(2)}, \dots, U^{(\tau)}$, such that $\wstr_e \defeq 1$ for any new edge $e$ added by edge insertions, and:
\begin{enumerate}
  \item $F$ has initially $O(m/k)$ connected components and $O(q \log^2 n)$ more after $q = \O(m)$ updates, i.e., $q \defeq \sum_{t=1}^{\tau} |U^{(i)}|.$ 
  \label{item:cccount}
  \item $\str^{F,\bell}_e \le \wstr_e \le O(k \log^6 n)$ for all $e \in E$ at all times, including inserted edges $e$. \label{item:stretchbound}
  \item $\sum_{e \in E^{(0)}} \bv_e \wstr_e \le O(\|\bv\|_1 \log^4 n)$, where $E^{(0)}$ is the initial edge set of $G$. \label{item:avgstretchbound}
  \item
  Initially, $\cW$ contains $O(m/k)$ subtrees.
  For any piece $W \in \cW, W \subseteq V$, $\Abs{\partial W} \le 1$ and $\vol_G(W \setminus R) \le O(k\log^2n)$ at all times, where $R \supseteq \partial \cW$ is the set of roots in $F$.
  Here, $\partial W$ denotes the set of \emph{boundary vertices} that are in multiple partition pieces.
    \label{item:degbound}
  \end{enumerate}
  We refer to the triple $(T, F, \wstr)$ as a Low-Stretch Forest (LSF) of G.
\end{lemma}

The lemma deviates from \cite[Lemma 6.5]{chen2022maximum} in the way they handle vertex splits.
Here, a vertex split adds $\O(1)$ roots to $F$ while the previous algorithm views it as a sequence of edge deletions and insertions.
The first property states that $F$ has $O(m/k)$ roots initially and each update, such as an edge update or a vertex split, adds $\O(1)$ roots to it on average.
For any edge $e$, the stretch overestimate $\wstr_e$ always stays the same and the average stretch, weighted by $\bv$, is always $\O(1).$
The final property is useful when applying the dynamic sparsifier of \cref{thm:spanner} to the contracted graph $G/F.$

We defer the proof of \cref{lemma:globalstretch} to \cref{app:globalstretch}.

\subsection{Worst-Case Average Stretch via Multiplicative Weights}
\label{sec:mwu}

By applying the multiplicative weights update procedure (MWU) on top of \cref{lemma:globalstretch}, we can build a distribution over partial spanning tree routings whose average stretch on every edge is $\O(1)$. This is very similar to MWUs done in works of \cite{R08,KLOS14} for building $\ell_\infty$ oblivious routings, and cut approximators~\cite{M10, S13}.
\begin{lemma}[MWU]
  \label{lemma:strMWU}
  There is a deterministic algorithm that when given a $m$-edge graph $G = (V, E)$ with lengths $\bell$ and a positive integer $k$, in in $\O(mk)$-time computes $k$ spanning trees, rooted spanning forests, and stretch overestimates $\{(T_j, F_j \subseteq T_j, \wstr^j_e)\}_{j=0}^{k-1}$ (\cref{lemma:globalstretch}) such that
  \begin{align}
    \label{eq:strMWU}
    \sum_{j = 0}^{k-1} \blambda_j \wstr^{j}_e \le O(\log^7 n)
    \text{ , }
    \forall e \in E,
  \end{align}
  where $\blambda \in \R_{>0}^{[k]}$ is the uniform distribution over the set $[t]$, i.e. $\blambda = \vec{1} / k$.
\end{lemma}
This lemma is nearly identical to \cite[Lemma 6.6]{chen2022maximum}, but we only build $k$ trees instead of $\O(k)$. We include the proof in \cref{app:mwu} for completeness.

The lemma guarantees that any given flow will be stretched by $\O(1)$ on average across the $k$ trees.
Thus, the flow is stretched by $\O(1)$ on at least one of the trees.
We will leverage this fact to design our data structure to prove \cref{thm:MMCHiddenStableFlow}.

\subsection{Sparsified Core Graphs and Path Embeddings}
\label{sec:core}
Given a rooted spanning forest $F$, we recursively process the graph $G/F$ where each connected component of $F$ is contracted to a single vertex represented by the root. We call this the \emph{core graph}, and define the lengths and gradients on it as follows. Below, we should think of $G$ as the result of updates to an earlier graph $G^{(0)}$, so $\wstr_e = 1$ for edges inserted to get from $G^{(0)}$ to $G$, as enforced in \cref{lemma:globalstretch}.

\begin{definition}[Core Graph, {\cite[Definition 6.7]{chen2022maximum}}]
  \label{def:coregraph}
  Consider a tree $T$ and a rooted spanning forest $E(F) \subseteq E(T)$ on a graph $G$ equipped with stretch overestimates $\wstr_e$ satisfying the guarantees of \cref{lemma:globalstretch}. We define the \emph{core graph} $\cC(G, F)$ as a graph with the same edge and vertex set as $G/F$. For $e = (u, v) \in E(G)$ with image $\hat{e} \in E(G/F)$ we define its length as $\bell^{\cC(G,F)}_{\hat{e}} \defeq \wstr_e\bell_e$ and gradient as $\bg^{\cC(G,F)}_{\hat{e}} \defeq \bg_e + \langle \bg, \bp(T[v, u]) \rangle$.
\end{definition}
In our application, we maintain $\cC(G, F)$ where $G$ is a dynamic graph and $F$ is a dynamic rooted spanning forest with a decremental edge set.
The tree $T$, the forest $F$, and the stretch overestimates $\{\wstr_e\}$ are initialized and maintained using \cref{lemma:globalstretch}.

Because the tree $T$ is static and the graph $G$ is dynamic, $T$ might not be a spanning tree of $G$ after some updates to $G.$
\cref{def:coregraph} handles the situation by allowing $T$ to be neither a spanning tree nor a subgraph of $G.$

Thus, for $e = (u, v) \in G$, $u$ and $v$ might not be connected in $T.$
In this case, we have $\bg^{\cC(G,F)}_{\hat{e}} = \bg_e.$
Moreover, the support of the gradient vector $\bg$ is $E(G) \cup E(T).$
This way, a deletion of some edge in $T$ from $G$ dose not affect the gradient $\bg^{\cC(G,F)}_{\hat{e}}.$

Note that the length and gradient of the image of any edge $e \in G$ in $\cC(G, F)$ does not change even if some edge in $F$ is removed, because they are defined with respect to the tree $T.$
This property will be useful in efficiently maintaining a sparsifier of the core graph using \cref{thm:spanner}, which reduces the number of edges from $m$ to $m^{1+o(1)}/k$.

\begin{definition}[Sparsified Core Graph, {\cite[Definition 6.9]{chen2022maximum}}]
  \label{def:sparsecore}
  Given a graph $G$, forest $F$, and parameter $k$, define a $(\gamma_{\ell}, \gamma_{c})$-\emph{sparsified core graph with embedding} as a subgraph $\SS(G, F) \subseteq \mathcal{C}(G, F)$ and embedding $\Pi_{\mathcal{C}(G, F) \to \SS(G, F)}$ satisfying
  \begin{enumerate}
  \item For any $\hat{e} \in E(\mathcal{C}(G, F))$, all edges $\hat{e}' \in \Pi_{\mathcal{C}(G, F) \to \SS(G, F)}(\hat{e})$ satisfy $\bell_{\hat{e}}^{\mathcal{C}(G,F)} \approx_2 \bell_{\hat{e}'}^{\mathcal{C}(G,F)}$. \label{item:samelen}
  \item $\length(\Pi_{\mathcal{C}(G, F) \to \SS(G, F)}) \le \gamma_l$ and $\econg(\Pi_{\mathcal{C}(G, F) \to \SS(G, F)}) \le k\gamma_c$.
  \item $\SS(G, F)$ has at most $m\gamma_{\ell}/k$ edges.
  \item The lengths and gradients of edges in $\SS(G, F)$ are the same as in $\mathcal{C}(G, F)$ (\cref{def:coregraph}).
  \end{enumerate}
\end{definition}

\subsection{Shifted Tree Chains}
\label{sec:treeChainShifts}

In this section, we introduce the notion of \emph{Shifted Tree Chains}.
Our data structure has $d$ levels with reduction factor $k \approx m^{1 / d}$.
The graph at $i$-th level has a size roughly $m / k^i$ and $k$ forests from \cref{lemma:strMWU}.
We recursively build the chain on one of the forests and keep the rest.
In the dynamic setting, we support operations to \emph{shift} a level, that is, recursively rebuild the chain on the next forest.
This is used to handle the case that the current tree chain cannot output any cycle of a small enough ratio.

\begin{definition}[Shifted Tree Chain]
\label{def:TreeChainShift}
For a graph $G$, recursion level $d$, and reduction factor $k =
m^{1/d}$, a \emph{$d$-level tree chain} is a collection of graphs $\cG
= \{G_0 = G, G_1,\ldots, G_d\}.$
For each $i < d$, we have the following:
\begin{enumerate}
\item A collection of \emph{low stretch forests} $\cF^{G_i} = \{(T_j \subseteq G_i, F_j \subseteq T_j, \wstr^j_e)\}_{j=0}^{k-1}$ that satisfies conditions in \cref{lemma:strMWU},
\item A shift index $\bran_i \in \{0, 1, \ldots, k - 1\}$ which is initially $0$; 
\item For each $ F \in \cF^{G_i}$, a $(\gamma_{\ell}, \gamma_c)$-sparsified core graph $\cS(G_i, F)$ and embedding $\Pi_{(\cC(G_i, F) \to \cS(G_i, F))}$;  
\item and, we recursively define $G_{i+1} = \cS(G_i, F)$, the sparsified core graph w.r.t. the current LSF $F = F^{G_i}_{\bran_i}.$
\end{enumerate}
Finally, for the last level graph $G_d$, we maintain a collection of
$\O(|E_{G_d}|)$ low stretch trees $T$ using our MWU procedure (\Cref{lemma:strMWU}) with $k=\O(|E_{G_d}|).$

The tree chain $\cG$ naturally corresponds to a spanning tree $T^{\cG}$ of $G$, which is the union of pre-images of the forests $F^{G_0}_{\bran_0}, F^{G_1}_{\bran_1}, \ldots, F^{G_d}_{\bran_d}.$
\end{definition}

Compared to the branching tree chain used in \cite{chen2022maximum}, our data structure maintains one graph at each level instead of $\O(1)^i.$

We can dynamically maintain a tree chain such that we re-initialize $\cF^{G_i}$ and $G_{i+1}$ from $G_i$ every approximately $m/k^i$ updates.
Between re-initializations, the forests in the collection of LSFs $\{F_i\}$ are decremental as guaranteed in \cref{lemma:globalstretch}.

In addition to edge updates, a $d$-level tree chain is subject to (1) rebuild at level $i$ and (2) shift at level $i.$
\begin{definition}[Rebuild at Level $i$]
\label{def:rebuild}
Given a \emph{$d$-level tree chain} $\cG = \{G_0 = G, G_1,\ldots, G_d\}$, $\textsc{Rebuild}(i)$ re-initializes graphs $G_i, G_{i+1}, \ldots, G_d.$
\end{definition}
\begin{definition}[Shift at Level $i$]
\label{def:shift}
Given a \emph{$d$-level tree chain} $\cG = \{G_0 = G, G_1,\ldots, G_d\}$, $\textsc{Shift}(i)$ increments $\bran_i \gets (\bran_i + 1) \mod{k}.$
And it re-initializes every graphs of $G_{i+1}, \ldots, G_d.$
\end{definition}

A shift at level $i$ does not change $\cF^{G_i},$ the collection of low stretch forests at level $i,$ but only increases the branching index by $1$ circularly.
Shifting a level $i$ with $\bran_i = k - 1$ resets $\bran_i = 0$ while the set of low stretch forests remains the same.
As we will show in \cref{sec:cycleshift}, one of the forests preserves a cycle of a small ratio.
The circular behavior of shifts ensures that we will reach such a forest using at most $k$ shifts.
An alternative would be to re-initialize the set of forests whenever $\bran_i$ hits $0.$
In this setting, we would require as many as $2k - 1$ shifts to reach the forest preserving a small-ratio cycle and the analysis remains roughly the same.

In the rest of the section, we show the following lemma, which is a data structure weaker than \cref{thm:MMCHiddenStableFlow} because the output circulation has a larger ratio than desired.
In \cref{sec:routing}, we further analyze the cycle maintained by the data structure.
We later boost this to an algorithm for \cref{thm:MMCHiddenStableFlow} by solving a \emph{shift-and-rebuild game} in \cref{sec:rebuilding}.

\begin{restatable}[Dynamic Tree Chain]{lemma}{dynTreeChain}
\label{lemma:hintedTreeChain}
\cref{algo:hintedTreeChain} takes as input a parameter $d$, a dynamic graph $G^{(t)}$ undergoes $\tau$ batches of updates $U^{(1)}, \dots, U^{(\tau)}$ with hidden stability (\cref{def:hiddenStableFlowChasing}).

The algorithm explicitly maintains a tree chain $\cG$ (\cref{def:TreeChainShift}) and $T^{\cG}$, the spanning tree corresponding to the tree chain.
At stage $t$, the algorithm outputs a circulation $\bDelta$ represented by $m^{o(1)}$ off-tree edges and tree paths w.r.t. $T^{\cG}$.
The output circulation $\bDelta$ satisfies $\mB^\top \bDelta = 0$ and 
\begin{align*}
    \frac{\abs{\l\bg^{(t)}, \bDelta\r}}{\norm{\diag(\bell^{(t)}) \bDelta}_1} \ge \frac{1}{\O(k)} \frac{\abs{\l\bg^{(t)}, \bc^{(t)}\r}}{\sum_{i=0}^d \|\bw^{(t), G_i}\|_1}
\end{align*}
where $\bw^{(t), G_i}$ is the width at stage $t$ passed down to $G_i$, the level $i$ graph in the tree chain (\cref{def:passcore,def:passsparsecore}).
The algorithm also outputs the changes to $T^{\cG}$ as an explicit list of edge insertions and deletions.

In addition, throughout the algorithm, suppose there are $s_i$
invocations of $\textsc{Shift}$ and $r_i$ invocations of
$\textsc{Rebuild}$ at each level $i$.
Note that we only allow $\textsc{Rebuild}$ to be called through
\cref{line:periodRebuild} in 
$\textsc{Update}$.
The algorithm is deterministic and runs in time 
\begin{align*}
m^{1/d}\O(\gamma_{\ell}\gamma_r)^{O(d)}\left(m + \sum_{i=0}^{d} (s_i + r_i) \cdot m^{1 - i/d}\right)
\end{align*}
The same quantity also bounds the total number of edge updates to the spanning tree $T^{\cG}.$
\end{restatable}

\cref{algo:hintedTreeChain} initializes a tree chain as in \cref{def:TreeChainShift}.
For graph $G_i$ at level $i$, it maintains a collection of forests, trees, and sparsified core graph using the dynamic data structure from \cref{lemma:hintedsparsecore}, which will be presented later in the section.

\begin{algorithm}
  \caption{Dynamically maintains a tree chain (\cref{def:TreeChainShift}).
 \label{algo:hintedTreeChain}}
  \SetKwProg{Globals}{global variables}{}{}
  \SetKwProg{Proc}{procedure}{}{}
  \Globals{}{
    $d$: number of levels in the maintained tree chain. \\
    $k \assign m^{1 / d}$: reduction factor used in
    \cref{lemma:globalstretch}. \\
    $\Psi \gets \log^{O(1)} m $ such that by \Cref{def:hiddenStableFlowChasing} we have
    $\log \| \bw^{(t)}\|_1 \in (-\Psi,\Psi)$.\\
    $\cG = \{G_0, G_1, \ldots, G_d\}$: the maintained tree chain.\\
    $\{F^{G_i}_0, F^{G_i}_1, \ldots, F^{G_i}_{k-1}\}$ for each level $i$: collection of LSFs of each $G_i.$ \\
    $\bran_0, \bran_1, \ldots, \bran_{d-1}$ : the branching index of each $G_i.$ \\
    $\cA^{(\mathrm{SparseCore})}$: the dynamic sparsified core graph algorithm
    (\cref{lemma:hintedsparsecore}).  \\
    $\passes_i$, a variable for each level $i$.
  }
  \Proc{$\Initialize(G^{(0)}, \bell,\bg)$}{
    $G_0 \gets G^{(0)}$ and  $\passes_i \gets 0$ for all $i \in \{ 0, \ldots, d \}$\;
    $\Rebuild(0)$ 
  }
  \Proc{$\Rebuild(i_0)$}{
    \For{$i = i_0, \dots, d - 1$}{
      $\bran_i \assign 0.$  \\
      $\Set{\cS(G_i, F^{G_i}_{j})}{j = 0, 1, \ldots, k-1} \assign
      \cA^{(\mathrm{SparseCore})}_{G_i}.\textsc{Initialize}(G_i, \bell_{G_i},\bg_{G_i})$ \\
      $G_{i+1} \assign \cS(G_i, F^{G_i}_{\bran_i})$ \\
    }
  }
  \Proc{$\Shift(i_0)$}{
    $\bran_i \assign (\bran_i + 1) \mod{k}.$ \\
    $G_{i+1} \assign \cS(G_i, F^{G_i}_{\bran_i})$ \\
    \For{$i = i_0+1, \dots, d - 1$}{ \label{line:shiftRebuild}
      $\bran_i \assign 0.$ \\
      $\Set{\cS(G_i, F^{G_i}_{j})}{j = 0, 1, \ldots, k-1} \assign
      \cA^{(\mathrm{SparseCore})}_{G_i}.\textsc{Initialize}(G_i, \bell_{G_i},\bg_{G_i})$ \\
      $G_{i+1} \assign \cS(G_i, F^{G_i}_{\bran_i})$\\
    }
  }    
  \Proc{$\Update(U^{(t)}, \bg^{(t)}, \bell^{(t)})$}{
    $U^{(t)}_{G^{(t)}} \assign U^{(t)}$  \\
    \For{$i = 0, \dots, d - 1$}{
        \lIf{The total number of updates to $G_i$ since its last rebuild is more than $m (\gamma_{\ell} / k)^{i+1} / \log^2 n$}{
            $\Rebuild(i)$ \label{line:periodRebuild}
      }
      $\Set{U^{(t)}_{\cS(G_i, F_j)}}{j = 0, 1, \ldots, k-1} \assign
      \cA^{(\mathrm{SparseCore})}_{G_i}.\textsc{Update}(G_i, U^{(t)}_{G_i})$
    }
  }
  \Proc{$\FindCycle()$}{
    Return the best \emph{fundamental spanner cycle/level-$d$ tree cycle} $\bDelta$ (mapped
    back to $G_0$) with
    largest ratio $  \frac{\abs{\l\bg^{(t)}, \bDelta\r}}{\norm{\bell^{(t)} \circ \bDelta}_1}.$
    See Proof of \Cref{lemma:hintedTreeChain} for details.
  }
\end{algorithm}
Towards proving \cref{lemma:hintedTreeChain}, we first discuss and define how to pass witness circulations and upper bounds $\bc$ and $\bw$ through the tree chain $\cG.$

We first describe how to pass $\bc, \bw$ from $G$ to a core graph $\cC(G, F)$ (\cref{def:coregraph}).
\begin{definition}[Passing $\bc, \bw$ to core graph]
\label{def:passcore}
Given a graph $G = (V, E)$ with a tree $T$, a rooted spanning forest $E(F) \subseteq E(T)$, and a stretch overestimates $\wstr_e$ as in \cref{lemma:globalstretch}, circulation $\bc \in \R^E$, and length upper bounds $\bw \in \R^E_{>0}$, we define vectors $\bc^{\cC(G, F)} \in \R^{E(\cC(G, F))}$ and $\bw^{\cC(G, F)} \in \R^{E(\cC(G, F))}_{>0}$ as follows.
For $\hat{e} \in E(\cC(G, F))$ with preimage $e \in E$, define $\bc^{\cC(G, F)}_{\hat{e}} \defeq \bc_e$ and $\bw^{\cC(G, F)}_{\hat{e}} \defeq \wstr_e\bw_e$.
\end{definition}
We verify that $\bc^{\cC(G, F)}$ is a circulation on $\cC(G, F)$ and that $\bw^{\cC(G, F)}$ are length upper bounds.
\begin{lemma}[Validity of \cref{def:passcore}, Lemma 7.4 of \cite{chen2022maximum}]
  \label{lemma:passcore}
  Let $\bc,\bw$ be a valid pair (\cref{def:validpair}) on a graph $G$ with lengths $\bell$.
  As defined in \cref{def:passcore},
  $\bc^{\cC(G, F)}, \bw^{\cC(G, F)}$ are a valid pair on $\cC(G, F)$ with lengths $\bell^{\cC(G,F)}$ (\cref{def:coregraph}),
  and
  \[ \left\|\bw^{\cC(G, F)}\right\|_1 \le \sum_{e \in E(G)} \wstr_e \bw_e. \]
\end{lemma}

We state an algorithm that takes a dynamic graph $G^{(t)}$ with hidden stability and produces a dynamic core graph.
Below, we let $\bc^{(t),\cC(G,F)},\bw^{(t),\cC(G,F)}$ denote the result of using \cref{def:passcore} for $\bc = \bc^{(t)}$ and $\bw = \bw^{(t)}$, and similar definitions for $\bg^{(t),\cC(G,F)},\bell^{(t),\cC(G,F)}$ used later in the section.
\begin{lemma}[Dynamic Core Graphs]
\label{lemma:hintedcore}
There is a deterministic algorithm that takes as input a parameter $k$, a dynamic graph $G^{(t)}$ undergoes $\tau$ batches of updates $U^{(1)}, \dots U^{(\tau)}$ that satisfies $\sum_{t=1}^{\tau} |U^{(t)}| \le m/(k\log^2 n)$ and has hidden stability (\cref{def:hiddenStableFlowChasing}).
  
For each $j = 0, 1, \ldots, k-1$, the algorithm maintains a static tree $T_j$, a 
decremental rooted forest $F^{(t)}_j$ with $O(m/k)$ components satisfying the conditions of \cref{lemma:globalstretch}, and a core graph $\cC(G^{(t)}, F^{(t)}_j)$.
The algorithm outputs update batches $U^{(t)}_j$ that produce $\cC(G^{(t)}, F^{(t)}_j)$ from $\cC(G^{(t-1)}, F^{(t-1)}_j)$ such that $\sum_{t' \le t} |U^{(t')}_j| = O\left(\sum_{t' \le t} |U^{(t')}| \cdot \log^2 n\right).$

The algorithm runs in $\O(mk)$-time.
\end{lemma}
\begin{proof}
This follows almost directly from Lemma 7.5 in \cite{chen2022maximum}.
Here we maintain core graphs of every low stretch forest $F_j$ using \cref{lemma:globalstretch}, which only add $O(\log^2 n)$ roots to $F_j$ per vertex split.
Adding one root to $F_j$ splits some component of $F_j$ into two and this splits some vertex in the core graph.
Thus, each update, which may be an edge update or vertex split, to $G$ corresponds to $O(1)$ edge updates and $O(\log^2 n)$ vertex splits to the core graph.
Thus, the total number of updates to the core graph $\cC(G, F_j)$ is $\sum_{t' \le t} |U^{(t')}_j| = O\left(\sum_{t' \le t} |U^{(t')}| \cdot \log^2 n\right).$

The runtime is $\O((m + Q) k)$ for $Q = \sum_t \Enc(U^{(t)}).$
We concludes the runtime analysis using \cref{lemma:encodingSize} to bound $Q = \O(m + \sum_{t=1}^{\tau} |U^{(t)}|) = \O(m).$
\end{proof}

We describe how to pass $\bc^{\cC(G, F)}, \bw^{\cC(G, F)}$ on a core graph to a sparsified core graph $\cS(G, F)$.
\begin{definition}[Passing $\bc,\bw$ to sparsified core graph]
  \label{def:passsparsecore}
  Consider a graph $G$ with spanning forest $F$, and circulation $\bc^{\cC(G, F)} \in \R^{E(\cC(G, F))}$ and upper bound $\bw^{\cC(G, F)} \in \R^{E(\cC(G, F))}_{>0}$, and embedding $\Pi_{\cC(G, F) \to \cS(G, F)}$ for a $(\gamma_c,\gamma_l)$-sparsified core graph $\SS(G, F) \subseteq \mathcal{C}(G, F)$. Define
  \begin{align} 
    \bc^{\SS(G, F)} &= \sum_{\hat{e} \in E(\mathcal{C}(G, F))} \bc^{\mathcal{C}(G, F)}_{\hat{e}}\bPi_{\mathcal{C}(G, F) \to \SS(G, F)}(\hat{e}) \label{eq:bcsparsecore} \\
    \bw^{\SS(G, F)} &= 2\sum_{\hat{e} \in E(\mathcal{C}(G, F))} \bw^{\mathcal{C}(G, F)}_{\hat{e}}\left|\bPi_{\mathcal{C}(G, F) \to \SS(G, F)}(\hat{e})\right| \label{eq:bwsparsecore}.
  \end{align}
\end{definition}
We check that $\bc^{\SS(G, F)}$ is a circulation on $\SS(G, F)$ and $\bw^{\SS(G, F)}$ are length upper bounds.
\begin{lemma}[Validity of \cref{def:passsparsecore}, {\cite[Lemma 7.7]{chen2022maximum}}]
  \label{lemma:passsparsecore}
  Let $\bc^{\mathcal{C}(G,F)}, \bw^{\mathcal{C}(G,F)}$ be a valid pair on graph $\mathcal{C}(G,F)$ with lengths $\bell^{\mathcal{C}(G,F)}$. As defined in \cref{def:passsparsecore}, $\bc^{\SS(G, F)}, \bw^{\SS(G,F)}$ is a valid pair on $\mathcal{S}(G,F)$ with lengths $\bell^{\mathcal{S}(G,F)}$ (\cref{def:sparsecore}). Also,
  \[ \|\bw^{\mathcal{C}(G,F)}\|_1 \le \|\bw^{\SS(G, F)}\|_1 \le O(\gamma_l)\|\bw^{\mathcal{C}(G,F)}\|_1. \]
\end{lemma}

We can maintain the sparsified core graph of a dynamic graph $G^{(t)}$ with hidden stability.
In particular, the total size of updates to $\cS(G^{(t)}, F^{(t)})$ is comparable to the one for $G^{(t)}.$
This is the building block for dynamically maintaining the tree chain (\cref{def:TreeChainShift}).
\begin{lemma}[Dynamic Sparsified Core Graphs]
  \label{lemma:hintedsparsecore}
  There is an algorithm $\cA^{(\mathrm{SparseCore})}$ takes as input a parameter $k$, a dynamic graph $G^{(t)}$ undergoes $\tau$ batches of updates $U^{(1)}, \dots, U^{(\tau)}$ with hidden stability (\cref{def:hiddenStableFlowChasing}) and $\sum_{t=1}^{\tau} |U^{(t)}| \le m/(k\log^2 n).$

Upon initialization via a call $\cA^{(\mathrm{SparseCore})}_{G}.\textsc{Initialize}(G,
\bell,\bg)$, the data structure $\cA^{(\mathrm{SparseCore})}$ maintains for each $j \in \{0, 1, \ldots, k-1\}$, a decremental forest $F^{(t)}_j$, a static tree $T_j$ satisfying the conditions of \cref{lemma:globalstretch}, and a $(\gamma_{\ell},\gamma_c)$-sparsified core graph $\cS(G^{(t)}, F^{(t)}_j)$ for parameters $\gamma_c = \gamma_l = \exp(O(\log^{8/9}m\log\log m))$ with embedding $\Pi_{\mathcal{C}(G^{(t)}, F^{(t)}_j)\to\SS(G^{(t)}, F^{(t)}_j)}$ and supports the operation
  \begin{align*}
      \cA^{(\mathrm{SparseCore})}.\textsc{Update}(G^{(t-1)}, U^{(t-1)})
  \end{align*}
  which outputs, for each $j$, an update batch $U_{\SS, j}^{(t)}$ that produces $\SS(G^{(t)}, F^{(t)}_j)$ from $\SS(G^{(t-1)}, F^{(t-1)}_j).$
  \begin{enumerate}
      \item \label{item:SCGLowRec}
      \underline{Sparsified Core Graphs have Low Recourse:}
      For each $j$, the update batches $\{U_{\SS, j}^{(t)}\}_t$ to $\cS(G, F_j)$ output by $\cA^{(\mathrm{SparseCore})}$ satisfies
      \begin{align*}
          \sum_{t' \le t} \left|U_{\SS, j}^{(t')}\right| &= \gamma_r \cdot \sum_{t' \le t} \left|U^{(t')}\right|, \text{ and} \\
          \sum_{t' \le t} \Enc(U_{\SS, j}^{(t')}) &= \gamma_r \cdot \left(\frac{m}{k} + \sum_{t' \le t} \left|U^{(t')}\right|\right)
      \end{align*}
      for some $\gamma_r = \exp(O(\log^{8/9}m\log\log m))$, and
      
      \item \label{item:SCGHSFC} \underline{Sparsified Core Graphs undergo Updates with Hidden Stability:}
      for each $j$, the update batches $U_{\SS, j}^{(t)}$ to the sparsified core graph along with the associated gradients $\bg^{(t),\SS(G^{(t)} F^{(t)}_j)}$, and lengths $\bell^{(t),\SS(G^{(t)}, F^{(t)}_j)}$ as defined in \Cref{def:sparsecore} satisfy the hidden stable-flow property (see \cref{def:hiddenStableFlowChasing}) with the hidden circulation $\bc^{(t),\SS(G^{(t)}, F^{(t)}_j)}$, and width $\bw^{(t),\SS(G^{(t)}, F^{(t)}_j)}$ as defined in \Cref{def:passsparsecore}.
  \end{enumerate}
The algorithm runs in total time $\O(m k \cdot \gamma_r)$.
\end{lemma}
\begin{proof}
Except Item~\cref{item:SCGHSFC}, the lemma follows by using \cref{lemma:hintedcore} to maintain all $k$ core graphs $\cC(G^{(t)}, F_j^{(t)})$ and \cref{thm:spanner} to maintain sparsified core graphs $\SS(G^{(t)}, F^{(t)}_j).$
In particular, for each $j$, we add $D^{(t)}_j$, the re-embeded set of edges in $\cS(G^{(t)}, F^{(t)}_j)$, as explicit edges updates to $U_{\cS, j}^{(t)}.$

We now show Item~\cref{item:SCGHSFC}.
$\bc^{(t),\SS(G^{(t)}, F^{(t)}_j)}$ and $\bw^{(t),\cS(G^{(t)}, F^{(t)}_j)}$ form a valid pair by \cref{lemma:passsparsecore}. Therefore, items \ref{item:circulation} and \ref{item:width} of \cref{def:hiddenStableFlowChasing} are satisfied.

Next, we prove item \ref{item:widthstable} of \cref{def:hiddenStableFlowChasing}.
At any stage $t \in [\tau]$ and any edge $e \in \cS(G^{(t)}, F_j^{(t)})$ for some $j$, suppose $e$ was not explicitly inserted after some earlier stage $t'$, i.e. $e \in \cS(G^{(t')}, F_j^{(t')})$ for some $t' < t.$
However, between stage $t'$ and $t$, endpoints of $e$ might change due to vertex splits.
But the insertion of $e$ is not included in any of $U^{(s)}_{\cS, j}, s \in (t', t].$
Thus, we have $(\Pi^{(t)}_j)^{-1}(e) \subseteq (\Pi^{(t')}_j)^{-1}(e)$ otherwise insertion of $e$ is included in some $U^{(s)}_{\cS, j}, s \in (t', t]$ due to the definition of re-embedded set (\cref{prop:lowRecourseSpanner} of \cref{thm:spanner}).

For any edge $e' \in (\Pi^{(t)}_j)^{-1}(e)$, it exists in the core graph at both stage $t$ and $t'$, i.e. $e' \in \cC(G^{(t)}, F_j^{(t)})$ and $\cC(G^{(t')}, F_j^{(t')}).$
Let ${e'}^G$ be its pre-image in $G.$
${e'}^G$ also exists in $G$ at both stage $t$ and $t'$.
Since $G$ is updates with hidden stability, by item \ref{item:widthstable} of \cref{def:hiddenStableFlowChasing} we have
\begin{align*}
    \bw^{(t), G^{(t)}}_{{e'}^G} \le 2 \cdot \bw^{(t'), G^{(t')}}_{{e'}^G}.
\end{align*}
\cref{def:passcore} and the immutable nature of $\wstr$ from \cref{lemma:globalstretch} yields
\begin{align}
\label{eq:coreWidthStable}
    \bw^{(t), \cC(G^{(t)}, F^{(t)}_j)}_{e'} = \wstr^{T_j, \bell}_{{e'}^G}  \bw^{(t), G^{(t)}}_{{e'}^G} \le 2 \cdot \wstr^{T_j, \bell}_{{e'}^G}  \bw^{(t'), G^{(t')}}_{{e'}^G} = 2 \cdot \bw^{(t'), \cC(G^{(t')}, F^{(t')}_j)}_{e'}.
\end{align}

Combining with the fact that $(\Pi^{(t)}_j)^{-1}(e) \subseteq (\Pi^{(t')}_j)^{-1}(e)$ and \cref{def:passsparsecore} yields the following and proves item \ref{item:widthstable} of \cref{def:hiddenStableFlowChasing}:
\begin{align*}
    \bw^{(t),\cS(G^{(t)}, F^{(t)}_j)}_{e} 
    &= 2 \cdot \sum_{e' \in \left(\Pi^{(t)}_j\right)^{-1}(e)} \bw^{(t), \cC(G^{(t)}, F^{(t)}_j)}_{e'} \\
    &\le 2 \cdot 2 \cdot \sum_{e' \in \left(\Pi^{(t)}_j\right)^{-1}(e)} \bw^{(t'), \cC(G^{(t')}, F^{(t')}_i)}_{e'} \\
    &\le 2 \cdot 2 \cdot \sum_{e' \in \left(\Pi^{(t')}_j\right)^{-1}(e)} \bw^{(t'), \cC(G^{(t')}, F^{(t')}_j)}_{e} = 2 \cdot \bw^{(t'),\cS(G^{(t')}, F^{(t')}_i)}_{e}.
\end{align*}

  Item~\ref{item:quasipoly} follows directly from the definition of $\bell^{(t),\cS(G^{(t)}, F^{(t)}_j)}$ and $\bw^{(t),\cS(G^{(t)}, F^{(t)}_j)}.$
\end{proof}

Given \cref{lemma:hintedsparsecore}, we are ready to prove \cref{lemma:hintedTreeChain} with the data structure described in \cref{algo:hintedTreeChain}.
Recall that the algorithm initializes a shifted tree chain (\cref{def:TreeChainShift}) and maintains a collection of low stretch forests and sparsified core graphs using \cref{lemma:hintedsparsecore} for graph $G_i$ at each level $i.$

However, the data structure of \cref{lemma:hintedsparsecore} can only take up to $m /(k\log^2 n)$ updates if the input graph has at most $m$ edges at all times.
This forces us to rebuild the data structure every once in a while.
In particular, we re-initialize everything at every level $i \ge i_0$ if any of the data structures of \cref{lemma:hintedsparsecore} on some level $i_0$ graph $G_{i_0}$ has accumulated too many updates (approximately $m/k^{i_0}$).
We will show that the cost of re-initializing amortizes well across dynamic updates.

\begin{proof}[Proof of \cref{lemma:hintedTreeChain}]
  At any stage $t$ and level $i > 0$, graph $G^{(t)}_i$ has at most $m \gamma_{\ell}^{i-1} / k^{i}$ vertices and $m (\gamma_{\ell} / k)^{i}$ edges due to \cref{lemma:hintedsparsecore}.

Ignoring the last level $d,$ we can count the entire cost of all data structure operations toward
the point where each sparse core graph is initialized. 
By \cref{lemma:hintedsparsecore}, the runtime cost of
$\cA^{(SparseCore)}$ for $G_i$ until its next initialization is
$\O(m(\gamma_{\ell}/k)^ik\gamma_r).$
  When we initialize all levels $i \geq h$, the running time cost is $\sum_{i =
    h}^d \O(m(\gamma_{\ell}/k)^ik\gamma_r) =
  \O(m(\gamma_{\ell}/k)^hk\gamma_r)$, as the costs decay
  geometrically.
  The very first initialization cost can thus be bounded by $\O(m
  k\gamma_r)$ by considering $h = 0$.
  All remaining initializations occur when (1) a level $i_0$ is shifted and all levels $i > i_0$
  are re-initialized or (2) when a level $i_0$ is rebuilt and all levels $i \geq i_0$
  are re-initialized.
We can bound these cost from (1) and (2) by $\sum_i s_i
\O(m(\gamma_{\ell}/k)^{i+1} k\gamma_r)$ and $\sum_i r_i
\O(m(\gamma_{\ell}/k)^{i} k\gamma_r)$ respectively.
For the last level, note that the initialization running time is
somewhat larger $\O(\gamma_{\ell}^{2d}),$ but this does not affect the
overall asymptotics of the above sum as long as $d$ is not too large. 
The overall runtime bound follows because $k = m^{1/d}$.

  Finally, we describe how to maintain and return a good enough
  circulation $\bDelta.$ We have two kinds of candidate circulations:
  \begin{enumerate}
      \item   fundamental spanner cycles: for every level $i,$ and every
  edge
  $e \in \mathcal{C}(G_i, F_{\bran_i}^{G_i}) \setminus \SS(G_i,
  F_{\bran_i}^{G_i}),$ we consider the cycle formed by $e$ and its
  spanner embedding path
  $\Pi_{ \mathcal{C}(G_i, F_{\bran_i}^{G_i})\to\SS(G_i,
    F_{\bran_i}^{G_i})}.$ This is a cycle in the core graph
  $ \mathcal{C}(G_i, F_{\bran_i}^{G_i})$ at level $i.$ 
  \item fundamental tree cycles at level $d,$ i.e., for every edge $e \in G_d$ and each a low-stretch tree $T$ for $G_d$ in the data structure, we consider the unique cycle obtained by adding $e$ to $T.$
  \end{enumerate}
  We use tree
  paths and spanner embeddings to map these cycles back to
  $G_0 = G^{(t)}.$
  We now claim that one of these cycles $\bDelta$ must satisfy
  \begin{align*}
    \frac{\abs{\l\bg^{(t)}, \bDelta\r}}{\norm{\bell^{(t)} \circ \bDelta}_1} \ge \frac{1}{\O(k)} \frac{\abs{\l\bg^{(t)}, \bc^{(t)}\r}}{\sum_{i=0}^d \norm{\bw^{(t), G_i}}_1}
  \end{align*}
  This follows by \cref{lemma:goodenough}, which is shown in \cite[Lemma 7.15-7.17]{chen2022maximum} -- we include a self-contained statement in the appendix for completeness.
  Since all embeddings are known explicitly, the algorithm can compute
  the quality of a fundamental spanner cycle when the embedding of an
  edge changes. Thus, the algorithm simply tracks the quality of all
  fundamental spanner cycles using a heap and returns the best among them.
\end{proof}

\section{Analyzing the Cycle Quality with Shifts and Rebuilds}
\label{sec:cycleshift}
The goal of this section is to complete the proof of \cref{thm:MMCHiddenStableFlow}. To do this, we start in \cref{sec:routing} by analyzing the quality of the circulation output by \cref{lemma:hintedTreeChain} at each stage of the dynamic updates (\cref{lemma:hintedTreeChainWidthBound}). This involves defining \emph{representative time steps} in \cref{def:repT}, which intuitively represent the points in time of the data structure that we measure our stretch with respect to. Because \cref{lemma:hintedTreeChainWidthBound} is a weaker version of our desired \cref{thm:MMCHiddenStableFlow}, in \cref{sec:datastructure-theorem} we describe a strategy for applying shifts and rebuilds (\cref{lem:AlgoIsGameInstance}) to our cycle quality bound to show \cref{thm:MMCHiddenStableFlow}.

\subsection{Cycle Qualities}
\label{sec:routing}
An important concept for analyzing cycle quality is our notion of a \emph{representative time}. 
A representative time stamp is used to record the quality of the witness approximation by our data structure at this past point in time.
This and the shift operation are the main differences in the data structure between this paper and \cite{chen2022maximum}.

\begin{definition}[Representative timesteps]
\label{def:repT}
Consider the setting of \cref{lemma:hintedTreeChain} at some stage t, and consider the state of the data structure, right after any sequence of completed procedure calls.
A set of timestamps associated with each level $\{\repT_i\}_{i=0}^{d}$ is a set of \emph{representative time stamps} if, at any level $i$, $\repT_i$ is between $t$ and level $i$'s previous re-initialization time and is after $\repT_{i-1}$.
We also define, at any level $i$,
for each branch $j$,
we define the stretch w.r.t. $\{\repT_i\}_i$ to be the value $\hstr_{i,j}$ such that
\begin{align*}
    \sum_{e \in G_i} \wstr^{F^{G_i}_{j}}_{e} \bw^{(\repT_i),
  G_i}_e = \hstr_{i,j} \norm{\bw^{(\repT_i), G_i}}_1
  .
\end{align*}
where  $F^{G_i}_{j} \subseteq G_i$ are the spanning forests associated with $G_i$.
We define the current stretch to be $\hstr_{i} \defeq \hstr_{i,\bran_i}$.
\end{definition}
Notice that for level $d$, we have $\repT_d = t$ and $\hstr_d = \O(1)$ because $G_d$ is built from scratch after every update and has size $m^{o(1)}$.

The crucial role that representative times play comes from a monotonicity property of our data structures;
if the witness was well-approximated in the past by a data structure level, then this remains the case until the level is re-initialized.
This behavior is captured in the following lemma, \cref{lemma:hintedTreeChainWidthBound}.
This lemma becomes powerful when we later combine it with \cref{cor:existsBranchGoodStr}, which guarantees that we eventually encounter forests with small stretch of the witness.
\begin{lemma}
\label{lemma:hintedTreeChainWidthBound}
Consider the setting of \cref{lemma:hintedTreeChain} at some stage t, and consider the state of the data structure, right after any sequence of completed procedure calls.
Given any set of representative time stamps $\{\repT_i\}$ and its corresponding set of current stretches $\{\hstr_i\}$, we have
\begin{align}\label{eq:hintedTreeChainWidthBound}
\sum_{i=0}^d \norm{\bw^{(t), G_i}}_1 \le O(\gamma_{\ell})^d \sum_{i=0}^{d} \left(\prod_{i'=i}^{d} \hstr_{i'} \right) \|\bw^{(\repT_i)}\|_1 
\end{align}
\end{lemma}
\begin{proof}
We will prove the inequality by induction on $t$ as well as the level $i.$
In particular, we will prove the following during any stage $t$ and at any level $i$:
\begin{align}\label{eq:treeChainIH}
\norm{\bw^{(t), G_i}}_1 \le O(\gamma_{\ell})^i \norm{\bw^{(t)}}_1 + O(\gamma_{\ell})^i \sum_{a=0}^{i - 1} \left(\prod_{b=a}^{i-1} \hstr_b \right) \|\bw^{(\repT_a)}\|_1
\end{align}

We first analyze $\|\bw^{(t), G_i}\|_1$ at each level $i.$
At level $0$, $\bw^{(t), G_0}$ is exactly $\bw^{(t)}$, the upper bound at stage $t$ and we have $\|\bw^{(t), G_0}\|_1 = \|\bw^{(t)}\|_1.$
At level $i + 1$, suppose $G_{i+1} = \cS(G_i, F)$ where $F$ is the current spanning forest $F^{G_i}_{\bran_i}.$
\cref{lemma:passsparsecore} says
\begin{align*}
\norm{\bw^{(t), G_{i+1}}}_1 \le O(\gamma_{\ell}) \norm{\bw^{(t), \cC(G_i, F)}}_1
\end{align*}

We now focus on bounding $\|\bw^{(t), \cC(G_i, F)}\|.$
Let $U \subseteq G_i$ be the set of edges newly inserted after stage $\repT_i.$
Because $\repT_i$ is no earlier than the previous re-initialization time of $G_i$ and every newly inserted edge $e \in U$ has $\wstr^F_e \defeq 1$ (\cref{lemma:globalstretch}), we have
\begin{align*}
\norm{\bw^{(t), \cC(G_i, F)}}_1 
&\underbrace{\le}_{\text{\cref{lemma:passcore}}} \sum_e \wstr^F_{e} \bw^{(t), G_i}_e \\
&= \sum_{e \not\in U} \wstr^F_{e} \bw^{(t), G_i}_e + \sum_{e \in U} \bw^{(t), G_i}_e \\
&\underbrace{\le}_{\text{Item~\ref{item:widthstable}}} 2 \sum_{e \not\in U} \wstr^F_{e} \bw^{(\repT_i), G_i}_e + \norm{\bw^{(t), G_i}}_1 \\
&= 2 \hstr_i \norm{\bw^{(\repT_i), G_i}}_1 + \norm{\bw^{(t), G_i}}_1
\end{align*}
Observe that $\{\repT_{i'}\}_{i'}$ is valid set of representative time stamps for levels $i' \le i$ during stage $\repT_i.$
This allows us to apply the induction on both $\|\bw^{(\repT_i), G_i}\|_1$ and $\|\bw^{(t), G_i}\|_1$ and yields
\begin{align*}
\norm{\bw^{(t), \cC(G_i, F)}}_1 
&= 2 \hstr_i \norm{\bw^{(\repT_i), G_i}}_1 + \norm{\bw^{(t), G_i}}_1 \\
&\le 2 \hstr_i \left(O(\gamma_{\ell})^{i-1}\norm{\bw^{(\repT_i)}}_1 + O(\gamma_{\ell})^{i-1}\sum_{a=0}^{i - 1} \left(\prod_{b=a}^{i-1} \hstr_b \right) \|\bw^{(\repT_a)}\|_1\right) \\
&+ \left(O(\gamma_{\ell})^{i-1}\norm{\bw^{(t)}}_1 + O(\gamma_{\ell})^{i-1}\sum_{a=0}^{i - 1} \left(\prod_{b=a}^{i-1} \hstr_b \right) \|\bw^{(\repT_a)}\|_1\right) \\
&\le O(\gamma_{\ell})^{i-1}\norm{\bw^{(t)}}_1 + O(\gamma_{\ell})^{i-1}\hstr_i\norm{\bw^{(\repT_i)}}_1 + O(\gamma_{\ell})^{i-1} \sum_{a=0}^{i - 1} \left(\prod_{b=a}^{\boldsymbol{i}} \hstr_b \right) \|\bw^{(\repT_a)}\|_1 \\
&= O(\gamma_{\ell})^{i-1}\norm{\bw^{(t)}}_1 + O(\gamma_{\ell})^{i-1}\sum_{a=0}^{\boldsymbol{i}} \left(\prod_{b=a}^{\boldsymbol{i}} O(\gamma_{\ell}) \cdot \hstr_b \right) \|\bw^{(\repT_a)}\|_1
\end{align*}

Combining with \cref{lemma:passsparsecore}, we have
\begin{align*}
\norm{\bw^{(t), G_{i+1}}}_1 &\le O(\gamma_{\ell}) \norm{\bw^{(t), \cC(G_i, F)}}_1 \\
&\le O(\gamma_{\ell})^i \norm{\bw^{(t)}}_1 +  O(\gamma_{\ell})^i \sum_{a=0}^{\boldsymbol{i}} \left(\prod_{b=a}^{\boldsymbol{i}} \hstr_b \right) \|\bw^{(\repT_a)}\|_1
\end{align*}

The lemma follows by taking the sum of \eqref{eq:treeChainIH} from $i = 0$ to $d$ and the fact that $\repT_d$ is always $t$ during any stage $t$ (\cref{def:repT}).
\end{proof}

By appealing to  \cref{lemma:hintedTreeChain} and
\cref{lemma:hintedTreeChainWidthBound}, we conclude that
the output circulation has a ratio small enough to meet the condition of \cref{thm:MMCHiddenStableFlow} if the following holds at time $t$ for some set of representative time stamps $\{\repT_i\}$, and its current stretches $\{\hstr_i\}$:
\begin{align}
\label{eq:goodW}
\O(k) O(\gamma_{\ell})^d \sum_{i=0}^{d} \left(\prod_{i'=i}^{d-1} \hstr_{i'} \right) \|\bw^{(\repT_i)}\|_1 \le 100(d+1) \cdot \O(k) \O(\gamma_{\ell})^d \norm{\bw^{(t)}}_1
\end{align}

The contrapositive implication tells us that if the data structure
cannot output a good cycle, we know \eqref{eq:goodW} is violated.
Next, we show that there are only two possible reasons
that \eqref{eq:goodW}  could fail to hold.
In particular, it must either be the case that some current
stretch  $\{\hstr_i\}$ is too large or that the current witness
norm $\norm{\bw^{(t)}}_1$ has dropped much below the largest value
$\max_{i=0}^{d}\|\bw^{(\repT_i)}\|_1$ at the representative
times for different levels.

\begin{lemma}
  \label{lem:badCycleMeansBadWtOrStr}
  Set depth $d \defeq O(\log^{1/18} m),$
  then for some $\kappa = \exp\left(-O\left(\log^{17/18} m \log\log
      m\right)\right)$ the following holds. 
Consider the setting of \cref{lemma:hintedTreeChain} at some stage t, and consider the state of the data structure, right after any sequence of completed procedure calls.
  Given any set of representative time stamps
  $\{\repT_i\}_{i=0}^{d}$ and its corresponding set of current
  stretches $\{\hstr_i\},$
  if the cycle quality output by \textsc{FindCycle}() (see \cref{algo:hintedTreeChain}) is not good, i.e.,
  \[
    \frac{\abs{\l\bg^{(t)}, \bDelta\r}}{\norm{\bell^{(t)} \circ \bDelta}_1}
    <
    \kappa \alpha
  \]
  then at least one of the following two conditions hold:
  \begin{equation}
    \label{eq:dsBadWtSum}
   \sum_{i=0}^d \norm{\bw^{(\repT_i)}}_1 > 2(d+1) \norm{\bw^{(t)}}_1
 \end{equation}
 or
   \begin{equation}
    \label{eq:dsBadStr}
    \text{For some level $i$, we have } \hstr_{i,j}  > \tilde{\Omega}(1)
    \text{ on the current branch } j = \bran_i.
 \end{equation}
\end{lemma}

\begin{proof}
  We prove the lemma by proving the contrapositive form of the implication.
  Thus, we assume both Condition~\eqref{eq:dsBadWtSum}
  and~\eqref{eq:dsBadStr} are false.

Then our cycle $\bDelta$ satisfies
\begin{align*}
\frac{\abs{\l\bg^{(t)}, \bDelta\r}}{\norm{\mL^{(t)} \circ \bDelta}_1}
&\underbrace{\ge}_{\text{\cref{lemma:hintedTreeChain}}} \frac{\abs{\l\bg^{(t)}, \bc^{(t)}\r}}{\O(k)\sum_{i=0}^d \norm{\bw^{(t), G_i}}_1} \\
&\underbrace{\ge}_{\text{\cref{lemma:hintedTreeChainWidthBound}}} \frac{\abs{\l\bg^{(t)}, \bc^{(t)}\r}}{\O(k) O(\gamma_{\ell})^d \sum_{i=0}^{d} \left(\prod_{i'=i}^{d} \hstr_{i'} \right) \|\bw^{(\repT_i)}\|_1} \\
  &\underbrace{\ge}_{\substack{\text{ Conditions } \\ \text{\eqref{eq:dsBadWtSum}
  and~\eqref{eq:dsBadStr} } \\ \text{are false.}}}
    \frac{1}{100(d+1) \cdot \O(k) \O(\gamma_{\ell})^d}
  \frac{\abs{\l\bg^{(t)}, \bc^{(t)}\r}}{\|\bw^{(t)}\|_1}
  \\
  &\underbrace{\ge}_{\substack{\text{stable $\alpha$-flow} \\ \text{by \cref{lemma:hintedTreeChain}}\\ \text{assumptions.}}}  \kappa \alpha
\end{align*}
where in the last inequality we choose $d = O(\log^{1/18} m)$ and
$\kappa$ such that
\begin{align*}
\frac{1}{\kappa} &= 100(d+1) \cdot \O(k) \O(\gamma_{\ell})^d \le d \cdot m^{1/d + o(1)} \cdot \exp\left(O\left(d \log^{8/9} m \log\log m\right)\right) \\
&= \exp\left(O\left(\log^{17/18} m \log\log m\right)\right)
\end{align*}
\end{proof}

\subsection{Rebuilding and Shifting}
\label{sec:datastructure-theorem}

\cref{lem:badCycleMeansBadWtOrStr} allowed us to conclude that if we
cannot find a good cycle, it must be either because our data structure
at some level $i$ has bad stretch $\hstr_{i} $ for the current
branch (Condition~\eqref{eq:dsBadStr} in the lemma), or because the current witness norm $\|\bw^{(t)}\|_1$ has dropped
much below the norm at the earlier representative times
(Condition~\eqref{eq:dsBadWtSum} in the lemma).
Next, we want to observe that on every level $i$, there exists a branch
$j$ that leads to good stretch.
This will help us ensure that if we use $\Shift$ to try out all the branches in an
appropriate order, we must eventually have good stretch on all levels.
\begin{corollary}\label{cor:existsBranchGoodStr}
Under the setting of \cref{lemma:hintedTreeChain}, during any stage $t$, consider any set of
representative time stamps $\{\repT_i\}_{i=0}^{d}$
(\cref{def:repT}) and any level $i$.
Then, for some branch $j$, $\hstr_{i,j} \leq \O(1)$.
\end{corollary}
\begin{proof}
Let $\{(F_j, \wstr^{F_j})\}$ be the collection of low stretch forests from \cref{lemma:strMWU}.
We have
\begin{align*}
    \frac{1}{k}\sum_{j=0}^{k-1} \sum_{e \in G_i} \wstr^{F_j}_e \bw^{(\repT_i), G_i}_e = O(\log^7 n) \sum_{e \in G_i} \bw^{(\repT_i), G_i}_e
\end{align*}
Thus, one of the $j^*$ satisfies that
\begin{align*}
    \sum_{e \in G_i} \wstr^{F_{j^*}}_e \bw^{(\repT_i), G_i}_e = O(\log^7 n) \sum_{e \in G_i} \bw^{(\repT_i), G_i}_e
\end{align*}
Thus $\hstr_{i,j^*} = O(\log^7 n).$
\end{proof}

Informally, we now note that if we have done shifts that ensure
all branches have good current stretch, and we \emph{still} cannot
find a good cycle, it must be because the norm $\|\bw^{(t)}\|_1$ has
dropped much below the norm at the earlier representative times
(Condition~\eqref{eq:dsBadWtSum} in the lemma).
Na\"{i}vely, we could fix this by re-initializing everything, which
would update the representative times to the current time, and then
trying every possible shift at each level.
But, this is slow, and in fact much too slow, once we realize that the norm
$\|\bw^{(t)}\|_1$ could drop again as soon as we find another good
cycle, and we might have to re-initialize everything again.

To deal with this, we note that there is an important constraint on
these norms: $\log(\|\bw^{(t)}\|_1) \in (-\log^{O(1)}m,\log^{O(1)}m)$, so the weight can only halve $2\log^{O(1)}m$
times.
If the data structure only had a single level, i.e., $d = 1$, this
would immediately lead to a working strategy: try all branches, and
re-initialize if they all failed.
However, the situation gets more complicated, because deep levels of
the data structure can only survive through few updates, and so we are
forced to re-initialize them frequently, which may happen at times when the
norm $\|\bw^{(t)}\|_1$ is large.

To deal with this, we develop an algorithmic strategy for choosing
when to shift and when to re-initialize levels when trying to find good cycles to route along. 
This yields our overall algorithm for firstly, updating the data structure in response to changes coming from the IPM,
secondly, finding a good cycle to route or shifting and re-initializing until one is found, and finally routing flow along a cycle.
This is encapsulated in the
$\textsc{Update}(\ldots)$ procedure (\cref{alg:rebuildScheduleForDS}).

At a high level, the strategy plays a game, called the \emph{shift-and-rebuild game} (\cref{def:game}), against the IPM.
At each stage, the IPM picks a witness $\bw^{(t)}$ unknown to the data structure.
The data structure outputs a circulation of a ratio small enough.
If it fails to find one, it can pick a level and shift it.
The stage does not finish until the data structure finds circulation of a small ratio.

\begin{algorithm}[!ht]
  \SetKwProg{Proc}{procedure}{}{}
  \Proc{\textsc{Update}$(U^{(t)}, \bg^{(t)}, \bell^{(t)}, \eta)$}{
    \textsc{Update}($U^{(t)}, \bg^{(t)}, \bell^{(t)}$) \tcp{call {\DynamicBranchingChain} update. This could update the $\bran$ variables.} 
    \label{lne:dsUpdate}
    \tcp{ GAME: Let $i^*$ be the smallest level
      index where the number of updates of  $G_i$ exceeds $m
      (\gamma_{\ell} / k)^{i+1} / \log^2 n$
      and hence during  \textsc{Update}(), we call $\Rebuild(i^*)$.
      If $i^*$ exists, the adversary
      declares a rebuild at level $i^*$.
    }
    $\bDelta  \gets \textsc{FindCycle}()$ \label{lne:queryDS}\;
    \If{ $  \frac{\abs{\l\bg^{(t)}, \bDelta\r}}{\norm{\bell^{(t)}
            \circ \bDelta}_1} \le \kappa \alpha$}{
      \tcp{ GAME: the round is not completing.}
      Let $i$ be the largest level index with $\passes_i <
      2\Wrange$, where $\Wrange = \log^{O(1)}m$.
            \tcp{ GAME: The player declares a \textbf{shift} at level $i$}
      \Shift(i) \tcp{$\bran$ variables are updated} 
      \If{$\bran_i = 0$}{
        $\passes_{i} \gets \passes_{i} +1 $. \tcp*{Passed through all branches.}
      }
      \lFor{level $ i' \in \{i+1, \ldots, d\}$}{
        $\passes_{i'} \gets 0 $.
      }
      \textbf{go to} \cref{lne:queryDS}.
    }
    Compute $\bDelta^\top \bg^{(t)}$ using link-cut trees\;
    $\bDelta \gets \frac{\eta}{\bDelta^\top \bg^{(t)}} \bDelta$\;
    $\bf \gets \bf - \bDelta$    using link-cut trees\;
    Also track $\abs{\bDelta}$ updates using link-cut trees to support \textsc{Detect}\;
    \tcp{ GAME: The round is completing. }
  }
  \caption{\textsc{Update}: Data structure shift-and-rebuild schedule.
    \\
     \texttt{// GAME comments are intended to
    help the reader verify the correspondence between 
    \cref{alg:rebuildScheduleForDS} and \cref{alg:rebuildStrategy}
    (for the proof of \cref{lem:AlgoIsGameInstance}).
      }}
  \label{alg:rebuildScheduleForDS}
\end{algorithm}

The efficiency of this strategy is established by
\cref{lem:AlgoIsGameInstance}, which bounds the number of shift calls
as function of the number of $\Rebuild()$ calls.
The number of $\Rebuild()$ calls can be easily bounded by the fact that
each such call only occur in response to the the data structure receiving many
updates, which implies that $\textsc{Update}()$ returned a
good cycle many times.
Roughly speaking, we show that the strategy shifts any level at most $\O(k)$-times until rebuilds that level.

\begin{restatable}{lemma}{AlgoIsGameInstance}
\label{lem:AlgoIsGameInstance}
 In the execution of $T$ calls to $\textsc{Update}()$ (\cref{alg:rebuildScheduleForDS}), let $r_i$
   equal the total number of calls to $\Rebuild(i)$ and let $s_i$ equal the number of calls to
   $\Shift(i)$. Then
   \[ s_i \leq \sum_{i' \leq  i} r_{i'} (k\log^{O(1)}(m))^{i+1-i'}. \]
\end{restatable}
\cref{sec:rebuilding} is dedicated to proving \Cref{lem:AlgoIsGameInstance}. Here we conclude this section by showing how to prove \Cref{thm:MMCHiddenStableFlow} using \Cref{lem:AlgoIsGameInstance}.
\MMCHSF*
\begin{proof}[Proof of \Cref{thm:MMCHiddenStableFlow}]

  We utilize the data structure provided by
  \cref{lemma:hintedTreeChain}, and combine it with link-cut
  trees~\cite{ST83}. Since \cref{lemma:hintedTreeChain} maintains a tree chain $\cG$ and outputs an amortized $m^{o(1)}$ edge updates to the underlying spanning tree $T^{\cG}$, we can use Link-Cut trees to
  represent $T^{\cG}.$ We will use these to maintain both the flow (as
  a sum of scaled $\bDelta^{(t)}$) and the absolute value of the updates
  $\abs{\bDelta^{(t)}}$ (for implementing \textsc{Detect}). This
  allows us to support updating gradients/lengths, routing flow along
  a tree-path, computing the inner product of a tree path with $\bg^{(t)}$
  or its length in $O(\log m)$ amortized time. In order to support
  \textsc{Detect}, we rewrite the condition as
  \[
  \sum_{t' \in [\last^{(t)}_e+1, t]} |\bDelta_e^{(t')}|
  -\frac{\eps}{\bell_e} \ge 0\,.
  \]
  We initialize each edge at
  $-\eps / \bell_e$ and repeatedly add $ |\bDelta_e^{(t')}|$ whenever
  we route some flow. Thus, detecting edges is now equivalent to
  finding edges with non-negative values, which can be done in
  $O(\log m)$ amortized time per edge returned by
  \textsc{Detect}. Since all link-cut tree operations taken
  $O(\log m)$ amortized time, and \cref{lemma:hintedTreeChain}
  guarantees that the number of edge insertions and deletions is bounded
  by the running time, our total running time remains unchanged up to
  polylog factors.

  The key piece now remaining is to prove that we can efficiently find
  a good circulation $\bDelta^{(t)}.$
  By \cref{lemma:hintedTreeChain}, the total time spent is
  \begin{align}
m^{1/d}\O(\gamma_{\ell}\gamma_r)^{O(d)}\left(m + \sum_{i=0}^{d}
    (s_i + r_i) \cdot m^{1 - i/d}\right)
    \label{eq:restatingTimeBoundSR}
  \end{align}
  where $s_i$ and $r_i$, the number of calls to
  $\textsc{Shift}$ and $\textsc{Rebuild}$ at each level $i$.

  Next, recall that $Q = \sum_t |U^{(t)}|$ denotes the total update count.
  We can bound the rebuild counts $r_i$ by, first, observing that when a rebuild
  at level $i$ occurs, the level must have received $m (\gamma_{\ell}
  / k)^{i+1} / \log^2 n$ updates, by the pseudo-code in $\Update()$ in
  \cref{algo:hintedTreeChain}.
  Second, 
  the number of updates received at level $i$ from is
  bounded by $\gamma_r^iQ$ by the recourse bounds from
  \cref{lemma:passsparsecore}.
  Thus,
  \[
    r_i \leq
\frac{\gamma_r^iQ }
{m (\gamma_{\ell}
  / k)^{i+1} / \log^2 n}
=
\left(
\frac{k\gamma_r }
{\gamma_l}
\right)^i
\frac{kQ\log^2m}{\gamma_{\ell} m}
\leq 
k^i
\frac{k \gamma_r^d Q}{ m}
\]
  By \cref{lem:AlgoIsGameInstance}, we know that $s_i \leq \sum_{j
    \leq  i} r_{j} (k\log^{O(1)}(m))^{i+1-j}$.
Combining these observations gives us an overall bound on the shift
counts of
\[
  s_i \leq
  \sum_{j
    \leq  i}
 (k\log^{O(1)}(m))^{i+1-j}
  k^{j}
  \frac{k \gamma_r^d Q}{ m}
  \leq
  k^i
  \frac{d k^2\log^{O(d)}(m) \gamma_r^d Q}{ m}
  =
  k^i
   \frac{k^2\gamma_r^{O(d)} Q}{ m}
  \]
Recalling that $k = m^{1/d}$, and plugging our bounds into
\cref{eq:restatingTimeBoundSR},
we bound the overall running time by
\begin{align*}
  m^{3/d}\O(\gamma_{\ell}\gamma_r)^{O(d)}\left(m + Q\right)
  = \left(m + Q\right) m^{o(1)}
  .
\end{align*}
\end{proof}

%% file: branch_shift_rebuild_game.tex
\section{The Shift-and-Rebuild Game}
\label{sec:rebuilding}

In this section we translate the execution of
repeated calls to $\textsc{Update}()$ (\cref{alg:rebuildScheduleForDS})
into an instance of a game played
between a player and an adversary and analyze strategies for the game to prove the
efficiency of the algorithm.
We first establish a more general game, called the \emph{shift-and-rebuild game} (\cref{def:game}), that abstracts away
most of the data structure behavior.
We analyze this game, and provide an efficient strategy for playing this game. 
This is our central technical
result of the section, \cref{lem:goodRebuildStrategy}.
We also prove that $\textsc{Update}()$ indeed
corresponds to particular instances of
the game, and that the algorithm corresponds to playing our efficient strategy.
Therefore, we show the efficiency of the algorithm in our main result of the section, \cref{lem:AlgoIsGameInstance}.

We start by briefly recalling elements of
\cref{alg:rebuildScheduleForDS}, before defining our game.
However, the game and its analysis in the proof of
\cref{lem:goodRebuildStrategy} can be read in isolation, without first
checking the correspondence to our data structure problem.

\paragraph{Shift-and-Rebuild game motivation.}
Recall that in each call to $\textsc{Update}()$, we first
update our data structure, rebuilding any sparse core graph levels as
necessary, and then find a cycle to route flow along.
If the cycle has a small enough ratio, we use it  -- in our game, this will
correspond to a round of the game completing.
If the cycle is not good enough, $\textsc{Update}()$
decides on a data structure level $i$ to call $\Shift(i)$ on and repeats this until a good cycle is found.
Every completed call to $\textsc{Update}()$ will correspond to a round of our game.

We set up our game so that it captures the behavior of $\textsc{Update}()$.
This game is played between a player whose actions correspond to
decisions made by $\textsc{Update}()$ about which levels
to call $\Shift()$ on and an adversary where 
firstly, its actions capture the
the times when we are
forced to rebuild data structure levels in
\cref{lne:dsUpdate} of $\textsc{Update}()$ (\cref{alg:rebuildScheduleForDS}), and secondly, 
its actions reflect the behavior of the witness
norm $\norm{\bw^{(t)}}_1$ and the data structure stretch values $\hstr_{i,j}$.
In the proof of \cref{lem:AlgoIsGameInstance}, we set up the
precise correspondence between our game and $\textsc{Update}()$.

\paragraph{Shift-and-Rebuild game parameters and definitions.}
  The Shift-and-Rebuild game has several parameters: integer depth $d > 0$, a
  branching factor $k$, 
  and a weight range $\Wrange \geq 1$.
  We say the game has $d+1$ \emph{levels} and each level has $k$
  \emph{branches}.
  In the game, these levels and branches are simply
  abstract indices, but ultimately, they will correspond to the data
  structure levels and branches in \cref{thm:MMCHiddenStableFlow}.

    The game is played between a player and an adversary and proceeds
    in rounds $t=1,2,\ldots,T$.
    At certain points in the game, the player or
    adversary may update certain variables.
    Furthermore, at special points in the game,
    the player or the
    adversary is required to perform a ``step,'' after which the value of some
    game variables are updated.
    These ``steps'' play an important role in our later analysis of
    strategies for the game and hence to index the steps,
    we keep a step counter 
    $s \in \{1,2,\ldots,S\}$.
    
  As part of the game, we define a number of variables. All of these
  quantities are updated at different points of the game, stated in
  the formal description in \cref{def:game}.
     For every level index $i \in [[d]]$,
 \begin{itemize}
     \item We define a
    current branch  $\bran^{(s)}_i \in [[k-1]]$.
    The variable gets updated at various points in the game.
    Both the player and
    adversary know its value.
    \item For every round $t \in [T]$ we define a ``weight'' $W^{(t)}$,
    satisfying $\log W^{(t)} \in (-\Wrange, \Wrange).$ 
  The weight values are \emph{hidden} from the player.   
  \item   We define a ``representative round'' index $\grepT_i \in [T]$.
    At certain points in the game, a new value for the representative
    round is set by the rules of the game, and the outcome is \emph{hidden}
    from the player.
  \item For every branch $j \in [[k-1]]$ of the level, 
    we define a ``stretch'' $\gstr^{(s)}_{i,j} \geq 0$.
    The adversary chooses these values anew at various points
    and the outcome of this choice is \emph{hidden} from the player.
  \end{itemize}

Note that we use different symbols to clearly distinguish the data
structure and game variables.
The game variables $\grepT_i$ and $\gstr^{(s)}_{i,j}$ will correspond
to data structure variables $\repT_i$ and $\str^{(s)}_{i,j}$, and we
formally establish the correspondence as 
part of the proof of \cref{lem:AlgoIsGameInstance}.
Both the data structures and the game uses  counters $\bran^{(s)}_i \in [[k-1]]$
and   $\passes_i \in [[W]]$ for each level $i$.
Here, the correspondence is trivial, and so, overloading notation, we
avoid defining new names for the game variables.

It is crucial to our game strategy that the adversary's choice of
$\gstr^{(s)}_{i,j}$ values is constrained -- we encapsulate this in
the following notion of \emph{valid} stretch values.

\begin{definition}[\emph{Valid} stretch values]
  \label{def:gameValidStretch}
  For each level $i \in [[d]]$, we say the stretch
  values for this level are \emph{valid} if there exists a branch
  $j$ satisfying $\gstr^{(s)}_{i,j} \leq \O(1)$.
\end{definition}
The \emph{value} of the completed game is given by two tuples
$(s_0,\ldots,s_{d}) \in \mathbb{N}^d$ and   $(r_0,\ldots,r_{d}) \in
\mathbb{N}^d$, where $s_i$ is the number of times the player took
the step \gameActParam{shift}{$i$}, and $r_i$ is the number of times the player took
the step \gameActParam{rebuild}{$i$}

\paragraph{Game outline.}
Having established the basic parameters and variable definitions, we
briefly outline the game, before giving the formal definition in \cref{def:game}.
At the start of each round $t$, the adversary decides weight
  $W^{(t)} \geq 0$ 
  (in our data structure analysis, this will correspond to the new
  value of the witness norm $\norm{\bw^{(t)}}_1$ after update).
  After this weight choice, the adversary may choose to take the step \gameActParam{rebuild}{$i$} which updates some of the game variables (this will correspond to $\Rebuild$ calls in our data structure).
  Then, the adversary decides if the round is completing or continuing
 (a completing round corresponds to the data structure finding a good
 enough cycle, while a continuing round corresponds to the cycle not
 being good enough).
 The game sets rules for when the adversary may decide to continue the
 round.
  Next, the player chooses a step, choosing either to 
  take a \gameAct{do-nothing} or to take
the step \gameActParam{shift}{$i$} for some level $i \in
  [[d]]$ (this will correspond to $\Shift$ calls in our data structure).
  The shift step updates some of the game variables.
  After this, the round either completes, or repeats from the point
  where the adversary decides whether to complete the round or not.

\begin{definition}[Shift-and-Rebuild Game]
  \label{def:game}
  The \emph{Shift-and-Rebuild Game} is defined as follows.
  At the beginning of the game, at round $t = 1$ and step $s =
  0$,
  and for all levels $i \in [[d]]$, we set
  $\bran_i \gets 0$, and $\grepT_i \gets t$. Additionally, the adversary
chooses
\emph{valid} stretch values
for each branch $j \in [[k-1]]$ of the level $i$,
  $\gstr_{i,j} \geq 0$.
The game then proceeds across rounds $t = 1,2,\ldots,T$.
In each round $t \in [T]$, we proceed through the following game stages in
order, and, depending on conditions outlined in
\cref{enu:gameRoundEnd}, we complete the round at the end of this
game stage, or move back to \cref{enu:fixingTest} and continue through the
game stages again from there.
  \begin{enumerate}
  \item\label[gameStage]{enu:gameStart}
    The adversary first chooses a positive real weight $W^{(t)}$ satisfying $\log_2 W^{(t)} \in (-\Wrange, \Wrange).$
    This weight is \emph{hidden} from the player.
  \item\label[gameStage]{enu:gameRebuild}
    Next, the adversary must choose a step.
    The adversary may pick any level $i$ and choose the step
    \gameActParam{rebuild}{$i$}, or
    choose the step \gameAct{do-nothing}.
    
    If the adversary chooses \gameActParam{rebuild}{$i$},
    for all levels $i' \geq i$, we set $\bran_{i'} \gets 0$ and $\grepT_{i'} \gets t$; 
     and the adversary chooses 
    \emph{valid} stretch values $\gstr_{i',j} \geq 0$
    for each branch $j \in [[k-1]]$.
    These stretch values are \emph{hidden} from the player.

    The player is informed of the step chosen by the adversary, including
    the level value $i$ if the adversary chose
    \gameActParam{rebuild}{$i$}.

    Then, the step counter is incremented:
    $s \gets s + 1$.
  \item 
    \label[gameStage]{enu:fixingTest}
    Next, the adversary must again choose a step.
    If at least one of the following two conditions is true,
    \begin{align}
      \label{eq:gameLossSumCond}
      &\sum_{i=0}^d W^{(\grepT_i)} > 2(d+1) W^{(t)}
      \end{align}
      or
      \begin{align}
      \label{eq:gameBadStr}
      &  \text{some level $i$, we have } \gstr_{i,j}  > \tilde{\Omega}(1) \text{ on the current
        branch } j = \bran_i
    \end{align}
    then the adversary must choose either the step
    \gameAct{round-completing}
    or must choose the 
    step \gameAct{round-continuing}.
    If \emph{neither} condition holds,
    the adversary must choose the step \gameAct{round-completing}.
    
    The player is informed of the step chosen by the adversary,
    i.e., \gameAct{round-completing} or \gameAct{round-continuing},
    but the player is not informed whether
    \eqref{eq:gameLossSumCond} or \eqref{eq:gameBadStr} held.

    In all cases, proceed to \cref{enu:gamePlayerStep}
    and then~\cref{enu:gameRoundEnd}.
  \item
    \label[gameStage]{enu:gamePlayerStep}
    The player must choose a step.
    The player can either choose the step \gameAct{do-nothing} or select
    a level $i$ and choose step \gameActParam{shift}{$i$}. In the
    latter case, we update variables as follows:
    \begin{itemize}      
    \item For all levels  $i' > i$, we set $\bran_{i'} \gets 0$; and
      $\grepT_{i'} \gets t$ and
      the adversary chooses
    \emph{valid} stretch values $\gstr_{i',j} \geq 0$
    for each branch $j \in [[k-1]]$.
    \item We set $\bran_i \gets (\bran_i + 1
      \mod k)$.
      If the updated $\bran_i$ is zero, then we set \\
      $\grepT_i \gets \argmin_{\grepT_i \le x \le t} W^{(x)}$,
      and the adversary
      chooses \emph{valid} stretch values $\gstr_{i',j} \geq 0$
      for each branch $j \in [[k-1]]$.
    \end{itemize}
    Then, the step counter is incremented:
    $s \gets s + 1$.
  \item
    \label[gameStage]{enu:gameRoundEnd}
    If, in the latest execution of \cref{enu:fixingTest},
    the adversary chose \gameAct{round-completing}, then the
    round counter $t$ is incremented, $t \gets t + 1$, and we move to
    the next round.
    Otherwise (i.e., the adversary chose
    \gameAct{round-continuing})
    we continue the round, moving to \cref{enu:fixingTest} and
    continuing from there.
  \end{enumerate}
  When all rounds of the game are complete, the  game returns a value \emph{value}
  given by tuples $(s_0,\ldots,s_{d}) \in \mathbb{N}^d$ and $(r_0,\ldots,r_{d}) \in
  \mathbb{N}^d$, where $s_i$ is the number of times the player took
  the step \gameActParam{shift}{$i$} and $r_i$ is the number of times the adversary took
  the step \gameActParam{rebuild}{$i$}.
\end{definition}

Now that we have set up the Shift-and-Rebuild game, we state a
strategy for the player in \cref{alg:rebuildStrategy} below.
The strategy is implemented by a simple pseudocode that uses an
additional variable $\passes_i$ for each level.

\begin{algorithm}[ht]
  \ForEach{level $i \gets 0,\ldots, d$.}{
    Player maintains a "passes count", $\passes_i$, initialized to zero.\;
  }
  \ForEach{round $t \gets 1,2, \ldots, T$ of the game}{
    The adversary chooses weight $W^{(t)}$, \emph{hidden} from the
    player.  \hfill\tcp{\cref{enu:gameStart}} 
    The adversary may choose \gameActParam{rebuild}{$i$},
    and
    if so, the player is informed of the level $i$. The $\bran$, $\grepT$, and  $\gstr$ update arbitrarily, following the rules of the game. \hfill\tcp{ \cref{enu:gameRebuild}} 
    \If{ the adversary chose \gameActParam{rebuild}{$i$}}{
      \ForEach{level $ i' \gets i,i+1, \ldots, d$ of the game}{
        $\passes_{i'} \gets 0 $.
      }
    }
    The adversary decides if the round is completing or not,
    by choosing \gameAct{round-completing} or \gameAct{round-continuing}
    and the player
    is informed of the choice. \hfill\tcp{ \cref{enu:fixingTest}}      \label{lne:roundCompleteDecision}
    \If{ the adversary chose \gameAct{round-continuing}.}{
      Let $i$ be the largest index level with $\passes_i < 2\Wrange$.\;
      The player chooses step \gameActParam{shift}{$i$}.
      Some  $\bran$, $\grepT$, and $\gstr$ variables update.\;
      \, \hfill \tcp{\cref{enu:gamePlayerStep} }
      \lIf(\tcp*[f]{Passed through all branches.}){ $\bran_i = 0$ }{
        $\passes_{i} \gets \passes_{i} +1 $.
      }
      \ForEach{level $ i' \gets i+1, \ldots, d$}{
        $\passes_{i'} \gets 0 $.
      }
    }
    \If{the adversary chose \gameAct{round-continuing}}{
      \textbf{go to} \cref{lne:roundCompleteDecision}. \hfill\tcp{\cref{enu:gameRoundEnd}}
    }
  }
  \caption{Player strategy for the rebuilding game.}
  \label{alg:rebuildStrategy}
\end{algorithm}

Now that we have a concrete strategy for the game, we state our main
technical lemma of this section, \cref{lem:goodRebuildStrategy}, which proves that the strategy
implemented by \cref{alg:rebuildStrategy} is efficient in a certain
sense.
In particular, the lemma upper bounds the number of times the player
may need to choose the step \gameActParam{shift}{$i$} for each level $i$, as a function
of how many times the adversary chooses \gameActParam{rebuild}{$i'$} at levels
$i' \leq i$.

\begin{lemma}
  \label{lem:goodRebuildStrategy}
  In the Shift-and-Rebuild game, if the player follows
  Strategy~\ref{alg:rebuildStrategy}, then when the game completes,
  for all levels $i \in [[d]]$, we have
  $s_i \leq \sum_{i' \leq  i} r_{i'} (10\Wrange k)^{i+1-i'}$.
\end{lemma}

We prove this lemma in \cref{sec:gameStrategyGood}, but first, we show
how to use it to prove that \cref{alg:rebuildScheduleForDS} is efficient.

\subsection{Game Playing Strategy}
We now prove prove
\cref{lem:AlgoIsGameInstance} (restated below) by showing that
\cref{alg:rebuildScheduleForDS} corresponds to executing
Strategy~\ref{alg:rebuildStrategy} in particular instances of the
Shift-and-Rebuild game, which then lets us apply
\cref{lem:goodRebuildStrategy} to shows its efficiency.

\AlgoIsGameInstance*

\begin{proof}[Proof of \cref{lem:AlgoIsGameInstance}]
  We first establish a correspondence that allows us to derive an
  instance of the Rebuild-and-Shift game from the execution of
  \cref{alg:rebuildScheduleForDS}.
  We define the game to have the same level count $d$ and branch count
  $k$ as the data structures of $\textsc{Update}()$.
  We let $\Psi = \log^{O(1)}m$.
  
  To establish the correspondence, we
  define the adversary's actions as follows:
  \begin{enumerate}
  \item At the beginning of every round the game adversary chooses
    $W^{(t)} = \norm{\bw^{(t)}}_1$.
    It is immediate that the choice of  $W^{(t)} = \norm{\bw^{(t)}}_1$
    satisfies $\log_2 W^{(t)} \in (-\log^{O(1)}m, \log^{O(1)}m)$ by
    \cref{item:quasipoly} in \cref{def:hiddenStableFlowChasing},
    as the updates received by $\textsc{Update}()$ are Hidden Stable Flow Chasing.
  \item When data structure level $i$ initiates a rebuild in
    \cref{alg:rebuildScheduleForDS}, the adversary chooses the action
    \gameActParam{rebuild}{$i$}.
    Thus the number of calls to $\Rebuild(i)$ equals
    $r_i$ by \cref{def:game}.
  \item When a data structure rebuild  leads to new stretches
    $\hstr_{i,j}$  for level $i$ and branch $j$ (which occurs either
    during a  $\Rebuild(i')$ or a  $\Shift(i')$ data structure
    call),
    the adversary chooses
    $\gstr_{i,j} \gets \hstr_{i,j}$.
    Observe that this choice of game stretch values are \emph{valid} by \cref{cor:existsBranchGoodStr}.
  \item When the data structure cycle query $\textsc{FindCycle}()$  fails, the
    adversary decides that the round is \gameAct{continuing} and when
    the cycle query succeeds, the adversary decides that the round is
    \gameAct{completing}.
    \label{enu:roundCompletingIffCycleFound}
  \end{enumerate}
  Furthermore, the game player chooses \gameActParam{shift}{$i$} at
  level $i$, exactly when in $\textsc{Update}()$ we call $\Shift(i)$.
      Thus the number of calls to $\Shift(i)$ equals
    $s_i$ by \cref{def:game}.
  
  With this correspondence, we see that we can trivially identify the
  $\bran_i$ and $\passes_i$ variables in
  \cref{alg:rebuildScheduleForDS} with the same variables in the
  Strategy~\cref{alg:rebuildStrategy} for the player.

  Next, we have to confirm that the adversary in the game only
  chooses \gameAct{round-continuing} at points when this
  is allowed according to the game rules in \cref{enu:fixingTest}.
  Observe that when the adversary makes this choice,
  we must have that the data structure failed to
  find a cycle in \cref{alg:rebuildScheduleForDS}.

  Notice that if we choose $\repT_i$ equal to $\grepT_i$,
these choices fullfill the requirements of
\cref{lem:badCycleMeansBadWtOrStr}.

  By the lemma, using our corresponce, if we failed to find a cycle in
  the data structure, at least one of game
  Conditions~\eqref{eq:gameLossSumCond} or ~\eqref{eq:gameBadStr} is
  true, which means that the adversary is allowed choose
  \gameAct{round-continuing}
  by the game rules in \cref{def:game}.
\end{proof}

\subsection{Analysis of the Strategy}
\label{sec:gameStrategyGood}

In this section, we prove \Cref{lem:goodRebuildStrategy}.
Throughout this section, we parameterize variables by step count $s$ in
our proofs.
The step counter is always updated immediately after updating
variables, and by $\gstr^{(s)}_{i,j}$, $\grepT_i^{(s)}$ etc. we may
the value of the variable immediately after the step counter is
updated.

We now summarize how the game representative rounds
$\{\grepT_i^{(s)}\}_i$ change when we either rebuild or shift some
level $i$.
\begin{definition}[The rule for updating $\{\grepT_i\}_i$]
\label{def:repTRule}

If the player chooses \gameActParam{shift}{$i$} at step $s$ and
the branching index then becomes $0$, i.e., $\bran_i = 0$, we set $\grepT_i^{(s+1)}$ to
\begin{align}\label{eq:repTAssign}
\grepT_i^{(s+1)} \gets \argmin_{\grepT_i^{(s)} \le x \le t} W^{(x)},
\end{align}
We also set $\grepT_{i'}^{(s+1)} \gets t$ for $i' > i.$

If the adversary chooses \gameActParam{rebuild}{$i$} in step $s$, then $\grepT_{i'}^{(s+1)} \gets t$ for $i' \geq i.$
\end{definition}

Given the rules for updating $\{\grepT_i\}$, the variable $\passes_i$ counts exactly how many times we update $\grepT_i$ using \eqref{eq:repTAssign} since last rebuild at level $i.$

In the rest of the proof, we use $\gstr^{(s)}_i$ to denote $\gstr^{(s)}_{i, j}$ where $j = \bran_i$, the current stretch.

\begin{lemma}\label{lem:goodStrBetween}
Let $s_1, s_2$ be the steps of any two consecutive shifts at level $i$ on the same set of stretch values $\{\gstr_{i, j}\}_j$.
There exists some step $s \in (s_1, s_2]$ such that
\begin{align*}
    \gstr^{(s)}_{i+1} = \ldots = \gstr^{(s)}_{d} = \O(1)
\end{align*}
Furthermore, we can ensure that all $\passes^{(s)}_{i+1}, \ldots, \passes^{(s)}_{d}$ are exactly $2 \Psi - 1.$
\end{lemma}
\begin{proof}
We prove by induction $i$.
The base cases of $i = d - 1$ holds trivially because $G_d$ has constant size and $\gstr^{(s)}_1$ is always $O(1).$
We can just take $s = s_2.$
In the strategy (\cref{alg:rebuildStrategy}), we shift level $d - 1$ at step $s_2$ because $\passes^{(s_2)}_d = 2\Psi.$

For $i < d - 1$, we know that in order to shift level $i$ at step $s_2$, we need to shift level $i+1$ for $2\Psi k$ times.
By \cref{def:gameValidStretch}
we know right after one of the last $k+1$ shifts, $\gstr^{(s)}_{i+1} = \O(1)$ and $\passes^{(s)}_{i+1} = 2\Psi - 1.$
Let $s'$ be the next step during which level $i+1$ is shifted.
By induction, there's some step $s^* \in (s, s']$ such that $\gstr^{(s^*)}_{i+2} = \ldots = \gstr^{(s^*)}_{d} = \O(1)$.
Observe that $\gstr^{(s^*)}_{i+1} = \gstr^{(s)}_{i+1} = \O(1)$
and $\passes^{(s^*)}_{i+1} = \passes^{(s)}_{i+1}$ because we have not shifted level $i+1$ between step $s$ and $s^*.$
This concludes the proof because $\gstr^{(s^*)}_{i+1} = \ldots = \gstr^{(s^*)}_{d} = \O(1)$ and $\passes^{(s^*)}_{i+1}, \ldots, \passes^{(s^*)}_{d}$ are all $2\Psi - 1.$
\end{proof}

Intuitively, our proof of \Cref{lem:goodRebuildStrategy} is based on the observation that $W^{(\grepT_i)}$ becomes smaller whenever we update $\grepT_i$ (see \eqref{eq:repTAssign}).
Furthermore, $W^{(\grepT_i)}$ decreae by at least a half if the current step has good stretches up to level $i$, i.e., $\gstr^{(s)}_0 = \ldots = \gstr^{(s)}_{i - 1} = \O(1)$
and $W^{(\grepT_i)}$ is larger than any $W^{(\grepT_{i'})}$  at level $i' < i.$
Because $W^{(\grepT_i)}$ is within the range $(2^{-\Psi}, 2^{\Psi})$, there will be at most $2\Psi$ updates to $\grepT_i.$
This implies the $2\Psi$ upper bound on $\passes_0$.
Therefore, between two rebuilds at level $i$ (or lower-index shifts that cause a reset of $\passes_i$),  there are at most $2k\Psi$ \gameActParam{shift}{$i$} chosen at level $i.$

\begin{proof}[Proof of \Cref{lem:goodRebuildStrategy}]
Our goal is to prove that at any time, $\passes_0$ cannot exceed $2\Psi.$
Given this fact, we can prove the lemma as follows:
At any level $i$, the number of times the player chooses \gameActParam{shift}{$i$}
w.r.t. a fixed set of stretch values $\{\gstr_{i, j}\}_j$ is at most $2k\Psi$ because $\passes_i$ is at most $2\Psi.$
Throughout the entire game, the stretch set can only be changed at most $\sum_{i' \le i} r_{i'} + \sum_{i' < i} s_{i'}$ times.
We conclude the lemma via induction on $i.$

The rest of the proof aims to show that $\passes_0$ cannot exceed $2\Psi.$
The proof works via induction on step $s$.
We first define two sets for technical reasons:
\begin{definition}[Prefix Max Set]
\label{def:preMaxS}
At any step $s$, we define
\begin{align*}
\cM^{(s)} = \Set{i \in \{0, 1, \ldots, d\}}{W^{(\grepT^{(s)}_{i'})} <  W^{(\grepT^{(s)}_i)}, \forall i' < i}
\end{align*}
\end{definition}
\begin{definition}[Prefix Good Stretch Set]
\label{def:preGoodStrS}
At any step $s$, we define
\begin{align*}
\cP^{(s)} = \Set{i \in \{0, 1, \ldots, d\}}{\gstr^{(s)}_{i'} = \O(1), \forall i' < i}
\end{align*}
\end{definition}
We prove the following holds at any step $s$.

\paragraph{Induction Hypothesis:}
At the start of any step $s$, we have
\begin{align}
\label{eq:IH}
    \log_2 W^{(\grepT^{(s)}_i)} < \Psi - \passes^{(s)}_i, \forall i \in \cM^{(s)} \cap \cP^{(s)}
\end{align}

To verify the induction hypothesis, we go through what happens during step $s$ and how \eqref{eq:IH} could be affected.

\paragraph{Event 1: The adversary chooses \gameActParam{rebuild}{$i$}.}
From the game (\cref{def:game}) and the strategy (\cref{alg:rebuildStrategy}), some variables are updated as follows:
\begin{align*}
    \grepT^{(s+1)}_{i'} &\gets t, \forall i' \le i \\
    \passes^{(s+1)}_{i'} &\gets 0, \forall i' \le i
\end{align*}
Because $W^{(t)} < 2^{\Psi}$ always holds, \eqref{eq:IH} even holds for levels $i' \ge i.$
Every other level, $i' < i$ are not affected, as the variables at step
$s+1$ equal the variables at step $s$.
So,  \eqref{eq:IH} still holds after the completion of a step \gameActParam{rebuild}{$i$}.

\paragraph{Event 2: The adversary or player chooses \gameAct{do-nothing}.}
In this case, every variable stays the same and \eqref{eq:IH} holds at the start of the next step, $s+1.$

\paragraph{Event 3: The player chooses \gameActParam{shift}{$i$}.} 
We will argue that \eqref{eq:IH} still holds at the start of the next step, $s + 1$.
First, let us look at the set $\cM^{(s)} \cap \cP^{(s)}$ before the shift.
\begin{claim}
\label{claim:badLargeLevel}
$\cM^{(s)} \cap \cP^{(s)} \cap \{i+1, \ldots, d\} = \phi.$
\end{claim}
\begin{proof}
Assume for contradiction that there is some $i' > i$ that is in $\cM^{(s)} \cap \cP^{(s)}.$
From the strategy, we know that $\passes^{(s)}_{i'} = 2\Psi$.
Combining with the induction hypothesis~\eqref{eq:IH}, we know
\begin{align*}
    \log_2 W^{(\grepT^{(s)}_{i'})} < \Psi - \passes^{(s)}_{i'} \le \Psi - 2\Psi = -\Psi
\end{align*}
which leads to a contradiction because $W^{(t)} > 2^{-\Psi}$ always holds.
\end{proof}
The following claim argues that we need to check \eqref{eq:IH} only at the level $i.$
\begin{claim}
\label{claim:focusLvlI}
At the start of step $s + 1$, the induction hypothesis~\eqref{eq:IH} holds for all levels $i' \neq i.$
\end{claim}
\begin{proof}
For any larger level $i' > i$, $\passes^{(s+1)}_{i'}$ is $0$ at the start of step $s+1.$
If there is some $i' > i$ in $\cM^{(s+1)} \cap \cP^{(s+1)}$, $\log_2 W^{(\grepT^{(s)}_{i'})} < \Psi$ and \eqref{eq:IH} holds because $W^{(t)} < 2^{\Psi}$.
For any smaller level $j < i$, $W^{(\grepT^{(s+1)}_{i'})}$, $\gstr^{(s+1)}_{i'}$, and $\passes^{(s+1)}_{i'}$ remain unchanged and \eqref{eq:IH} holds trivially.
\end{proof}

Next, we do a case analysis on whether $i \in \cM^{(s)}$ and/or $i \in \cP^{(s)}.$

\paragraph{Case A: $i \not\in \cM^{(s)}.$}
Because level $i$ is not rebuilt, from \cref{def:repTRule} we know $W^{(\grepT^{(s+1)}_i)}$ is at most $W^{(\grepT^{(s)}_i)}.$
And a shift at level $i$ does not affect $\grepT_j$ for any level $i' < i.$
So $i$ cannot be in $\cM^{(s+1)}$ and \eqref{eq:IH} holds for level $i$ at step $s+1.$

\paragraph{Case B: $i \in \cM^{(s)}$ but $i \not\in \cP^{(s)}.$}
After a shift at level $i$, the current stretch at level $i' < i$ remains unchanged, i.e., $\gstr^{(s + 1)}_{i'} = \gstr^{(s)}_{i'}$ for all $i' < i$.
Because $i$ is not in $\cP^{(s)}$, $\gstr^{(s + 1)}_{i'} = \gstr^{(s)}_{i'} > \O(1)$ at some level $i' < i.$
Thus, $i$ is still not in $\cP^{(s + 1)}$ and \eqref{eq:IH} holds for level $i$ at step $s+1.$

\paragraph{Case C: $i \in \cM^{(s)} \cap \cP^{(s)}.$}
If $\grepT^{(s+1)}_i = \grepT^{(s)}_{i}$, $i \in \cM^{(s+1)} \cap \cP^{(s+1)}$ and \eqref{eq:IH} still holds for level $i$ at step $s+1.$

Otherwise, $\grepT^{(s+1)}_i$ is determined using \eqref{eq:repTAssign} and this increments $\passes^{(s+1)}_i \gets \passes^{(s)}_i + 1.$
We verify \eqref{eq:IH} by showing that
\begin{align*}
    W^{(\grepT^{(s+1)}_i)} \le \frac{1}{2}W^{(\grepT^{(s)}_i)}
\end{align*}

By \cref{lem:goodStrBetween}, there must be a step $x$ after the previous shift at level $i$ such that $\gstr^{(x)}_{i+1} = \ldots \gstr^{(x)}_{d} = \O(1)$ and $\passes^{(x)}_{i'} \ge 2\Psi - 1, \forall i' > i$.

Because after the previous shift at level $i$, every level $i' \le i$ is not affected, i.e. $\gstr^{(x)}_{i'} = \gstr^{(s)}_{i'}, \grepT^{(x)}_{i'} = \grepT^{(s)}_{i'}, \forall i' \le i.$
Combining with the fact that $i \in \cP^{(s)}$, we know that at step $x$, all stretches are small, i.e.
\begin{align*}
    \gstr^{(x)}_{0} = \ldots \gstr^{(x)}_{d} = \O(1)
\end{align*}
That is, $\cP^{(x)} = \{0, 1, \ldots, d\}.$

However, at step $x$, the data structure cannot find a good cycle, and \eqref{eq:gameLossSumCond} fails to hold.
Let $t_x$ be the corresponding round during step $x.$
Because \eqref{eq:gameLossSumCond} fails, we know that
\begin{align}
\label{eq:failW}
\sum_{i'=0}^d W^{(\grepT^{(x)}_{i'})} > 2(d+1) W^{(t_x)}
\end{align}

Now, we want to prove that $W^{(\repT^{(x)}_i)}$ is the largest among all levels by contradiction.
Let $i^*$ be the level to shift at step $x.$
Assume for contradiction that $i^* \in \cM^{(x)}$ and $i^* \neq i.$
\Cref{claim:badLargeLevel} and the fact that $\cP^{(x)} = \{0, 1, \ldots, d\}$ says that $i^* = \max \cM^{(x)}$ and we have
\begin{align*}
    W^{(\grepT^{(x)}_{i^*})} = \max_{i' \in \{0, 1, \ldots, d\}} W^{(\grepT^{(x)}_{i'})} > 2 W^{(t_x)}
\end{align*}
where the last inequality comes from \eqref{eq:failW}. 

Because $\passes^{(x)}_{i^*} = 2\Psi - 1$ and $\passes^{(s)}_{i^*} = 2\Psi$, we must update $\repT_{i^*}$ using \eqref{eq:repTAssign} at some step $y$ between $x$ and $s.$
We know from the update rule that
\begin{align*}
    W^{(\grepT^{(y+1)}_{i^*})} \le W^{(t_x)} < \frac{1}{2}W^{(\grepT^{(x)}_{i^*})}
\end{align*}
and $\passes^{(y+1)}_{i^*} = 2\Psi.$

According to the strategy, all levels $i' < i^*$ are not affected between step $x$ and $y$, $i^*$ is still in $\cM^{(y+1)} \cap \cP^{(y+1)}.$
Using the inductive hypothesis \eqref{eq:IH} on step $y+1 \le s$ and level $i^*$, we have that
\begin{align*}
    \log_2 W^{(\grepT^{(y+1)}_{i^*})} < \Psi - \passes^{(y+1)}_{i^*} \le \Psi - 2\Psi = -\Psi,
\end{align*}
which leads to a contradiction.

Therefore, we know $i = \max \cM^{(x)}$ and $W^{(\repT^{(x)}_i)} > 2 W^{(t_x)}.$
At step $s$, the player choosing \gameActParam{shift}{$i$} updates $\grepT^{(s+1)}_i$ so that
\begin{align*}
    W^{(\grepT^{(s+1)}_{i})} \le W^{(t_x)} < \frac{1}{2} W^{(\grepT^{(x)}_i)} = \frac{1}{2}W^{(\grepT^{(s)}_i)}
\end{align*}
This concludes the case as well as the proof.
\end{proof}

%% file: spanner.tex
\section{Decremental Spanner and Embedding}
\label{sec:spanner}

The main result of this section is a new deterministic static algorithm to find a sparsifier $\tilde{J}$ of a graph $J$. As mentioned in the overview (\Cref{sec:tech_overview}), we also need to find a low-congestion embedding of $J$ into $\tilde{J}$. In fact, we need the following additional stronger property of $\tilde{J}$: given access to an embedding $\Pi_{J \to H'}$ from $J$ to $H'$, the graph $\tilde{J}$ also must satisfy that the composition $\Pi_{J \to H'} \circ \Pi_{J \to\tilde{J}}$ has almost the same bounds on vertex congestion and length as $\Pi_{J \to H'}$. It is worth pointing out that the composed embedding $\Pi_{J \to H'} \circ \Pi_{J \to\tilde{J}}$ is well-defined because $E(\tilde{J}) \subseteq E(J)$, as $\tilde{J}$ will be a spanner of $J$. \cite{chen2022maximum} showed that this holds if $\tilde{J}$ was a random sample of $J$ by applying concentration bounds. This section gives a deterministic method for achieving this.

\begin{restatable}{theorem}{staticEmbedding}\label{lma:staticEmbed}
Given unweighted, undirected graphs $H'$ and $J$ with $V(J) \subseteq V(H')$ and an embedding $\Pi_{J \to H'}$ from $J$ into $H'$. Then, there
is a deterministic algorithm $\textsc{Sparsify}(H', J, \Pi_{J \to H'})$ that returns a sparsifier $\tilde{J} \subseteq J$ with an embedding $\Pi_{J \to\tilde{J}}$ from $J$ to $\tilde{J}$ such that
\begin{enumerate}
    \item $\vcong(\Pi_{J \to H'} \circ \Pi_{J \to\tilde{J}}) \leq \gamma_{c} \cdot \length(\Pi_{J \to H'}) \cdot \left(\vcong(\Pi_{J \to H'}) + \Delta_{\max}(J)\right)$, and 
    \item $\length(\Pi_{J \to H'} \circ \Pi_{J \to\tilde{J}}) \leq \gamma_{\ell} \cdot \length(\Pi_{J \to H'})$, and
    \item $|E(\tilde{J})| = \tilde{O}(|V(J)| \gamma_{\ell})$.
\end{enumerate}
The algorithm runs in time $\tilde{O}(|E(J)|\gamma_{\ell}^2 \length(\Pi_{J \to H'}))$.
\end{restatable}

Given this result, one can obtain a deterministic dynamic algorithm to maintain a spanner of a graph $G$ undergoing edge deletions and vertex splits using the reduction presented in \cite[Sec.~5.1, arXiv]{chen2022maximum}. 
We point out that in \cite{chen2022maximum}, an analogous result to \Cref{lma:staticEmbed} was given with slightly better vertex congestion, length and runtime guarantees. However, the result in \cite{chen2022maximum} was randomized, while the goal of this section is to provide a determinstic algorithm.
Applying the reduction to \Cref{lma:staticEmbed}, one obtains the following result.

\begin{theorem}
\label{thm:spanner}
Given an $m$-edge $n$-vertex unweighted, undirected, dynamic graph $G$ undergoing update batches $U^{(1)}, U^{(2)}, \ldots$ consisting only of edge deletions and $\tilde{O}(n)$ vertex splits. There is a deterministic algorithm with parameter $1 \le L \le o\left(\frac{\log^{1/6} m}{\log\log m}\right)$, that maintains a spanner $H$ and an embedding $\Pi_{G \to H}$ such that for some $\gamma_{\ell}, \gamma_c = \exp(O(\log^{2/3} m \cdot \log\log m))$, we have
\begin{enumerate}
    \item \label{prop:sparsitySpanner} \underline{Sparsity and Low Recourse:} initially $H^{(0)}$ has sparsity $\tilde{O}(n \gamma_{\ell})$. At any stage $t \geq 1$, the algorithm outputs a batch of updates $U_H^{(t-1)}$ that when applied to $H^{(t-1)}$ produce $H^{(t)}$ such that $H^{(t)} \subseteq G^{(t)}$, $H^{(t)}$ consists of at most $\tilde{O}(n \gamma_{\ell})$ edges and 
    \[\sum_{t' \leq t} \Enc(U_H^{(t')}) = \tilde{O}\left(n\gamma_{\ell} + \sum_{t' \leq t} |U^{(t')}| \cdot n^{1/L}\gamma_{\ell}\right) \text{ , and}\] 
    \[\sum_{t' \leq t} |U_H^{(t')}| = \tilde{O}\left(\sum_{t' \leq t} |U^{(t')}| \cdot n^{1/L}\gamma_{\ell}\right) \text{ , and}\]
    \item \underline{Low Congestion, Short Paths Embedding:} $\length(\Pi_{G \to H}) \le (\gamma_{\ell})^{O(L)}$ and $\vcong(\Pi_{G \to H}) \le (\gamma_{\ell})^{O(L^2)} (\gamma_c)^{O(L)} \Delta_{\max}(G)$, and
    \item
    \underline{Low Recourse Re-Embedding:} \label{prop:lowRecourseSpanner}
    the algorithm further reports after each update batch $U^{(t)}$ at stage $t$ is processed, a (small) set $D^{(t)} \subseteq E(H^{(t)})$ of edges, such that for all other edges $e \in E(H^{(t)}) \setminus D^{(t)}$, there exists \textbf{no} edge $e' \in E(G^{(t)})$ whose embedding path $\Pi^{(t)}_{G \to H}(e')$ contains $e$ at the current stage but did not before the stage. The algorithm ensures that at any stage $t$, we have $\sum_{t' \leq t} |D^{(t')}| = \tilde{O}\left( \sum_{t' \leq t} |U^{(t)}| \cdot n^{1/L}(\gamma_c\gamma_{\ell})^{O(L^2)}\right)$, i.e. that the sets $D^{(t)}$ are roughly upper bounded by the size of $U^{(t)}$ on average. \label{item:lowReEmbed} 
\end{enumerate}
The algorithm takes initialization time $\tilde{O}(m \gamma_{\ell})$ and processing the $t$-th update batch $U^{(t)}$ takes amortized update time $\tilde{O}(\textsc{Enc}(U^{(t)}) \cdot n^{1/L}(\gamma_c\gamma_{\ell})^{O(L^2)}\Delta_{\max}(G))$.
\end{theorem}

Taking $L = \log^{1/9} m$ in \cref{thm:spanner} gives a parameter $\gamma_s = \exp(O(\log^{8/9}m\log\log m))$ such that the lengths of the embeddings, amortized recourse of the spanner, and amortized size of $D$ are all $O(\gamma_s)$. The vertex congestion and amortized runtime are bounded by $O(\gamma_s \Delta_{\max}(G))$. Ultimately, $\Delta_{\max}(G)$ will be chosen to be around $\exp(\log^{17/18} \log \log m) \gg \gamma_s$. We emphasize that the guarantees \ref{prop:sparsitySpanner} and \ref{item:lowReEmbed} of \Cref{thm:spanner} are with respect to the number of updates in each batch $U^{(t)}$ and \emph{not} with respect to the (possibly much larger) encoding size of $U^{(t)}$. This allows us to bound the recourse of our algorithm by a quantity that is $m^{o(1)}$, independent of the maximum degree.

The rest of this section is concerned with proving \Cref{lma:staticEmbed}.

\paragraph{Additional Tools.} At a high level, the proof of \Cref{lma:staticEmbed} follows by performing an expander decomposition, and producing a sparsifier on each expander by applying a data structure for finding short paths in decremental expanders. This differs from the previous approach of \cite{chen2022maximum} that first produced a sparsifier by uniformly sampling each edge of the expander, and then applying the decremental shortest path data structure to find an embedding. Here, we also use the data structure and embedding to deterministically find a sparsifier. To formalize this, we start by surveying some tools on expander graphs. Recall the definiton of expanders.

\begin{definition}[Expander]\label{def:expander}
Let $G$ be an unweighted, undirected graph and $\phi \in (0,1]$, then we say that $G$ is a $\phi$-expander if for all $\emptyset \neq S \subsetneq V$, $|E_G(S, V \setminus S)| \geq \phi \cdot \min\{\vol_G(S), \vol_G(V \setminus S)\}$.
\end{definition}

Sparse expander graphs can be constructed efficiently by a deterministic algorithm.

\begin{theorem}[see \cite{CGLNPS21}]\label{thm:constructDetExpander}
Given an integer $n > 1$, and a weight vector $\bw \in \mathbb{R}_{\geq 1}$, there is a deterministic algorithm $\textsc{ConstructExpander}(n, \bw)$, that constructs an (unweighted) $\phi_{Const}$-expander $W$ for $\phi_{Const} = \Theta(1)$ on $n$ vertices with $\bw \leq \bDeg_W \leq 18 \cdot \bw$. The runtime is $O(n + \|\bw\|_1)$.
\end{theorem}

We can further decompose any graph into a collection of expanders. The proof of this statement follows almost immediately from \cite{SW19, CGLNPS21} and is therefore deferred to \Cref{app:expanderStatement}.

\begin{restatable}[see Corollary 6.2 in arXiv v2 in \cite{CGLNPS21}]{theorem}{decompose}\label{thm:expanderStatement}
Given an unweighted, undirected graph $G$, there is an algorithm $\textsc{Decompose}(G, r)$ that computes an edge-disjoint partition of $G$ into graphs $G_0, G_1, \dots, G_{\ell}$ for $\ell = O(\log n)$ such that for each $0 \leq i \leq \ell$, $|E(G_{i})| \leq 2^{i} n$ and for each $0 < i \leq \ell$ and nontrivial connected component $X$ of $G_i$, $G_i[X]$ is a $\phi$-expander for $\phi = \tilde{\Omega}(1/\exp((\log m)^{2/3}))$, and each $x \in X$ has $\deg_{G_i}(x) \geq \phi 2^{i}$. The algorithm runs in time $\tilde{O}(m \cdot \exp\left({O(\log(m)^{2/3} \log\log(m))}\right))$.
\end{restatable}

Further, we use the following result from \cite{CS21} (see Theorem 3.9 in arXiv v1 in \cite{CS21}). Given a $\phi$-expander undergoing edge deletions the data structure below implicitly maintains a vertex subset that stays an expander using standard expander pruning techniques (see for example \cite{NSW17,SW19}). Further, on this vertex subset, it can output a path of length $m^{o(1)}$ between any pair of queried vertices. We note that we cannot use \cite{CS21} directly and instead replace its internal randomized expander computation by using Corollary 6.2 in arXiv v2 in \cite{CGLNPS21} to compute the expanders deterministically.

\begin{theorem}[see \cite{CS21, CGLNPS21}]\label{thm:apspDataStructure}
There is a deterministic data structure that when given an unweighted, undirected graph $G$ that is $\phi$-expander for $\phi = \tilde{\Omega}(1/\exp((\log m)^{2/3}))$.  
explicitly maintains a monotonically increasing vertex subset $\hat{V} \subseteq V(G)$ and handles the following operations:
\begin{itemize}
    \item $\textsc{Delete}(e)$: Deletes edge $e$ from $E(G)$ and then explicitly outputs a set of vertices that were added to $\hat{V}$ due to the edge deletion. 
    \item $\textsc{GetPath}(u,v)$: for any $u,v \in V(G) \setminus \hat{V}$ returns a simple path consisting of at most $\gamma_{{ExpPath}}$ edges
    between $u$ and $v$ in the graph $G[V(G) \setminus \hat{V}]$. Each path query can be implemented in time $\gamma_{{ExpPath}}$, where $\gamma_{{ExpPath}} = \exp\left({O(\log(m)^{2/3} \log\log(m))}\right)$ for some $\gamma_{ExpPath} \geq 1/\phi$. The operation does not change the set $\hat{V}$.
\end{itemize}
The data structure ensures that after $t$ deletions $\vol_{G^{(0)}}(\hat{V}) \leq \gamma_{{del}} t/\phi$ for some constant $\gamma_{{del}} = O(1)$. The total time to initialize the data structure and to process all deletions is $O(|E(G^{(0)})| \gamma_{{ExpPath}})$.
\end{theorem}

\begin{algorithm}[!hp]
\small
\SetKwRepeat{Do}{do}{while}
$\tau_1\defeq 2\gamma_{{ExpPath}} \length(\Pi_{J \to H'}) \gamma_{del} \left(\vcong(\Pi_{J \to H'}) + \Delta_{max}(J) \right)/\phi$; $\tau_{2} \defeq {\gamma_{{ExpPath}}\gamma_{{del}}}/\phi$.\\
\tcc{Task 1: Decompose $J$ into expanders and construct witness graph $W$}
$J_0, J_1, \ldots, J_{\ell} \gets \textsc{Decompose}(J, r)$.\;
$W \gets (V, \emptyset)$.\;
\ForEach(\label{lne:constructWitnessGraph}){$i \in [0,\ell]$ and connected component $X$ in $J_i$}{
    For $i > 0$, invoke $\textsc{ConstructExpander}(|X|, \bDeg_{J_i[X]} / (\phi 2^i))$ and identify the vertices of the expander $W_{J_i[X]}$ with the vertices of $X$ (arbitrarily); for $i = 0$, set $W_{J_i[X]}$ to $J_i[X]$. \\
    Add the edges of $W_{J_i[X]}$ to $W$ with weight $\phi2^i$.\label{lne:defCapWitness}\\
}

\tcc{Task 2: Embed Witness Graph $W$ into $J$ with low vertex congestion and small support.}
$\eta_1 \gets 0$; \label{lne:startSecondPhase} \lForEach{$e \in E(W)$}{$\Pi_{W \rightarrow J}(e) \gets \begin{cases} e & e \in E(J) \\ \emptyset & \text{otherwise}\end{cases}$}

\Do(\label{lne:doWhile1}){$\exists e \in E(W)$ such that $\Pi_{W \rightarrow J}(e) = \emptyset$}{
    $\eta_1 \gets \eta_1 + 1$.\\
    
    \ForEach{$i \in (0,\ell]$ and connected component $X$ in $J_i$}{
        $J^{APSP}_{J_i[X]} \gets $ a copy of $J_i[X]$.\\
        Initialize a data structure from \Cref{thm:apspDataStructure} denoted by $\mathcal{DS}_{J_i[X]}$ on graph $J^{APSP}_{J_i[X]}$ with parameter $\phi \defeq \phi$ maintaining set $\hat{V}_{J_i[X]}$.
    }

    \While(\label{lne:whileConditionEmbed}){$\exists e = (u,v) \in E(W)$ such that $\Pi_{W \rightarrow J}(e) = \emptyset$ \textbf{and} for some $i \in [0,\ell]$ and connected component $X$ in $J_i$, we have $u,v \in X \setminus \hat{V}_{J_i[X]}$}{
        $\Pi_{W \rightarrow J}(e) \gets 
 \mathcal{DS}_{J_i[X]}.\textsc{GetPath}(u,v)$.\label{lne:computeMCFlowGetPath1}\\

        \While(\label{lne:whileEmbed1}){$\exists x \in V(H'),$  $\vcong(\Pi_{J \to H'} \circ \Pi_{W \to J}, x) \geq \eta_1 \cdot \tau_1$}{
            \ForEach{edge $e'$ in $J^{APSP}_{J_i[X]}$ with $x \in \Pi_{J \to H'}(e')$}{
                    Remove edge $e'$ from $J^{APSP}_{J_i[X]}$ via $\mathcal{DS}_{J_i[X]}.\textsc{Delete}(e')$. }
            }
        }    
}
\tcc{Task 3: Embed $J$ into witness graph $W$ with low edge congestion.}
$\eta_2 \gets 0$; \label{lne:startThirdPhase}\lForEach{$e \in E(J)$}{$\Pi_{J \rightarrow W}(e) \gets \begin{cases} e & e \in E(W) \\ \emptyset & \text{otherwise}\end{cases}$.}

\Do(\label{lne:secondDoWhileStart}){$\exists e \in E(J)$ such that $\Pi_{J \rightarrow W}(e) = \emptyset$}{
    $\eta_2 \gets \eta_2 +1$.\\    
    \ForEach{$i \in (0,\ell]$ and connected component $X$ in $J_i$}{
        $W^{APSP}_{J_i[X]} \gets $ an unweighted copy of $W_{W_{J_i[X]}}$.\\
        Initialize a data structure from \Cref{thm:apspDataStructure} denoted by $\mathcal{DS}_{W_{J_i[X]}}$ on graph $W^{APSP}_{J_i[X]}$ with parameter $\phi_{Const}$ (see \Cref{thm:constructDetExpander}) maintaining set $\hat{V}_{W_{J_i[X]}}$.
    }

    \While(\label{lne:while2ConditionEmbed}){$\exists e = (u,v) \in E(J)$ such that $\Pi_{J \rightarrow W}(e) = \emptyset$ \textbf{and} for some $i \in [0,\ell]$ and connected component $X$ in $J_i$, we have $u,v \in X \setminus \hat{V}_{W_{J_i[X]}}$}{
        $\Pi_{J \rightarrow W}(e) \gets \mathcal{DS}_{W_{J_i[X]}}.\textsc{GetPath}(u,v)$.\label{lne:computeMCFlowGetPath2}\;

     \While(\label{lne:removeCongestedEdges}){$\exists e' \in E(W_{J_i[X]})$ for some $i$ and $X$ with $\econg(\Pi_{J \rightarrow W}, e') \geq \eta_2 \cdot \tau_{2}$}{
        Remove edge $e'$ from $W^{APSP}_{J_i[X]}$ via $\mathcal{DS}_{W_{J_i[X]}}.\textsc{Delete}(e')$. }
        
    }    
}
\Return $(\tilde{J} \gets \Pi_{W \to J}(W), \Pi_{J \to \tilde{J}} \gets \Pi_{W \to J} \circ \Pi_{J \to W})$
\caption{$\textsc{Sparsify}(H', J, \Pi_{J \to H'})$}
\label{alg:sparsify}
\end{algorithm}

\paragraph{The Algorithm.} We can now use these tools to give \Cref{alg:sparsify} that implements the procedure $\textsc{Sparsify}(H', J, \Pi_{J \to H'})$. Instead of attempting to directly find a graph $\tilde{J}$, the procedure in \Cref{alg:sparsify} divides the task into three subtasks:
\begin{itemize}
    \item \textbf{(Task 1)} The algorithm performs an expander decomposition on $J$, and applies \Cref{thm:constructDetExpander} to deterministically build a graph on each expander, which we refer to as a \emph{witness} (as in previous works). Thus the union of these graphs is a witness graph $W$ on the vertex set of $J$. $W$ will be a sparse graph with $\O(n)$ edges.
    \item \textbf{(Task 2)} Then, starting in \Cref{lne:startSecondPhase}, the algorithm finds the graph $\tilde{J}$ by finding an embedding $\Pi_{W \to J}$ of $W$ into $J$. We then later take $\tilde{J}$ to be the image of $\Pi_{W \to J}$ which is again sparse since $W$ is sparse and the embedding maps to paths of short length.
    \item \textbf{(Task 3)} Finally, starting in \Cref{lne:startThirdPhase} the algorithm finds an embedding $\Pi_{J \to W}$ from $J$ into $W$. This allows us to take the embedding from $J$ to $\tilde{J}$ as $\Pi_{W \to J} \circ \Pi_{J \to W}$ which maps each edge in $J$ to a path in the image of $\Pi_{W \to J}$, i.e. $\tilde{J}$.
\end{itemize}
Crucially, we require the embeddings $\Pi_{W \to J}$ and $\Pi_{J \to W}$ to be of low congestion. To achieve this goal, we embed using only short paths, found efficiently via shortest-path data structures.

\paragraph{Analysis.} We start by proving a simple invariant that states that the congestion of $\Pi_{J \to H'} \circ \Pi_{W \to J}$ increases slowly throughout the first do-while loop.

\begin{invar}\label{inv:steadyIncreaseCong}
$\vcong(\Pi_{J \to H'} \circ \Pi_{W \to J}) \leq (\eta_1 + \frac{1}{2}) \cdot \tau_1$ holds before and after each iteration of the do-while loop starting in \Cref{lne:doWhile1}.
\end{invar}
\begin{proof}
Initially, i.e., before the first do-while loop iteration, each edge mapped by $\Pi_{W \to J}$ is either already in $J$ and mapped to itself or is set to an empty path. Thus, each edge in $J$ is used at most once and therefore initially $\econg(\Pi_{W \to J}) \leq 1$. Using \Cref{fact:transitiveEmbeddingCong}, we thus have that initially $\vcong(\Pi_{J \to H'} \circ \Pi_{W \to J}) \leq \vcong(\Pi_{J \to H'}) \leq \frac{1}{2}\tau_1$.

Let us take the inductive step for $\eta_1$-th iteration. Since before the $\eta_1$-th iteration, the embeddings and $\eta_1$ are in the state of either the base case or in the same state as after the $(\eta_1-1)$-th iteration, we have that the claim holds before the iteration. Further, we observe that $\vcong(\Pi_{J \to H'} \circ \Pi_{W \to J}, x)$ increases during the iteration only in \Cref{lne:computeMCFlowGetPath2} when a new path is added to the embedding $\Pi_{W \to J}$. Since thereafter, vertices with congestion at least $\eta_1 \cdot \tau_1$ are removed from all data structures, thus can not appear on any paths in this iteration of the do-while loop, we have that the congestion can be exceeded by at most as much as a single path can contribute to the congestion. 

Let us finally bound the maximal contribution of embedding a single edge $e \in E(W)$ to the maximum vertex congestion. Let us first observe that the graphs $J_i$ with $2^i > \Delta_{max}(J)/\phi$ are empty since by  \Cref{thm:expanderStatement} every vertex $x$ in a non-trivial component $X$, we have $deg_{J_i}(x) \geq \phi 2^i$ with $\Delta_{max}(J_i) \leq \Delta_{max}(J_i)$ since $J_i \subseteq J$. Thus, the weight of any edge $e \in E(W)$ as defined in \Cref{lne:defCapWitness} is at most $\Delta_{max}(J)/\phi$. Finally, we use that when we embed $e$ in \Cref{lne:computeMCFlowGetPath1}, we do so by using a simple path from the data structure from \Cref{thm:apspDataStructure}. Thus, on each vertex, we add at most $\Delta_{max}(J)/\phi \leq \tau_1/2$ congestion.
\end{proof}

It thus only remains to upper bound the number of iterations that the first do-while loop takes to bound the total vertex congestion of $\Pi_{J \to H'} \circ \Pi_{W \to J}$. 

\begin{claim}
At the end of \Cref{alg:sparsify}, we have $\eta_1 = O(\log m)$. 
\end{claim}
\begin{proof}
Clearly, $\eta_1$ is exactly the number of iterations of the do-while loop starting in \Cref{lne:doWhile1}. Let us analyze one such iteration. 

Let us denote by $E^{embed}$ the set of edges embed during the current iteration, and let us denote by $E^{notEmbed}$ the set of edges that remain not embedded after the iteration (both are subsets of $E(W)$). Let us first observe that by the while-loop condition in \Cref{lne:whileEmbed1}, each edge $e \in E^{notEmbed}$ that lives in graph $W_{J_i[X]}$ for some $i$ and $X$ has at least one of its endpoints in the set $\hat{V}_{J_i[X]}$. Note further, that we have by construction that for each $i$ and connected component $X$ in $J_i$, and $\phi(r) 2^i \bDeg_{W_{J_i[X]}} \leq 18\cdot \bDeg_{J_i[X]}$ by \Cref{thm:constructDetExpander}. Recall that each edge $e$ in $W_{J_i[X]}$ receives weight $\phi 2^i$ in $W$. We can therefore deduce that 
\begin{equation}\label{eq:lowerBoundEmbed}
    \sum_{i, X} \vol_{J_i[X]}(\hat{V}_{J_i[X]}) \geq \frac{\vol_W(E^{notEmbed})}{18}.
\end{equation}

On the other hand, whenever we embed an edge $e \in E(W)$ with weight $2^i$ for some $i$ in the current do-while loop, the amount of new vertex congestion added to $\sum_{x \in V(H')} \vcong(\Pi_{J \to H'} \circ \Pi_{W \to J}, x)$ is at most $2^i \cdot \gamma_{{ExpPath}} \length(\Pi_{J \to H'})$ since the data structure from \Cref{thm:apspDataStructure} returns paths of length at most $\gamma_{{ExpPath}}$ and $\Pi_{J \to H'}$ maps each edge to a path of length at most $\length(\Pi_{J \to H'})$. But, every time a vertex $x \in V(H')$ is deleted in the while-loop starting in \Cref{lne:whileEmbed1} (or more precisely, the edges in $J$ that embed into $x$ are removed from all shortest path data structures), we must have by \Cref{inv:steadyIncreaseCong} that $\vcong(\Pi_{J \to H'} \circ \Pi_{W \to J}, x)$ has increased by at least $\frac{1}{2}\tau_1$ since the last do-while loop iteration. The number of all edges removed from $J$ due to the deletion of a vertex $x$ is meanwhile bound by $\vcong(\Pi_{J \to H'}, x) \leq \vcong(\Pi_{J \to H'})$. Let $E^{del}$ be the edges deleted from $J$ in the while-loop in \Cref{lne:whileEmbed1} during the current do-while loop iteration. Then, we can conclude that $\vol_W(E^{embed}) \frac{2\gamma_{{ExpPath}} \length(\Pi_{J \to H'}) \vcong(\Pi_{J \to H'})}{\tau_1} \geq \vol_J(E^{del})$.

Finally, we can conclude by the guarantees of \Cref{thm:apspDataStructure} that 
\[
\vol_W(E^{embed}) \frac{2 \gamma_{{ExpPath}} \length(\Pi_{J \to H'}) \gamma_{del} \vcong(\Pi_{J \to H'})}{\phi \tau_1} \geq \sum_{i, X} \vol_{J_i[X]}(\hat{V}_{J_i[X]}).
\]
Replacing $\tau_1$ by its value, we get $\vol_W(E^{embed}) \geq \sum_{i, X} \vol_{J_i[X]}(\hat{V}_{J_i[X]})$ and combined with (\ref{eq:lowerBoundEmbed}), we derive $\vol_W(E^{embed}) \geq 18 \cdot \vol_W(E^{notEmbed})$. This implies that in a single round, we embed at least a constant fraction of the volume of edges not embedded yet. Since the total volume of edges in $W$ is a small polynomial in $m$, we have that after $O(\log m)$ rounds, the algorithm has embedded all edges, as desired.
\end{proof}

The following corollary is immediate from our analysis.

\begin{corollary}\label{cor:boundVertexCong}
At the end of the algorithm, we have $\vcong(\Pi_{J \to H'} \circ \Pi_{W \to J}) = O(\tau_1 \cdot \log(m))$.
\end{corollary}

We use the same proof strategy to bound the congestion of $\Pi_{W \to J}$.

\begin{invar}\label{invar:secondInvariantSpanner}
$\econg(\Pi_{J \to W}) \leq (\eta_2 + \frac{1}{2}) \tau_2$ holds before and after each iteration of the do-while loop starting in \Cref{lne:secondDoWhileStart}.
\end{invar}
\begin{proof}
Again, before the first do-while loop iteration, the invariant is trivially true as the embedding maps edges either to themselves or to empty paths. Consider now the $\eta_2$-th iteration of the do-while loop. We have that the invariant holds before the do-while loop since it holds before the first iteration and clearly also after the $(\eta_2-1)$-th iteration for $\eta_2 > 1$ by the induction hypothesis.

In a do-while loop iteration, $\eta_2$ is increased by one. After this increase, we have $\econg(\Pi_{J \to W}) \leq \eta_2 \tau_2$. While $\eta_2$ remains fixed throughout the rest of the iteration, $\econg(\Pi_{J \to W})$ might increase. However, whenever we increase the congestion by a new embedding path, we check the edge congestion of every edge in $E(W)$ and remove it if it exceeds $\eta_2\tau_2$ (see the while-loop starting in \Cref{lne:removeCongestedEdges}).

Thus, the congestion can exceed $\eta_2\tau_2$ by at most the amount that a single embedding path can increase the congestion. But since each edge in $J$ has weight $1$ and the paths returned by the data structure from \Cref{thm:apspDataStructure} returns simple paths, the maximum increase can be at most $1 \leq \frac{1}{2}\tau_2$.
\end{proof}

\begin{claim}
At the end of \Cref{alg:sparsify}, we have $\eta_2 = O(\log m)$. 
\end{claim}
\begin{proof}
We have that $\eta_2$ being the number of iterations of the do-while loop starting in \Cref{lne:secondDoWhileStart}. Let us again analyze one such iteration.

Let us denote by $E^{embed}$ the set of edges embed during the current iteration, and let us denote by $E^{notEmbed}$ the set of edges that remain not embedded after the iteration (both are subsets of $E(J) \setminus E(J_0)$). We first observe that for each edge $e \in E^{notEmbed}$ that lives in $J_i[X]$ for some $i$ and $X$, we have that at least one of its endpoints is in $\hat{V}_{W_{J_i[X]}}$ (see \Cref{lne:while2ConditionEmbed}). Note that by construction of $W_{J_i[X]}$ using the procedure in \Cref{thm:constructDetExpander}, we have that $\bDeg_{J_i[X]} \leq \phi 2^i \bDeg_{W_{J_i[X]}}$. Recall that each edge in $W_{J_i[X]}$ receives weight $\phi 2^i$ in $W$. Therefore, we have $\vol_J(E^{notEmbed}) \leq \sum_{i, X} \vol_{W}(E(\hat{V}_{W_{J_i[X]}}))$.

On the other hand, whenever we embed an edge $e \in E(J)$, the amount of edge congestion added to $\sum_{e \in E(W)} \econg(\Pi_{J \to W}, e)$ is at most $\gamma_{ExpPath}$ since we add paths from the data structure from \Cref{thm:apspDataStructure}. But we remove an edge $e' \in W_{J_i[X]}$ only after its congestion $\econg(\Pi_{J \mapsto W}, e')$ has increased by at least $\frac{1}{2} \tau_2$ since the start of the do-while loop iteration by the while-loop condition in \Cref{lne:removeCongestedEdges} and \Cref{invar:secondInvariantSpanner}. Thus, we have by the guarantees of \Cref{thm:apspDataStructure} and the fact that edges in each $W_{J_i[X]}$ receive a uniform weight that 
$\sum_{i, X} \vol_{W}(E(\hat{V}_{W_{J_i[X]}})) \leq \frac{\gamma_{del} \gamma_{ExpPath} \vol_J(E^{embed})}{\phi \tau_2}$.

Combined with our previous insight and plugging in the value of $\tau_2$, this holds $\vol_J(E^{notEmbed}) \leq \vol_J(E^{embed})$. Thus, again, we have that in each round the volume of edges embedded is at least a constant fraction of the remaining edges to be embedded into $W$. The do-while loop thus terminates after $O(\log m)$ iterations.
\end{proof}

\begin{corollary}\label{cor:boundEdgeCong}
At the end of the algorithm, we have $\econg(\Pi_{J \to W}) = O(\tau_2 \cdot \log(m))$.
\end{corollary}

With the congestion of both embeddings bound tightly, we can prove \Cref{lma:staticEmbed}.

\begin{proof}[Proof of \Cref{lma:staticEmbed}.]
Let us first argue about correctness. To establish the bound on the vertex congestion, it remains to use \Cref{fact:transitiveEmbeddingCong}, \Cref{cor:boundVertexCong} and \Cref{cor:boundEdgeCong} and the values of $\tau_1, \tau_2$ to conclude that $\vcong(\Pi_{J \to H'} \circ \Pi_{J \to \tilde{J}}) \leq \gamma_{c} \cdot 
 \length(\Pi_{J \to H'}) \cdot \left(\vcong(\Pi_{J \to H'}) + \Delta_{max}(J) \right)$ for some 
 $$\gamma_{c} = O\left(\frac{\gamma^2_{{ExpPath}} \gamma^2_{del} \log^2(m)}{\phi^2}\right) = \exp\left({O(\log(m)^{2/3} \log\log(m))}\right).$$

 For to bound $\length(\Pi_{J \to H'} \circ \Pi_{J \to \tilde{J}})$, we recall that $\Pi_{J \to \tilde{J}} = \Pi_{W \to J} \circ \Pi_{J \to W}$ and both embeddings $\Pi_{W \to J}$ and $\Pi_{J \to W}$ find the embedding paths using the data structure from \Cref{thm:apspDataStructure} which produces paths of length at most $\gamma_{ExpPath}$. Thus $\length(\Pi_{J \to H'} \circ \Pi_{J \to \tilde{J}}) \leq \length(\Pi_{J \to H'}) \cdot  \length( \Pi_{W \to J}) \cdot \length( \Pi_{J \to W}) \leq \gamma_\ell \cdot \length(\Pi_{J \to H'})$ for $\gamma_{\ell} = \gamma_{ExpPath}^2$.

For the runtime analysis, we first see that we can bound the total runtime of all shortest paths data structures by $\tilde{O}(|E(J)| \cdot \gamma_{\ell})$. Further, we have to track the vertex congestion of all vertices in $H'$ when constructing the embedding $\Pi_{W \to J}$. For the sparsity analysis, we observe that $|E(W)| = \tilde{O}(|V(J)|/\phi)$ by the construction of $W$ in the foreach-loop starting in \Cref{lne:constructWitnessGraph}, and each edge is embedded via a path of length at most $\gamma_{\ell}$ and each such edge is then mapped to a path in $H'$ of length at most $\length(\Pi_{J \to H'})$, it is not hard to verify that we can track the vertex congestion in time $\tilde{O}(|V(J)|\gamma_{\ell} \length(\Pi_{J \to H'})/\phi) = \tilde{O}(|E(J)|\gamma_{\ell}^2 \length(\Pi_{J \to H'}))$. The rest of the algorithm can straightforwardly be implemented in time asymptotically bounded by either of the above-analyzed components.

Finally, for the sparsity analysis, we use again that $|E(W)| = \tilde{O}(|V(J)|/\phi)$ and the fact that $\tilde{J}$ is just the image of $\Pi_{W \to J}$ which ensures that $|E(\tilde{J})| = O(|V(J)| \length(\Pi_{W \to J}) / \phi)$ and we finally use that $ \length(\Pi_{W \to J}) \leq \gamma_{ExpPath}$ and $\gamma_{ExpPath} \geq 1/\phi$, holding $|E(\tilde{J})| = O(|V(J)| \gamma_{\ell})$.
\end{proof}

%% file: lsst.tex
\section{Deterministic Low-Stretch Spanning Tree}
\label{sec:lsst}

In this section, we prove \Cref{thm:lsst} regarding deterministic dynamic low-stretch spanning tree algorithms. Here, we give a proof of the result in \Cref{thm:lsst} with amortized update time, however, since the algorithm uses a standard batching technique, it can be de-amortized via standard techniques (see for example \cite{gutenberg2020fully}) at the cost of only a constant factor increase in runtime.

Similarly to our min-ratio cycle data structure, the data structure maintains a hierarchy of partial trees and uses a deterministic spanner algorithm. It is worth noting that the algorithm \emph{does not} need anything resembling the rebuilding game or shift-and-rebuild game \Cref{sec:rebuilding}. 

Let us formalize some notation for the data structure. We will let $E$ denote the edge set of the graph $G$. $E$ and $G$ may change over time. For simplicity, we assume that the number of edges in the graph is always in the range $[m, 2m]$ for some parameter $m$. Otherwise, when the number of edges halves or doubles, we restart the data structure. This only increases the amortized runtime by a constant. At time $t$, let $\bell^{(t)} \in \R^E$ be the lengths. We want to maintain a tree $T \subseteq E(G)$ such that $\sum_{e \in E} \str^{T,\bell}(e) \le m^{1+o(1)}$.

\subsection{Data Structure Description}
\label{subsec:lsstds}
The data structure is a hierarchy of partial trees and spanners with low-congestion routings. Let $k = m^{1/d}$. Let $G_0 = G$ be the input graph. For $i = 0, 1, \dots, d-1$ we will inductively define $G_{i+1}$ given $G_i$. $G_i$ will have around $m/k^i$ edges.

\paragraph{Partial tree.} Let $F_i \subseteq G_i$ be a partial tree maintained by \Cref{lemma:globalstretch} for weights $\bv = 1$, and parameter $k - m^{1/d}$ as defined. Let $\wstr_e$ be the stretch overestimates. Define the contracted graph $H_i = G_i/F_i$, with edge lengths $\bell^{H_i}_e := \wstr_e \bell^{G_i}_e$. Rebuild level $i, i+1, \dots, d-1$ whenever the total number of updates passed to $G_i$ exceeds $m/k^{i+1}$.

\paragraph{Spanner with embeddings.} We first partition the graph $H_i$ into $\O(1)$ subgraphs. For $j' \le \O(1)$ and $j \le \log_2 k$, define $H_{i,j,j'}$ to be the subgraph of $H_i$ consisting of edges $e$ satisfying $\wstr_e \in [2^j, 2^{j+1}]$ and $\bell^{H_i}_e \in [2^{j'}, 2^{j'+1}]$. Let $G_{i+1,j,j'}$ be a decremental spanner with embedding of $H_{i,j,j'}$, maintained using \Cref{thm:spanner}. Critically, we know that $|E(H_{i,j,j'})| \le \O(|E(G_i)|/2^j)$ so we will ensure that all edge congestions in the embedding $\Pi_{H_{i,j,j'} \to G_{i+1,j,j'}}$ are bounded by $m^{o(1)}k/2^j$. Finally, we define $G_{i+1} = \bigcup_{j,j' \le \O(1)} G_{i+1,j,j'}$. The edge lengths in $G_{i+1}$ are defined to be the same as those in $H_i$.

\paragraph{The low stretch tree.} Define $T_i = F_i \cup F_{i+1} \cup \dots \cup F_d$. Note that $T_i$ is a spanning tree of $G_i$. The low-stretch tree maintained by our data structure will be $T_0$.

\subsection{Algorithm Analysis}
We now prove \Cref{thm:lsst} by analyzing the amortized runtime and stretch.
\begin{proof}[Proof of \Cref{thm:lsst}]
We use the algorithm described in \Cref{subsec:lsstds}. It is clearly deterministic. To bound the runtime, recall that the recourse of $F_i$ from \Cref{lemma:globalstretch} is $\O(1)$, and the recourse from the spanner in \Cref{thm:spanner} is at most some $\gamma_r \le \exp(O(\log^{8/9} \log \log m))$, because there are $\O(1)$ buckets $(j, j')$. Thus the amortized runtime of the data structure is bounded by $k \cdot \O(\gamma_r)^d = m^{1/d} \O(\gamma_r)^d \le \exp(O(\log^{17/18} \log \log m))$ for the choice $d = \log^{1/18} m$.

We will bound $\sum_{e \in E(G_i)} \str^{T_i,\bell^{G_i}}(e)$ by induction.
We calculate
\begin{align*}
    \sum_{e \in E(G_i)} \str_e^{T_i,\bell^{G_i}}(e) &\overset{(i)}{\le} 4\sum_{e \in E(G_i)} \wstr_e \str^{T_{i+1}, \bell^{H_i}}(e) \\
    &\overset{(ii)}{\le} 4\sum_{j,j'} \sum_{e \in E(H_{i,j,j'})} \wstr_e \sum_{f \in \Pi_{H_{i,j,j'} \to G_{i+1,j,j'}}(e)} \frac{\str^{T_{i+1}, \bell^{G_{i+1}}}(f) \bell^{H_i}(f)}{\bell^{H_i}(e)} \\
    &\overset{(iii)}{\le} 16\sum_{j,j'} \sum_{e \in E(H_{i,j,j'})} \sum_{f \in \Pi_{H_{i,j,j'} \to G_{i+1,j,j'}}(e)} \wstr_f \str^{T_{i+1}, \bell^{G_{i+1}}}(f) \\
    &= 16\sum_{j,j'} \sum_{f \in G_{i+1,j,j'}} \econg(\Pi_{H_{i,j,j'} \to G_{i+1,j,j'}}, f)\wstr_f \str^{T_{i+1}, \bell^{G_{i+1}}}(f) \\
    &\overset{(iv)}{\le} O(k\gamma_c) \sum_{f \in E(G_{i+1})} \str^{T_{i+1}, \bell^{G_{i+1}}}(f),
\end{align*}
for some $\gamma_c \le \exp(O(\log^{8/9} m \log \log m))$. Step $(i)$ follows from the forest portal routing, $(ii)$ follows from the embedding, and $(iii)$ follows from the fact that $\wstr_e \le 2\wstr_f$ and $\bell^{H_i}(f) \le 2\bell^{H_i}(e)$ for $e, f \in H_{i,j,j'}$. To see $(iv)$, recall that $\sum_{e \in E(G_i)} \wstr_e \le \O(|E(G_i)|)$, so $|E(H_{i,j,j'})| \le \O(|E(G_i)|/2^j)$. Hence \Cref{thm:spanner} ensures that the \[ \econg(\Pi_{H_{i,j,j'} \to G_{i,j,j'}}, f) \le O(\gamma_c) \cdot \frac{|E(H_{i,j,j'})|}{|V(H_{i,j,j'})|} \le O(\gamma_c k /2^j) \le O(\gamma_c k / \wstr_f). \]

Hence we conclude that $\sum_{e \in E(G_i)} \str^{T_i,\bell^{G_i}}(e) \le O(\gamma_c)^{d-i} \cdot m/k^i$ for all $i = 0, 1, \dots, d-1$.
\end{proof}

%% file: jtree_appendix.tex
\section{Additional \texorpdfstring{$j$}{j}-tree Proofs}

In this section, we give a brief description of how to modify the proof of \cite[Lemma 6.5]{chen2022maximum} to show our \cref{lemma:globalstretch}. The only difference between \cite[Lemma 6.5]{chen2022maximum} and \cref{lemma:globalstretch} is that we require our data structure to handle vertex splits. Fortunately, the data structure of \cite[Lemma 6.5]{chen2022maximum} can be trivially modified to handle vertex splits at no extra runtime cost.

\subsection{Proof of \texorpdfstring{\Cref{lemma:globalstretch}}{lemma:globalstretch}}
\label{app:globalstretch}
In this section, we use the same notations as in \cite[Appendix B.3]{chen2022maximum}, arXiv version. We will use the notions of \emph{branch-free set} \cite[Definition B.3]{chen2022maximum}, the forest $F_T(R, \pi)$ given a tree $T$, a set of roots $R$, and a permutation on tree edges $\pi$ (\cite[Definition B.4]{chen2022maximum}).

We now describe how to modify the proof of \cite[Lemma 6.5]{chen2022maximum}, which appears at the end of \cite[Appendix B.3]{chen2022maximum}, to handle vertex splits. When a vertex $u$ is split into $u$ and $u^{\text{NEW}}$, we create a new isolated vertex for $u^{\text{NEW}}$ and add both $u^{\text{NEW}}$ and $u^{\uparrow T_H}$ to the set of roots $R$. This set is branch free by \cite[Lemma B.9]{chen2022maximum}.
We also update the forest $F \defeq F_T(R,\pi)$.
Because $R$ is incremental and $\pi$ is a total ordering, $F$ is decremental. \cite[Lemma B.9]{chen2022maximum} item 1 tells us that we only add $O(\log^2 n)$ vertices to $R$ per vertex split, so the operation can be implemented efficiently.

The remainder of the proof is identical to the proof of \cite[Lemma 6.5]{chen2022maximum} in \cite[Appendix B.3]{chen2022maximum}.

\subsection{Proof of \texorpdfstring{\Cref{lemma:strMWU}}{lemma:strMWU}}
\label{app:mwu}
The MWU in this section is very similar to previous standard arguments, and is based on \cite[Section B.4]{chen2022maximum}. The main change is to run the MWU for only $k$ steps so that we have $k$ trees, as opposed to $\O(k)$. This is not necessary in the argument, but simplifies some notation.
\begin{proof}
  Let $W = O(\log^4 n)$ be such that items \ref{item:stretchbound}, \ref{item:avgstretchbound} of \cref{lemma:globalstretch} imply
  \begin{align*}
    \sum_{e \in E} \bv_e \wstr_e &\le W \norm{\bv}_1, \text{ and} \\
    \max_{e \in E} \wstr_e &\le k W\log^2 n.
  \end{align*}
  Let $\rho = 10 k W\log^2 n = \O(k)$.
  The algorithm sequentially constructs edge weights $\bv_1, \ldots, \bv_k$ using a multiplicative weight update algorithm, and finds low-stretch trees $T_1, \dots, T_k$, forests $F_1, \dots, F_k$, and stretch overestimates $\wstr^1, \dots, \wstr^k$ with respect to these weights via \cref{lemma:globalstretch}.

  Initially, $\bv_1 \defeq \mathbf{1}$ is an the all 1's vector.
  After computing $T_i$, $\bv_{i+1}$ is defined as
  \begin{align*}
    \bv_{i+1, e} \defeq \bv_{i, e} \exp\left(\frac{\wstr^{i}_e}{\rho}\right) = \exp\left(\frac{1}{\rho} \sum_{j=1}^i \wstr^{j}_e\right) \forall e \in E.
  \end{align*}
  Finally we define the distribution $\blambda$ to be uniform over the set $\{1, \dots, k\}.$

  To show the desired bound \eqref{eq:strMWU}, we first relate it with $\norm{\bv_{k+1}}$ using the following:
  \begin{align*}
    \max_{e \in E} \frac{1}{k} \sum_{i=1}^k \wstr^{i}_e = \frac{\rho}{k} \max_{e \in E} \frac{1}{\rho} \sum_{i=1}^k \wstr^{i}_e\le \frac{\rho}{k}\log \left(\sum_e \exp\left(\frac{1}{\rho} \sum_{i=1}^k \wstr^{i}_e\right)\right) = \frac{\rho}{k}\log \norm{\bv_{k+1}}_1,
  \end{align*}
  where $\bv_{t+1}$ is defined similarly even though it is never used in the algorithm.

  Next, we upper bounds $\norm{\bv_{i}}_1$ inductively for every $i=1, \dots, k+1$.
  Initially, $\bv_{1} = \mathbf{1}$ and we have $\norm{\bv_{1}}_1 = m.$
  To bound $\norm{\bv_{i+1}}$, we plug in the definition and have the following:
  \begin{align*}
    \norm{\bv_{i+1}}_1
    &= \sum_e \bv_{i, e} \exp\left(\frac{\wstr^{i}_e}{\rho}\right) \le \sum_e \bv_{i, e} \left(1 + 2 \cdot \frac{\wstr^{i}_e}{\rho}\right) \\
    &= \norm{\bv_{i}}_1 + \frac{2}{\rho} \sum_{e} \bv_{i, e} \wstr^{i}_e \le \norm{\bv_{i}}_1 + \frac{2}{\rho} W \norm{\bv_{i}}_1 = \left(1 + \frac{2W}{\rho}\right) \norm{\bv_{i}}_1,
  \end{align*}
  where the first inequality comes from the bound $\wstr^{i}_e \le k W \log^2 n = 0.1 \rho$ and $e^x \le 1 + 2x$ for $0 \le x \le 0.1.$
  Applying the inequality iteratively yields
  \begin{align*}
  \max_{e \in E} \frac{1}{k} \sum_{i=1}^k \wstr^{i}_e \le \frac{\rho}{k}\log \norm{\bv_{k+1}}_1 \le \frac{\rho}{k} \log \left(1 + \frac{2W}{\rho}\right)^k \norm{\bv_{1}}_1 \le 2 W + \frac{\rho \log m}{k} = O(W \log^3 n)
  \end{align*}
  The desired bound \eqref{eq:strMWU} now follows by taking the logarithm of both sides.
\end{proof}

\subsection{Cycle Maintenance in a Tree Chain}
\label{app:fundChainCycle}

In this section, we recall the relevant pieces of \cite[Section 7.2]{chen2022maximum} in order to show how to find a cycle $\bDelta$ from a tree chain that has sufficient ratio, thus showing \cref{lemma:hintedTreeChain}. We start with some preliminary definitions.
Consider a tree chain $\cG = \{G_0, G_1, \dots, G_d\}$ , with a corresponding tree $T \defeq T^{\cG}$ as defined in \cref{def:TreeChainShift}.
For $\bg,\bell$ and a valid pair $\bc,\bw$ (\cref{def:validpair}), define $\bc^{G_0} \defeq \bc$ and $\bw^{G_0} \defeq \bw$, and $\bc^{G_i}$ and $\bw^{G_i}$ recursively for $1 \le i \le d$ via \cref{def:passcore,def:passsparsecore}. Let $\bell^{G_i},\bg^{G_i}$ be the lengths and gradients on the graphs $G_i$, and $\bell^{\mathcal{C}(G_i,F_i)},\bg^{\mathcal{C}(G_i,F_i)}$ be the lengths and gradients on the core graphs.

Note that every edge $e^G \in E(G) \setminus E(T)$ has a ``lowest'' level that the image of it (which we call $e$) exists in a tree chain, after which it is not in the next sparsified core graph. In this case, the edge plus its path embedding induce a cycle, which we call the \emph{sparsifier cycle} associated to $e$. In the below definition, we assume that the path embedding of a self-loop $e$ in $\mathcal{C}(G_i,F_i)$ is empty.
\begin{definition}
\label{def:sparsifiercycle}
Consider a tree-chain $G_0 = G,\dots,G_d$ (\cref{def:TreeChainShift}) with corresponding tree $T \defeq T^{G_0,\dots,G_d}$ where for every $0 \le i \le d$, we have a core graph $\mathcal{C}(G_i, F_i)$ and sparsified core graph $\SS(G_i, F_i) \subseteq \mathcal{C}(G_i, F_i)$, with embedding $\Pi_{\mathcal{C}(G_i, F_i)\to\SS(G_i, F_i)}$.

We say an edge $e^G \in E(G)$ is at level $\lvl_{e^G} = i$ if its image $e$ is in $E(\mathcal{C}(G_i, F_i)) \backslash E(\SS(G_i, F_i))$. Define the \emph{sparsifier cycle} $a(e)$ of such an edge $e = e_0 \in \mathcal{C}(G_i, F_i)$ to be the cycle $a(e) = e_0  \oplus \rev(\Pi_{\mathcal{C}(G_i, F_i)\to\SS(G_i, F_i)}(e_0)) = e_0  \oplus e_1 \oplus \cdots \oplus e_L$. We define the preimage of this sparsifier cycle in $G$ to be the \emph{fundamental chain cycle}
     \[ 
     a^G(e^G) = e_0^G \oplus T[v^G_0, u^G_1] \oplus e_1^G \oplus T[v^G_1, u^G_2] \oplus \dots \oplus e_L^G \oplus T[v^G_L, u^G_{L+1}],
     \]
    where $e^G_i = (u_i^G, v_i^G)$ is the preimage of edge $e_i$ in $G$ for each $i \in [L]$ and where we define $u^G_{L+1} = u^G_0$.
\end{definition}
We let $\ba(e)$ and $\ba^G(e^G)$ be the associated flow vectors for the sparsifier cycle $a(e)$ and fundamental chain cycle $a^G(e^G)$.
The following result shows that there is some fundamental chain cycle with good ratio.
\begin{lemma}[{\cite[Lemma 7.17]{chen2022maximum}}]
\label{lemma:goodenough}
Let $\bc, \bw$ be a valid pair. Let $T = T^{\cG}$ for a tree-chain $\cG=\{G_0,\dots,G_d\}$. Then
\[ \max_{e^G \in E(G) \setminus E(T)} \frac{|\l \bg, \ba^G(e^G) \r|}{\l \bell, |\ba^G(e^G)| \r} \ge \frac{1}{\O(k)} \frac{|\l \bg, \bc \r|}{\sum_{i=0}^d \|\bw^{G_i}\|_1}. \]
\end{lemma}
While \cite[Lemma 7.17]{chen2022maximum} has a $\wlen_{e^G}$ term in the denominator on the LHS, \cite[Lemma 7.14]{chen2022maximum} shows that $\wlen_{e^G} \ge \l \bell, |\ba^G(e^G)| \r$.

%% file: spanner_appendix.tex
\section{Proof of Expander Statement}
\label{app:expanderStatement}

The goal of this section is to show \Cref{thm:expanderStatement}.
The proof is through a standard reduction to deterministic expander decomposition {\cite{SW19, CGLNPS21}}.

\begin{theorem}[see {\cite{SW19, CGLNPS21}}] \label{thm:getExpander}
Given an unweighted, undirected $m$-edge graph $G$, there is an algorithm that finds a partition of $V(G)$ into sets $V_1, V_2, \ldots, V_k$ such that for each $1 \leq j \leq k$, $G[V_j]$ is a $\phi$-expander for $\phi = \tilde{\Omega}(1/\exp((\log m)^{1/3}))$, and there are at most $m/4$ edges that are not contained in any one of the expander graphs. The algorithm runs in time $\tilde{O}(m \cdot \exp((\log m)^{2/3}\log\log(m)))$.
\end{theorem}

\noindent We run \Cref{alg:decompose} (given below) to obtain the graphs $G_i$ as desired in \cref{thm:expanderStatement}.

\begin{algorithm}[!ht]
$\ell \gets \lceil\log_2 \Delta_{max}(G) \rceil + 1; G_{\ell} = G$\label{lne:decompFirstLine}\;
\For{$i = \ell, \ell - 1, \ldots, 1$}{
    Let $G_i^{\rcirclearrow}$ denote the graph $G_i$ with $2^i$ self-loops added to each vertex.\label{lne:GiLoopy}\;
    Compute an expander decomposition $V_0, V_1, \ldots, V_k$ of $G_i^{\rcirclearrow}$ using \Cref{thm:getExpander}.\;
    $G_{i-1} \gets \left(\bigcup_{0 \le j \le k} E_{G_i}(V_j, V \setminus V_j)\right)$.\label{lne:initGi}\;
    $G_{i} \gets G_i \setminus G_{i-1}$.

}
\caption{$\textsc{Decompose}(G)$}
\label{alg:decompose}
\end{algorithm}

\begin{claim}\label{clm:initializdGiIsGood}
At initialization, for each $i$ the graph $G_i$ has in \Cref{lne:decompFirstLine} or \Cref{lne:initGi} has at most $2^i n$ edges.
\end{claim}
\begin{proof}
We proceed by induction on $i$. For the base case, $i = \ell$, observe that $2^{\ell} \geq \Delta_{max}(G)$ and since $G_{\ell}$ is a subgraph of $G$, we have $|E(G_{\ell})| \leq 2^{\ell} n$.

For $i \mapsto i-1$, we observe that $G_i$ is unchanged since its initialization until at least after $G_{i-1}$ was defined in \Cref{lne:initGi}. Thus, using the induction hypothesis and the fact above, we conclude that $G_i^{\rcirclearrow}$ (defined in \Cref{lne:GiLoopy}) consists of at most $2^i n$ edges from $G_i$ plus $2^i n$ edges from all self-loops. Thus \Cref{thm:getExpander} implies that $|\bigcup_{0 \leq j \leq k} E_{G_i}(V_j, V \setminus V_j)| = |\bigcup_{0 \leq j \leq k} E_{G_i^{\rcirclearrow}}(V_j, V \setminus V_j)| \leq 2^{i+1}n/4 = 2^{i-1}n$, and since this is exactly the edge set of $G_{i-1}$, the claim follows.
\end{proof}

\begin{proof}[Proof of  \Cref{thm:expanderStatement}]
Using \Cref{clm:initializdGiIsGood} and the insight that each graph $G_i$, after initialization, can only have edges deleted from it, we conclude that $|E(G_i)| \leq 2^i n$ for each $i$. 

For the minimum degree property of each $G_i$ with $i > 0$, we observe by \Cref{thm:getExpander}, that for $G_i^{\rcirclearrow}$ and vertex $v$ in expander $V_j$, $\deg_{G_i}(v) = |E_{G_i}(v , V_j \setminus \{v\})| = |E_{G_i^{\rcirclearrow}}(v , V_j \setminus \{v\})| \geq \phi \cdot 2^{i}$.

The guarantee that the connected components in each graph $G_1, G_2, \ldots, G_{\ell}$ (but not necessarily $G_0$) are $\phi$-expanders stems from \Cref{thm:getExpander}.
\end{proof}